\documentclass[10pt]{amsart}

\usepackage{amsmath, latexsym, amsfonts, amssymb, amsthm,MnSymbol,stmaryrd}
\usepackage{graphicx, xcolor,hyperref,dsfont,epsfig,caption,wrapfig,pifont,tikz,subcaption}
\usepackage{datetime}
\usepackage{movie15}
\hypersetup{
    linktoc=page,
    linkcolor=red,          
    citecolor=blue,        
    filecolor=blue,      
    urlcolor=cyan,
    colorlinks=true           
}

\newcommand{\nocontentsline}[3]{}
\newcommand{\tocless}[2]{\bgroup\let\addcontentsline=\nocontentsline#1{#2}\egroup}
\numberwithin{equation}{section}

\newtheorem{theorem}{Theorem}[section]
\newtheorem{lemma}[theorem]{Lemma}
\newtheorem{proposition}[theorem]{Proposition}

\newtheorem{remark}[theorem]{Remark}
\newtheorem{definition}[theorem]{Definition}

\theoremstyle{definition}

\renewcommand{\tilde}{\widetilde}          
\DeclareMathSymbol{\leqslant}{\mathalpha}{AMSa}{"36} 
\DeclareMathSymbol{\geqslant}{\mathalpha}{AMSa}{"3E} 
\DeclareMathSymbol{\eset}{\mathalpha}{AMSb}{"3F}     
\renewcommand{\leq}{\;\leqslant\;}                   
\renewcommand{\geq}{\;\geqslant\;}                   
\newcommand{\dd}{\text{\rm d}}             

\newcommand{\C}{\mathbb{C}}
\renewcommand{\H}{\mathbb{H}}
\newcommand{\D}{\mathbb{D}}
\newcommand{\R}{\mathbb{R}}
\newcommand{\Z}{\mathbb{Z}}
\newcommand{\N}{\mathbb{N}}

\newcommand{\A}{\mathbb{A}}

\newcommand{\E}{\mathds{E}}
\renewcommand{\P}{\mathds{P}}

\newcommand{\cjd}{\rangle}
\newcommand{\cjg}{\langle}

\newcommand{\hf}{\frac{_1}{^2}}

\newcommand{\pl}{\partial}
\newcommand{\bbar}{\overline}
\newcommand{\mc}{\mathcal}
\newcommand{\la}{\lambda}

\usepackage{xparse}

\def\eps{\epsilon}

\def\T{\mathbb{T}}

\def\bi{\begin{itemize}}
\def\ei{\end{itemize}}

\def\bnum{\begin{enumerate}}
\def\enum{\end{enumerate}}
\def\<#1{\langle #1 \rangle}

\newcommand{\caA}{{\mathcal A}}
\newcommand{\caB}{{\mathcal B}}
\newcommand{\caC}{{\mathcal C}}

\newcommand{\caN}{{\mathcal N}}

\newcommand{\caP}{{\mathcal P}}

\author{Colin Guillarmou}
\address{Universit\'e Paris-Saclay, CNRS,  Laboratoire de math\'ematiques d'Orsay, 91405, Orsay, France.}
\email{colin.guillarmou@universite-paris-saclay.fr}

\author{R\'emi Rhodes}
\address{Aix-Marseille Universit\'e, CNRS  (UMR 7373), Institut de Math\'ematiques de Marseille (I2M), and Institut Universitaire de France (IUF)}
\email{remi.rhodes@univ-amu.fr}

\author{Baojun Wu}
\address{Peking University, Beijing International Center for Mathematical Research (BICMR)}
\email{ wubaojunmathe@outlook.com }
\title{Conformal Bootstrap for surfaces with boundary in Liouville CFT. Part 1: Segal axioms.}    

\date{}
\begin{document}

\begin{abstract} 
This paper is the first part of the proof of the conformal bootstrap for Liouville conformal field theory on surfaces with a boundary, devoted to Segal's axioms in this context.  We introduce the notion of Segal's amplitudes on surfaces with corners and prove the gluing property for such amplitudes. The semi-group of half-annuli and its generator are studied and we develop the necessary material for proving its spectral decomposition using scattering theory in the companion paper \cite{GRW2}. The Segal gluing properties and the spectral decomposition allows us to prove the conformal bootstrap formula for correlation functions of Liouville conformal field theory with a boundary. This  has several important applications to the study of conformal blocks (analyticity and convergence) in \cite{remypreprint}, in the construction of a unitary representation of mapping class group in the space of conformal blocks,  and the study of random moduli \cite{ARSmoduliRPM} in Liouville quantum gravity.
  \end{abstract}

\maketitle
\tableofcontents

\section{Introduction}


The study of Conformal Field Theory (CFT) defined on Riemann surfaces with a boundary, dubbed boundary CFT, started 40 years ago in physics and has found important applications ranging from condensed matter physics to particle physics, from cosmology to string theory. Major contributions were developed, for instance,  by John Cardy \cite{CARDY1984514,MR1048596}. The conformal bootstrap, pioneered by \cite{BPZ84}, is a general philosophy to solve CFT in physics, which consists in encoding the symmetries of the system and imposing some consistency conditions in order to derive  constraints that lead, eventually, to an exact expression for the correlation functions of the CFT in terms of representation theoretic special functions. In the case of boundary CFT, the conformal bootstrap led to specific results, which turned out to have many important applications as we will (at least partly) see. The purpose of this manuscript is to initiate  the mathematical study of the conformal bootstrap of a specific boundary CFT: the boundary Liouville CFT, which we will describe soon.

In mathematics, the paper \cite{BPZ84} has been as influential as challenging. Various methods have been developed to axiomatize CFT in mathematical terms. Borcherds \cite{Borcherds} and Frenkel-Lepowsky-Meurman \cite{Frenkel:1988xz}  introduced the concept of Vertex Operator Algebras, based on representation theory and algebraic geometry, providing a formal framework  in line with the conformal bootstrap philosophy in physics. However, this approach is mainly limited so far to rational CFTs, such as minimal models.  
Friedan and Shenker \cite{FriedanShenker87} formulated 2D CFT using analytic geometry on the moduli space of Riemann surfaces, treating correlation and partition functions as analytic functions on this space. They introduced a holomorphic vector bundle and a projectively flat Hermitian connection, viewing the partition function as the squared norm of a holomorphic section of this bundle. This method has inspired further work in algebraic geometry.
In   \cite{Segal87}, Graeme Segal  proposed a set of  axioms  to capture the conformal bootstrap approach to CFTs using a geometrical perspective, inspired by heuristics based on the path integral approach to CFT (for a path integral oriented introduction to mathematicians see \cite{Gawedzki96_CFT}, or the lecture notes \cite{LectureAndre}). A related description has been developed by Moore and Seiberg \cite{Moore-Seiberg} in the case of Rational CFTs.
 In the case of closed Riemann surfaces, Segal's axioms basically assume that one can define the CFT on Riemann surfaces with analytic boundary, namely a closed Riemann surface with analytic disks removed (hence the resulting Riemann surface has a boundary made up of several analytic circles), with some Dirichlet type boundary conditions imposed for the CFT  along the boundary circles. Such objects are called amplitudes and the boundary conditions are assumed to live in some prescribed Hilbert space. As such, they can be identified with operators acting on and taking values in tensor products of the Hilbert space, each copy of the Hilbert space being attached to one boundary circle. The main axiom is then that the amplitudes, seen as operators, compose in a natural way   under gluing of the underlying Riemann surfaces along the boundary circles\footnote{The gluing should actually be compatible with an orientation prescribed on the boundary circles but we will sweep  this point under the rug in this introduction.}. This can be rephrased in terms of functorial correspondence between the categories of 2-cobordisms and trace class operators acting on tensor products of the Hilbert space, hence the name of Segal functor. Then Segal considered the semigroup of amplitudes obtained by gluing annuli. The generator of this semigroup is called the Hamiltonian ${\bf H}$ of the CFT, acting on the Hilbert space. The conformal bootstrap can then be understood as an iterated Plancherel formula associated to a basis of eigenfunctions of this Hamiltonian: 
\begin{itemize}
\item Decomposing the surface into pairs of pants, the correlation/partition functions can be expressed as pairings of amplitudes of such pairs of pants, and using the spectral decomposition of ${\bf H}$, they can be written as multiple integrals over the spectrum of  ${\bf H}$ of matrix coefficients of pants amplitudes on the eigenbasis. 
\item Using the conformal symmetries (via Ward identities), these matrix coefficients can be expressed purely in terms of the $3$-point correlation functions of the CFT  on the Riemann sphere and algebraic coefficients involving the central charge and the eigenvalues of ${\bf H}$.  The $3$-point correlation functions depend on the considered model of CFT but the algebraic coefficient do not.
 \end{itemize}
In the final expression for correlation/partition functions, the structure constants are "sewed" together by means of the modulus square of a universal holomorphic functions on the Teichm\"uller space of a given  closed surface with marked points called conformal block, which encapsulates the conformal structure of the surface. Segal's approach of CFT is in some aspects the most intuitive way to understand the conformal bootstrap starting from the statistical physics description of a given CFT. The Segal functor   is crucial in this picture (see \cite{guillarmou2024} for further details). However,  examples of CFT for which Segal's axioms were established are rare (mainly for free theories or variants \cite{posthuma,Tener2017} and Liouville CFT \cite{GKRV21_Segal,CILT}).   The definition of this functor extends to the case of boundary CFT\footnote{Private communication with Andr\'e Henriques.} but no example of CFT obeying this definition is known so far.

\subsection{Liouville CFT}
Liouville CFT was originally introduced by Polyakov \cite{Polyakov81}   as a model for random Riemannian metrics in 2 dimensions arising in string theory and the paper \cite{BPZ84} was primarily aimed at solving the Liouville CFT. At the physics level, on a closed Riemann surface $\Sigma$ equipped with a metric $g$,  it corresponds to the path integral  
\begin{equation}\label{thePI}
F\mapsto \int F(\phi) e^{-S_{{\rm L}}(\phi ,g  )} \,D\phi 
\end{equation}
where $D\phi$ is the formal Lebesgue measure on some space of maps $\phi:\Sigma\to\R$,  $F$ is a  test functions on this functional space and the Liouville functional $S_{{\rm L}} $ is given by
 \begin{equation}\label{AL}
S_{{\rm L}}(\phi ,g  )=       \frac{1}{4\pi} \int_{\Sigma}\big(|d\phi |^2_{{g }}+QK_{{g }} \phi  +4\pi \mu e^{ \gamma \phi  }\big)\, {\rm dv}_{{g }}  .
\end{equation}
 Here $|\cdot|_g$ is the induced metric on $T^\ast\Sigma$,  ${\rm v}_g$ the Riemannian volume measure,
  $K_{{g }}$  is  the scalar curvature of the metric $g$, and the parameters are $\gamma\in (0,2)$, $\mu>0$ and $Q=\frac{\gamma}{2}+\frac{2}{\gamma}$.
The local fields for Liouville CFT  are formally given by $V_\alpha(z,\phi)=e^{\alpha\phi(z)}$ with $\alpha\in\C$ so that the correlation functions are then formally defined by
 \begin{equation}\label{correl}
\langle\prod_{i=1}^mV_{\alpha_i}(z_i)\rangle=\int_{\Sigma}  \Big(\prod_{i=1}^ne^{\alpha_i\phi(z_i)}\Big)e^{- S_{{\rm L}}(\phi ,g  )}D\phi 
\end{equation}
 where $z_1,\dots,z_m$ are distinct points on $\Sigma$ and $\alpha_1,\dots,\alpha_m\in\C$ are complex numbers. The probabilistic construction of this path integral was carried out in \cite{DKRV16,Guillarmou2019} using Gaussian Multiplicative Chaos theory (GMC) \cite{Kahane85,rhodes2014_gmcReview} (this is recalled in Section \ref{sec:backclosed}), the structure constants were computed in \cite{KRV19_local,KRV_DOZZ}, and shown to coincide with the DOZZ formula proposed in physics \cite{DornOtto94,Zamolodchikov96}, and     the proof of  the conformal bootstrap for the sphere was then obtained   from the harmonic analysis of the Hamiltonian in \cite{GKRV20_bootstrap}, pioneered in physics by \cite{Teschner_revisited}, and the construction of the Segal functor and the bootstrap on all surfaces in \cite{GKRV21_Segal}.   
 
 \subsection{Liouville CFT with (Neumann) boundary}
 In this paper we initiate the mathematical study of the conformal bootstrap for boundary Liouville CFT.  Boundary Liouville CFT is defined on Riemann surfaces with boundary. The boundary of such surfaces is a collection of circles. On such a Riemann surface $\Sigma$ (with boundary denoted by $\partial \Sigma$) equipped with a Riemannian metric $g$ compatible with the complex structure, Liouville CFT formally corresponds to the path integral \eqref{thePI} where the Liouville action is now
  \begin{equation}\label{ALboundary}
S_{{\rm L}}(\phi ,g  )=       \frac{1}{4\pi} \int_{\Sigma}\big(|d\phi |^2_{{g }}+QK_{{g }} \phi  +4\pi \mu e^{ \gamma \phi  }\big)\, {\rm dv}_{{g }} +  \frac{1}{2\pi} \int_{\partial\Sigma}\big(Qk_{{g }} \phi  +2\pi \mu_B e^{\frac{ \gamma}{2} \phi  }\big)\, \dd \ell_g 
\end{equation}
where $k_g$ is the geodesic curvature, $\dd \ell_g$ is the line element on $\partial\Sigma$, and $\mu_B$ is a nonnegative piecewise constant function on $\partial\Sigma$. The local fields are now of two types: the bulk insertions are $V_\alpha(z,\phi)=e^{\alpha\phi(z)}$, with $\alpha\in\C$ and $z\in \mathring{\Sigma}$,  and the boundary insertions are $V_\beta(s,\phi)=e^{\frac{\beta}{2}\phi(s)}$ with $\beta\in\C$ and $s\in \partial{\Sigma}$. The correlation functions involve both bulk and boundary insertions, and represent the formal path integral
  \begin{equation}\label{correlboundary}
\langle\prod_{i=1}^mV_{\alpha_i}(z_i)\prod_{j=1}^{m_B}V_{\beta_j}(s_j)\rangle=\int_{\Sigma}  \Big(\prod_{i=1}^ne^{\alpha_i\phi(z_i)})\prod_{j=1}^{m_B}e^{\frac{\beta_j}{2}\phi(s_j)})\Big)e^{- S_{{\rm L}}(\phi ,g  )}D\phi 
\end{equation}
 where $z_1,\dots,z_m$ are distinct points in $\mathring{\Sigma}$, $s_1,\dots, s_{m_B}$ are distinct points on $\partial\Sigma$, and $\alpha_1,\dots,\alpha_m,\beta_1,\dots,\beta_{m_B}$ are complex numbers.  The rigorous path integral construction using probabilistic methods was initiated for boundary Liouville CFT  in \cite{huang2018,MR3843631} and then extended to all Riemann surfaces in \cite{Wu1}. The construction, recalled in Section \ref{sec:open}, involves the Gaussian Free Field (GFF) with Neumann condition on $\pl \Sigma$.
 
 The CFT structure of boundary Liouville CFT is richer than the case of closed surfaces. There are several possible structure constants, which corresponds to correlation functions on the unit disk with  1 bulk insertion, 1 bulk and 1 boundary insertion, or 2 or 3 boundary insertions. Their expressions was conjectured in physics (see \cite{Nakayama} for a review) and were mathematically proved in \cite{Remy20,MR4483018} when the bulk cosmological constant is null $\mu=0$, and then in the general case in \cite{ang2022fzz,Ang_Remy_Sun_Zhu}.
 
\subsection{Segal functor}
The first goal of this paper is to introduce the necessary mathematical formalism and construct the Segal functor for boundary Liouville CFT.  Segal functor can be interpreted as the operation of decomposing the path integral on a surface into a composition of operators, called \emph{amplitudes}, associated to a geometric decomposition of the surface into smaller pieces (namely surfaces with boundary or with corners) along cuts. These operators are typically bounded, or even Hilbert-Schmidt, on some Hilbert spaces where one associates a Hilbert space  to each cut.\\ 

\textbf{Surface decomposition.}
A surface $\Sigma$ with boundary can be cut into pieces along circles not intersecting $\pl\Sigma$ or half-circles intersecting $\pl \Sigma$, decomposing this surface into Riemann surfaces with \emph{corners} (see Section \ref{sec:back} for more details).  In our probabilistic construction, the cut (circles or half-circles) 
will always be associated to Dirichlet boundary conditions for the  GFF  on the surface, 
while the original boundary before cutting was associated to the Neumann condition for the GFF. 
This leads us to consider surfaces with corners where each boundary circle or half-circle will be equipped with a marking, N or D (for Neumann and Dirichlet), and two adjacent half circles have necessarily different markings. We further impose that such a surface is, roughly speaking,``half a Riemann surface with boundary": this means  that if we double the surface with corners along the boundary with marking N,  we get a Riemann surface with boundary, see Figure \ref{f:corner}. Technically speaking, this forces the corner angles to be $\pi/2$. We shall also assume that the D-marked boundary components are real-analytic.\\

\textbf{Hilbert spaces.} To the D-marked boundary circles or half-circles, we associate two types of 
Hilbert spaces, which will serve to encore the Dirichlet boundary conditions. A circle will be parametrized by $\T=\{ z\in \C\,|\, |z|=1\}$ and a half-circle by $\T^+:=\{ z\in \C\,|\, |z|=1, {\rm Im}(z)\geq 0\}$, and we define 
\begin{align*} 
& \textrm{Hilbert space of the circle:}\qquad \mc{H}:=L^2(H^{-s}(\T),\mu_0) ,   \\
&\textrm{Hilbert space of the half-circle:}\qquad \mc{H}_+:=L^2(H^{-s}_{\rm even}(\T),\mu_0^+) 
\end{align*} 
where $H^{-s}(\T)$ is the Sobolev space of order $-s<0$ (for some $s$ fixed), $H^{-s}_{\rm even}(\T)$ is the closed subspace of $H^{-s}(\T)$ consisting of even functions with respect to the involution $z\mapsto \bar{z}$, $\mu_0$ (resp. $\mu_0^+$) is respectively the distribution law of the random variable $c+\varphi$ on $\T$ (resp. $c+\varphi^h$ on $\T$) where $c\in \R$ is a constant sampled using Lebesgue measure $dc$ and $\varphi,\varphi^h$ are the random centred Gaussian distributions on $\T$ with covariances 
\[\E[\varphi(e^{i\theta})\varphi(e^{i\theta'})]=-\log|e^{i\theta}-e^{i\theta'}|,\quad \E[\varphi^h(e^{i\theta})\varphi^h(e^{i\theta'})]=-\log(|e^{i\theta}-e^{i\theta'}||e^{i\theta}-e^{-i\theta'}|).\]

\textbf{Amplitudes of surfaces with corners.} To each surface $\Sigma$ equipped with a compatible Riemannian metric $g$, and with some D-marked boundary circles $\mc{C}_1,\dots,\mc{C}_{b_\ell}$ and half-circles $\mc{B}_1,\dots,\mc{B}_{b_h}$
parametrized by some real-analytic maps $\zeta_j^{\ell}:\T \to \mc{C}_j$ and $\zeta^h_{j}:\T^+\to \mc{B}_j$, the Segal functor   associate an amplitude (with $\boldsymbol{\zeta}$ denoting the collection of parametrizations)
\[ \mc{A}_{\Sigma,g,\boldsymbol{\zeta}}\in \mc{H}^{\otimes b_\ell}\bigotimes \mc{H}_+^{\otimes b_h}\] 
More generally we can add a collection of marked points ${\bf x}=(x_1,\dots,x_m)$ in the interior $\mathring{\Sigma}$, with associated weights $\boldsymbol{\alpha}=(\alpha_1,\dots,\alpha_m)\in (-\infty,Q)^m$, and points ${\bf s}=(s_1,\dots,s_{m_B})$ on the Neumann boundary,   with weights $\boldsymbol{\beta}=(\beta_1,\dots,\beta_{m_B})\in (-\infty,Q)^{m_B}$, satisfying the Seiberg bounds (with $\chi(\Sigma)$ the Euler characteristic)
\[\sum_{i}\alpha_i + \sum_j \frac{\beta_j}{2}>Q\chi(\Sigma).\]
Under this assumption, we define the Liouville amplitude 
\[\mc{A}_{\Sigma,g,{\bf x},{\bf s},\boldsymbol{\alpha},\boldsymbol{\beta},\boldsymbol{\zeta}}\in \mc{H}^{\otimes b_\ell}\bigotimes \mc{H}_+^{\otimes b_h}=L^2(H^{-s}(\T)^{b_\ell}\times H^{-s}_{\rm even}(\T)^{b_{h}},\mu_0^{\otimes b_\ell}\otimes {\mu_0^+}^{\otimes b_h})\]
on a surface with corners with marked points in Definition \ref{def:amp} as a (limit of) conditional 
expectations of some random variables constructed using the GFF on $\Sigma$ with Neumann boundary condition on the N-marked boundary and with conditional value on the D boundary components. Technically this is done using the GFF with mixed (Neumann on N and Dirichlet on D-marked boundaries) condition and the harmonic extension on $\Sigma$ of $b_\ell$ independent GFF $\varphi\in H^{-s}(\T)$ on the D-marked boundary circles and $b_h$ independent GFF $\varphi^h\in H^{-s}_{\rm even}(\T)$ on the D-marked half-circles.\\

\textbf{Gluing theorem.} Our main theorem is the gluing Segal axioms\footnote{In the statement below, we number the D-marked circles and half-circles in a way that the glued circles (resp. glued half-circles) are the first one (resp. the last ones) only to simplify the notations in the composition formula.}. 
 \begin{theorem}
For $j=1,2$, let $(\Sigma_j,g_j,\boldsymbol{\zeta}_j)$ be two surfaces with corners and parametrized D-marked analytic boundary circles 
$(\mc{C}_{jk})_{k=1,\dots,b_\ell^j}$ and half-circles $(\mc{B}_{jk})_{k=1,\dots,b_h^j}$, let $({\bf x}_j,\boldsymbol{\alpha}_j)$ be a collection of interior marked points with weights and $({\bf s}_j,\boldsymbol{\beta}_j)$ boundary marked points with weights.\\ 
1) \textbf{Gluing of circles.} Let $(\Sigma,g,\boldsymbol{\zeta})$ be the surface obtained by gluing $\Sigma_1$ to $\Sigma_2$ by identifying $\mc{C}_{11}$ to $\mc{C}_{21}$ thanks to the parametrizations, and let $({\bf x},\boldsymbol{\alpha}), ({\bf s},\boldsymbol{\beta})$ the collection of marked points with weights inherited from $\Sigma_1,\Sigma_2$. Then  for $\tilde{\boldsymbol{\varphi}}_j 
\in H^{-s}(\T)^{b_\ell^j-1}\times H^{-s}_{\rm even}(\T)^{b_h^j}$
\[\mc{A}_{\Sigma,g,{\bf x},{\bf s},\boldsymbol{\alpha},\boldsymbol{\beta},\boldsymbol{\zeta}}(\tilde{\boldsymbol{\varphi}}_1,\tilde{\boldsymbol{\varphi}}_2)=C\int_{H^{-s}(\T)} \caA_{\Sigma_1,g_1,{\bf x}_1,\boldsymbol{\alpha}_1,{\bf s}_1,\boldsymbol{\beta}_1,\boldsymbol{\zeta}_1}(\tilde{\varphi} ,\tilde{\boldsymbol{\varphi}}_1)\caA_{\Sigma_2,g_2,{\bf x}_2,\boldsymbol{\alpha}_2,{\bf s}_2,\boldsymbol{\beta}_2,\boldsymbol{\zeta}_2}(\tilde{\varphi},\tilde{\boldsymbol{\varphi}}_2) \dd\mu_0(\tilde{\varphi})\]
for some explicit constant $C$.
\\
2) \textbf{Gluing of half-circles.} Let $(\Sigma,g,\boldsymbol{\zeta})$ be the surface obtained by gluing $\Sigma_1$ to $\Sigma_2$ by identifying $\mc{B}_{1b_h^1}$ to $\mc{B}_{2b_h^2}$ thanks to the parametrizations, and let $({\bf x},\boldsymbol{\alpha}), ({\bf s},\boldsymbol{\beta})$ the collection of marked points with weights inherited from $\Sigma_1,\Sigma_2$. Then  for $\tilde{\boldsymbol{\varphi}}_j 
\in H^{-s}(\T)^{b_\ell^j}\times H^{-s}_{\rm even}(\T)^{b_h^j-1}$
\[\mc{A}_{\Sigma,g,{\bf x},{\bf s},\boldsymbol{\alpha},\boldsymbol{\beta},\boldsymbol{\zeta}}(\tilde{\boldsymbol{\varphi}}_1,\tilde{\boldsymbol{\varphi}}_2)=C^+\int_{H^{-s}_{\rm even}(\T)} \caA_{\Sigma_1,g_1,{\bf x}_1,\boldsymbol{\alpha}_1,{\bf s}_1,\boldsymbol{\beta}_1,\boldsymbol{\zeta}_1}(\tilde{\boldsymbol{\varphi}}_1,\tilde{\varphi}^h)\caA_{\Sigma_2,g_2,{\bf x}_2,\boldsymbol{\alpha}_2,{\bf s}_2,\boldsymbol{\beta}_2,\boldsymbol{\zeta}_2}(\tilde{\boldsymbol{\varphi}}_2,\tilde{\varphi}^h) \dd\mu_0^+(\tilde{\varphi}^h)\]
for some explicit constant $C^+$.
\end{theorem}
For a more precise statement (including the gluing of several circles or half-circles), we refer to Proposition \ref{glueampli}; in particular we ask that the metric has a special form, called admissible, near the glued circles/half-circles, which requires them to be geodesics of lengths $2\pi$ and $\pi$.
We can also glue together two circles or two-half-circles of the same   connected Riemann surface with corners and a gluing formula for the corresponding amplitude 
is   proved in Proposition \ref{selfglueampli}.

%
 
 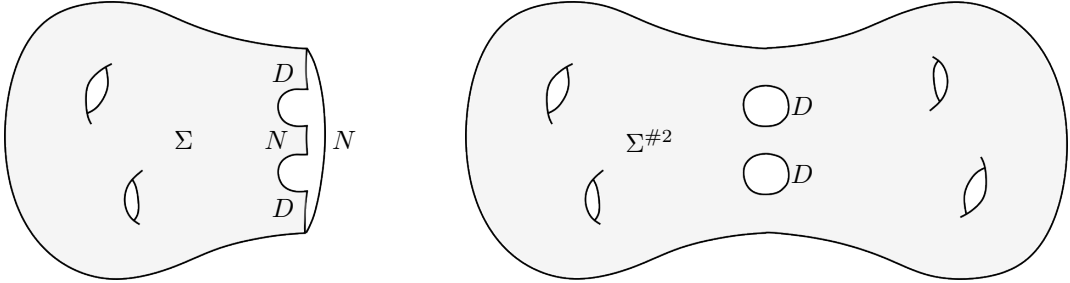
\begin{figure}[h]   
 \begin{tikzpicture}
    \node[inner sep=0pt] (pant) at (0,0)
{ \includegraphics[scale=0.55]{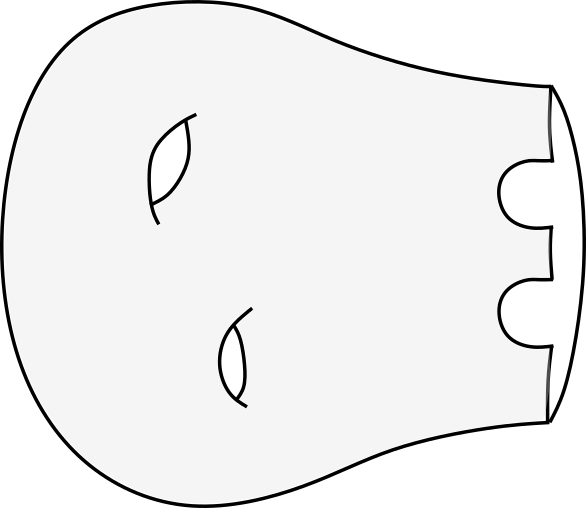}};
   \node[inner sep=0pt] (pant) at (8,0)
   { \includegraphics[scale=0.55]{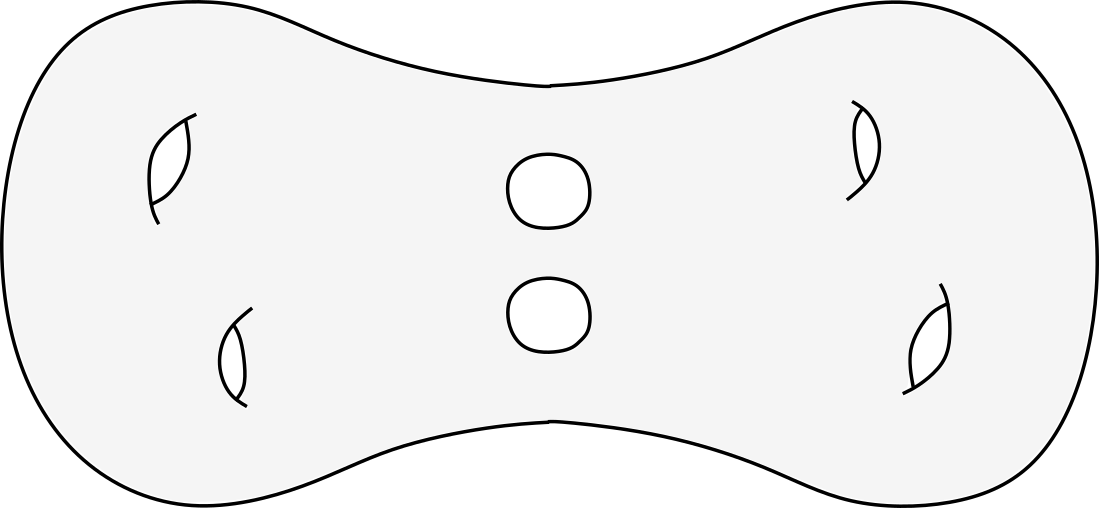}};
     \draw (2.1,0) node[right,black]{$N$} ;
      \draw (1.2,0) node[right,black]{$N$} ;
     \draw (1.3,0.9) node[right,black]{$D$} ;
       \draw (1.3,-0.9) node[right,black]{$D$} ;
        \draw (0,0) node[right,black]{$\Sigma$} ;
          \draw (8.2,0.45) node[right,black]{$D$} ;
       \draw (8.2,-0.45) node[right,black]{$D$} ;
         \draw (6,0) node[right,black]{$\Sigma^{\#2}$} ;
\end{tikzpicture}
\caption{On the left, a Riemann surface $\Sigma$ with corners and its boundary markings N and D. On the right, its Neumann double with its D-marked boundary circles.}\label{f:corner}
\end{figure}
 
\subsection{Segal semigroups} 
Next, we show that the Segal functor produces two semigroups: one corresponding to the family of round annuli 
\[\A_{e^{-t}}:=\{z\in \C\,|\, |z|\in [e^{-t},1]\} \quad \textrm{ for }t>0\] 
whose generator is an unbounded self-adjoint operator on $\mc{H}$ called the \emph{Hamiltonian (or bulk Hamiltonian)} ${\bf H}$ of the CFT, 
and another one corresponding to the gluing of half-annuli 
\[\A_{e^{-t}}^+:=\{z\in \C\,|\, |z|\in [e^{-t},1], {\rm Im}(z)\geq 0\} \quad \textrm{ for }t>0,\] 
which possesses a different generator  ${\bf H}_+$ and which we call the \emph{boundary Hamiltonian}, see Section \ref{sec:semigroup}. It is also an unbounded self-adjoint operator on $\mc{H}_+$.    
We emphasize that the Segal amplitudes of these annuli or half-annuli do not satisfy the Seiberg bounds, which means that they are not integral kernels of Hilbert-Schmidt operators, but they can still be given  sense   as bounded semi-groups on $\mc{H}$ and $\mc{H}_+$ respectively. 
The spectral resolution for the bulk Hamiltonian was established in \cite{GKRV20_bootstrap}. Here, we identify the quadratic form having the boundary Hamiltonian as Friedrichs extension in Section \ref{s:semigroup_half}.  This analysis  is more intricate than the bulk case and, to perform it,  we need to restrict  the range of parameters to $\gamma <\sqrt{2}$, or $\gamma\in (0,2)$ and $\mu_B=0$. This is related to the fact that it is complicated to extract a tractable  core for the quadratic form when $\gamma\geq \sqrt{2}$. 
In the subsequent work \cite{GRW2}, we shall provide a spectral resolution for the boundary Hamiltonian under the assumption $\gamma <\sqrt{2}$, or $\gamma\in (0,2)$ and $\mu_B=0$. 
We will also combine these results with the Ward identities to establish the conformal bootstrap in the case of boundary Liouville CFT. 

Let us finally mention here that the paper \cite{WuConfBoot} treats the conformal bootstrap in the very special case when the Riemann surface with boundary is an annulus with some marked points in a way that only the bulk Hamiltonian is involved in the conformal bootstrap. Our current paper together with the companion paper \cite{GRW2} bridge the gap towards the general bootstrap formula for boundary Liouville theory. 
 
 \subsection{Applications.}
The boundary Liouville theory is not only an extension of the Liouville CFT, but it turns out to have plenty of applications, including surprisingly for the Liouville theory itself. These applications are in part due to the link between boundary Liouville theory and the SLE  (see \cite{MatingOfTrees,Sheffield16_zipper, Nolin:2023vkn, ang2022fzz} etc.) but also to the relation of boundary Liouville CFT to quantum groups. In that respect, 
they give a strong motivation for developing its study.
We summarize below some of these applications.
\begin{enumerate}
\item \textbf{Conformal blocks.} 
A surprising application of the conformal bootstrap for boundary Liouville CFT is to provide a probabilistic representation of the spherical 4pt or toroidal 1pt conformal blocks; this was observed in \cite{remypreprint} inspired by the physics paper \cite{ponsot1999}. The work \cite{remypreprint}  shows in particular the analytic extension of the conformal block in terms of its parameters and its pointwise convergence (the convergence of the conformal block in \cite{GKRV20_bootstrap,GKRV21_Segal} was proved in $L^2$ sense as a function of the spectral parameter $Q+ip\in Q+i\R^+$,  but not pointwise in $p$).

\item \textbf{Representation of mapping class group and fusion kernels.} 
The present manuscript is an important step in a series of papers \cite{BGKR1,BGKR2} devoted to providing a unitary representation of the mapping class group in the space of conformal blocks, thus producing a geometric  
quantization of Teichm\"uller space as a consequence of a CFT construction. 
This representation was claimed in physics  \cite{Teschner2016} to coincide with the Chekhov-Fock-Goncharov-Kashaev representations of the mapping class group in their quantization of Teichm\"uller space \cite{Fock,kashaevQuantization,FockGoncharovIHES,FockGoncharovInventiones} based on hyperbolic geometry tools. 
The representation of each mapping class element can be represented in a basis of conformal blocks associated to a pair of pants decomposition of $\Sigma$, and can be expressed purely in terms of the operators, called \emph{fusion kernels}, corresponding to a change of basis by an elementary move in the pant decomposition. In the physics literature, it is established that for rational CFT the main fusion kernel involved in this representation is identical to the boundary three-point function \cite{Runkel:1998he, Behrend:1999bn, Fuchs:2004xi}. For Liouville theory, a similar relation/formula for the integral fusion kernels    is given in physics \cite{ponsot1999}  in terms of the structure constants of boundary Liouville CFT.  A recent preprint \cite{remypreprint} provides a mathematical proof of this fact, based in a fundamental way on the conformal bootstrap for boundary Liouville theory proved in the present work and its companion paper \cite{GRW2}.

\item \textbf{Law of random modulus.} Building on the special case of the conformal bootstrap in boundary Liouville proved in  \cite{WuConfBoot}, the law of the random modulus    for the Brownian Map   with the topology of an annulus was obtained in a recent work \cite{ARSmoduliRPM}.
\end{enumerate}

There are some other possible applications of the conformal bootstrap for boundary Liouville CFT. For example, it is conjectured in a recent physics preprint \cite{Chen:2024unp}  that the Liouville CFT partition function of a surface $\Sigma$ can be derived from quantities associated with a $3$-dimensional manifold $M$, where $\partial M=\Sigma$. We expect that the boundary bootstrap could lead to a rigorous derivation of such correspondence and that  
constructing $3$-dimensional invariants associated to Liouville CFT in the spirit of AdS/CFT holography  are promising directions for future research. We also expect that the conformal bootstrap for boundary Liouville CFT could give new results on the law of Gaussian Multiplicative Chaos on the circle.

 \subsection{Organisation of the paper and differences with Liouville CFT}
 
 The route we follow in this paper is inspired by the work \cite{GKRV20_bootstrap,GKRV21_Segal} for Liouville CFT on closed surfaces. However, there are several new difficulties and differences that need to be adressed. 

First, the (complex) geometric setting involves to work with Riemann surfaces with corners and requires some care about the type of Riemannian metrics that can be used to construct Liouville amplitudes -- this is done in Sections \ref{sec:surfaces} to \ref{sec:admissible_metrics}. In sections  \ref{sec:backclosed} and \ref{sec:open}, we recall the construction of Liouville path integral on  closed surfaces and surfaces with Neumann boundary.

Second, the definition of Liouville amplituds on surfaces with corners involves a Gaussian Free Field with mixed boundary conditions and to describe the properties of Dirichlet-to-Neumann maps (and their link to Green's functions) on such surfaces with corners, this is done in Section \ref{sec:decomp}. For the Segal gluing proof, we also need gluing formulas for determinants of Laplacians and Dirichlet-to-Neumann maps on such surfaces with corners, which has not been proved in the literature. This is done here in Appendix \ref{app:determinants}. We often rely on a doubling of the surface with corners argument to adress these issues. The proof of the gluing formula for Segal amplitudes of surfaces with corners is done in Section \ref{sec:gluing}. 

Third, the study in Section \ref{sec:semigroup} of the semi-group of half-annuli and in particular its generator ${\bf H}_+$ involves new singularities in the Gaussian Multiplicative Chaos measure and several potentials in the Hamiltonians.  
This require again some particular care and new estimates on the GMC (see Appendix \ref{app:GMC_estimates}).  
In terms of  the spectral theory of the boundary Hamiltonian, the study of the case $\gamma\geq \sqrt{2}$ involves more severe problems than for the bulk Hamiltonian dealt with in \cite{GKRV20_bootstrap}:  it is not clear at the moment if the generator ${\bf H}_+$ of the half-annuli semigroup can be related to the  Friedrichs extension of a quadratic form. This is a  technical issue that makes it hard to use the scattering approach from  
 \cite{GKRV20_bootstrap} when $\gamma\geq \sqrt{2}$.

\vskip 2mm
\noindent\textbf{Acknowledgements.} R. Rhodes acknowledges the support of the ANR-21-CE40-0003 and of the   Institut Universitaire de France (IUF).  This work was performed in part while C. Guillarmou was visiting Aspen Center for Physics, which is supported by National Science Foundation grant PHY-2210452.  B. Wu   was supported by National Key R\&D Program of China (No. 2023YFA1010700).

\section{Background and notations}\label{sec:back}

\subsection{Closed Riemann surfaces and Riemann surfaces with boundary/corners}\label{sec:surfaces}

We start with a few notations: we introduce for $\delta<1$
\[ \D := \{z\in \C\,|\, |z|\leq 1\}, \quad \H=\{z\in \C\,|\, {\rm Im}(z)\geq 0\},\quad  \D^+:=\D\cap \H, \quad  \D^{++}:=\D^+\cap \{{\rm Re}(z)\geq 0\}\]
\[ \T :=\{z\in \C\,|\, |z|=1\}, \quad \T^+:=\T\cap \H=\{ z\in \C\,|\, |z|=1, {\rm Im}(z)\geq 0\}.\]
\[\mathbb{A}_{\delta}:=\{z\in \C\,|\, |z|\in [\delta,1]\}, \quad \mathbb{A}_{\delta_1,\delta_2}:=\{z\in \C\,|\, |z|\in [\delta_1,\delta_2]\} , \quad \A^+_\delta:=\H\cap \A_\delta, \quad  \A^+_{\delta_1,\delta_2}:=\H\cap \A_{\delta_1,\delta_2}.\]

A closed Riemann surface $(\Sigma,J)$ is a smooth oriented compact surface $\Sigma$ with no boundary, equipped with a complex structure $J$ (i.e. $J \in {\rm End}(T\Sigma)$ with $J^2=-{\rm Id}$), or equivalently a set of charts $\omega_j:U_j\to \mathring{\D}\subset \C$ such that $\omega_j\circ \omega_{k}^{-1}$ is a biholomorphic map where it is defined. The complex structure $J$ is the canonical one (i.e. $J\pl_x=\pl_y, J\pl_y=-\pl_x$) when viewed in $\mathring{\D}$ via the charts. \\

Next, we consider compact Riemann surfaces with boundary.
\begin{definition}[\textbf{Riemann surfaces with boundary}]\label{d:surf_with_bdry}
A Riemann surface with boundary is a smooth compact surface $\Sigma$ with $b\geq 1$ boundary circles $\pl \Sigma=\sqcup_{j=1}^b\pl_j\Sigma$, that is 
equipped with an atlas made of smooth charts $\omega_i:V_i\to \D\subset \C$ 
such that $\omega_i(V_i\cap \pl \Sigma)\subset \R$ and for each $j,k$, the map 
$\omega_i\circ \omega_{k}^{-1}$ extends as a biholomorphic map on a neighborhood of 
$\omega_{k}(V_k\cap V_i)$ in $\C$. This set of charts equips $\Sigma$ with a complex structure and 
the boundary $\pl \Sigma$ is real analytic with respect to this complex structure. A map $f:U\subset \Sigma\to \C$  is 
said holomorphic if $f\circ \omega_i^{-1}:\omega_i(V_i\cap U)\to \C$ is holomorphic for all $i$ such that $V_i\cap U\not=\emptyset$.

It will be later convenient to further consider a marking of the boundary circles: each boundary circle is assigned with a label $D$ or $N$. We will write $\partial \Sigma_N$, resp. $\partial \Sigma_D$, for the union of boundary components with marking $N$, resp. $D$. These markings will later correspond to boundary conditions imposed to our fields on the boundary circles (Dirichlet or Neumann). If $\pl \Sigma_D=\emptyset$, we call $\Sigma$ a compact Riemann surface with Neumann boundary.
\end{definition}
We emphasize that our convention for Riemann surfaces with boundary is that the boundary has a real analytic structure, and that such Riemann surface can be extended to an open Riemann surface in a way that $\pl \Sigma$ is a union of analytic curves. 
If one only asks that the coordinate changes $\omega_k\circ \omega_j^{-1}$ are holomorphic on $\omega_k(V_k\cap V_i)$ but only smooth up to the boundary $\omega_k(V_k\cap V_i)\cap \R$, we will say that the surface is a Riemann surface with \textbf{smooth} boundary. Taking a small annular neighborhood $A_j$ of $\pl_j\Sigma$ and gluing a disk $\D$ to $A_j$ we produce another disk $D_j$ with a Riemann surface structure as above, and by the uniformisation theorem  there is a biholomorphic map $\omega_j:D_j\to \D$ which is analytic (resp. smooth) up to boundary if  $\pl_j\Sigma$ is analytic (resp. smooth). Therefore, there exist a biholomorphic map 
\begin{equation}\label{omega_j}
\omega_j:V_j\to \A_\delta
\end{equation} 
where $V_j$ are neighborhoods of the boundary circles $\pl_j\Sigma$ (with  $\delta<1$), $\omega_j(\pl_j\Sigma)=\T$ and 
$\omega_j$ is analytic up to $\pl_j\Sigma$ (resp. smooth up to $\pl_j\Sigma$). It will be convenient later to include these particular charts, called \textbf{annular holomorphic charts}, as part of the atlas when dealing with objects defined near the boundary circles.\\

We also need to consider more general surfaces which are surfaces with corners. We refer to \cite[Chapter 1]{Melrose} for the definition of a smooth surface with corner, and the notion of boundary hypersurfaces. We simply recall that boundary hypersurfaces are codimension $1$ manifolds with or without boundary, that are contained in the topological boundary; for example $[0,1]^2$ has $4$ boundary hypersurfaces given by the $4$ closed edges of the square. 
In our case, the boundary hypersurfaces are either circles or half-circles (i.e. diffeomorphic to $[0,1]$). The connected components of $\pl \Sigma$ are either smooth embedded circles or a piecewise smooth union of half-circles, with singularities at corner points.
\begin{definition}[\textbf{Riemann surfaces with corners}]\label{def:mfd_with_corners}
Let $\Sigma$ be a smooth compact oriented surface with corners, with topological boundary being denoted $\partial\Sigma$ and the connected components of $\pl \Sigma$ are denoted $\pl_1\Sigma,\dots,\pl_b\Sigma$. The corners, of codimension $2$, is a collection of points $\mc{K}:=\{p_1,\dots, p_c\}\subset \pl \Sigma$.
We say that $\Sigma$ is a Riemann surface with corners if it is
equipped with a set of smooth charts $\omega_j:V_j\to \omega_j(V_j)\subset \C$ satisfying the following properties:
\begin{enumerate}
\item $\omega_j(V_j)=\D$ if $V_j\cap \pl \Sigma=\emptyset$,
\item $\omega_j(V_j)=\D^+$ if $V_j\cap \pl \Sigma\not=\emptyset$ but $V_j\cap \mc{K}=\emptyset$,
\item $\omega_j(V_j)=\D^{++}$ if $V_j\cap \mc{K}\not=\emptyset$,
\end{enumerate}
such that for each $i,k$, the map $\omega_i\circ \omega_{k}^{-1}$ extends as a biholomorphic map on a neighborhood of 
$\omega_{k}(V_k\cap V_i)$ in $\C$.
 The charts induce a complex structure $J$ and a map $f:U\subset \Sigma\to \C$  is 
said holomorphic if $f\circ \omega_j^{-1}:\omega_i(V_i\cap U)\to \C$ is holomorphic for all $i$ such that $V_j\cap U\not=\emptyset$. The connected boundary hypersurfaces without corner points are analytic circles while the boundary hypersurfaces with corner points are diffeomorphic to $[0,1]$ and called \textbf{boundary half-circles}. Moreover $\pl_j\Sigma$ is either a boundary circle or a finite union of boundary half-circles containing corner points.

 We  further impose a marking of the boundary: each boundary hypersurface is assigned with a label $D$ or $N$ in such a way that each corner point is an endpoint of exactly two half-circles,  one with marking $N$ and one with marking $D$. 
 This marking imposes in particular that each boundary component of $\Sigma$ has an even number of corner points (see Figure \ref{f:corner}) and that the numbers of half-circles marked $D$ or $N$ are the same. A Riemann surface with corners with $b_\ell$ boundary circles, among which $b_\ell^D$ marked $D$ and $b_\ell^N$ marked $N$,   and $b_h$ half-circles, among which $b_h^D$ marked $D$ and $b_h^N$ marked $N$, is said to have $(b_\ell,b_h)$ boundary components. 
 
We will write $\partial \Sigma_N$, resp. $\partial \Sigma_D$, for the union of boundary components with marking $N$, resp. $D$. 
Note that both of $\partial \Sigma_N$ or $\partial \Sigma_D$ could be empty.  
\end{definition}

We notice that our convention for Riemann surfaces with corners is that the boundary has a piecewise real analytic structure, and the Riemann surface can be extended to an open Riemann surface in a way that $\pl \Sigma$ is a union of piecewise analytic curves. A manifold with boundary is also a particular case of manifold with corners (but actually without corners).

 There is an orientation on $\Sigma$ given by a non-vanishing $2$-form $\omega_\Sigma$ which has the property that for each holomorphic chart $\omega_j$, $(\omega_j^{-1})^*\omega_\Sigma=e^{\rho_j}dx\wedge dy$ where $z=x+iy$ is the complex coordinate on $\C$ and $\rho_j$ is a smooth function. Each boundary hypersurface (more precisely the interior of such boundary hypersurface) inherits an orientation by pulling back the $1$-form $-\iota_{\nu}\omega_\Sigma$ on the hypersurface if $\nu$ is a non-vanishing inward pointing vector field.
 
On a Riemann surface with boundary or with corners, we say that a smooth Riemannian metric $g$ is compatible with the complex structure if in each holomorphic chart $\omega_i:V_i \to \C$, one has $g|_{V_i}=\omega_i^*(e^\rho |dz|^2)$ for some smooth $\rho\in C^\infty(\omega_i(V_i))$.
Such a metric induces a volume form ${\rm dv}_g$ and a boundary measure ${\rm d\ell}_g$ on $\partial \Sigma$. We will denote $K_g$ the scalar curvature and $k_g$ the geodesic curvature along the interior of the boundary hypersurfaces. At the corner point, there is a $\pi/2$ angle (independent of the choice of compatible metric $g$) between two boundary hypersurfaces in view of our particular choices of holomorphic charts in Definition \ref{def:mfd_with_corners}.
  
We consider the standard non-negative Laplacian $\Delta_{g}=d^*d$ where $d^*$ is the $L^2$ adjoint of the exterior derivative $d$. We denote by $ \partial_\nu$   the (unit) normal inward derivative along the boundary with respect to the metric $g$. A fixed set of points $z_i\in\Sigma$, $i=1,\dots,m$ on $\Sigma$ are called \emph{marked points} on $\Sigma$; they could belong to $\partial \Sigma$.

\subsection{Parametrized boundaries}\label{sub:param}
On a  Riemann surface with corners, the boundary is decomposed as $\partial\Sigma=\partial\Sigma_N\cup \partial\Sigma_D$, which can be described as a collection of simple non overlapping curves made up of $b_\ell$ analytic closed curves $\mc{C}:=\sqcup_{j=1}^{b_\ell}\mathcal{C}_j$ (the index $\ell$ stands for loops), enumerating the boundary circles,  and $b_h$ analytic half circles $\mc{B}:=\sqcup_{j=1}^{b_h}\mathcal{B}_j$ (the index $h$ stands for half circles), enumerating the boundary half-circles between corner points. Also we can split these (half-)circles according to their markings: $\mc{C}:=(\sqcup_{j=1}^{b^D_\ell}\mathcal{C}^D_j)\sqcup (\sqcup_{j=1}^{b^N_\ell}\mathcal{C}^N_j )$ and $\mc{B}:=(\sqcup_{j=1}^{b^D_\ell}\mathcal{B}^D_j)\sqcup (\sqcup_{j=1}^{b^N_\ell}\mathcal{B}^N_j )$. In what follows, we will introduce parametrizations of the Dirichlet boundary components.

Any non-vanishing vector field $v$ on $\mc{C}_j^D$ (resp. $\mc{B}_j^D$) induces an orientation: we say it is outgoing if $(v,\nu)$ is a positively oriented basis of $T\Sigma$ on $\mc{C}_j^D$ (resp. on $\mc{B}_j^D$), where we recall that $\nu$ is any inward pointing non-vanishing vector field at $\mc{C}_j^D$ (resp. $\mc{B}_j^D$); conversely it is incoming if  $(v,\nu)$ is negatively oriented in $T\Sigma$. Note that we could introduce In/Out orientations of the Neumann boundary components as well but we will refrain from doing so because this will be useful only in the case of Dirichlet boundary components to define the Segal amplitudes.\\

%

For each $j\in [1,b^D_\ell]$, we call \textbf{parametrization} of $\mc{C}_j^D$ any real analytic diffeomorphism $\zeta_j^{\ell}:\T\to \mc{C}^D_j$, where the real analyticity means that for each open set $V_i$ such that $V_i\cap \mc{C}_j^D\not=\emptyset$, 
and $\omega_i:V_i\to \D^+$  a holomorphic chart (with  $\omega_i(V_i\cap \mc{C}_j)\subset (-1,1)$), the map $\omega_i \circ \zeta_j^{\ell}: I_{ij}\to \omega_i(V_i\cap \mc{C}_j^D)$ extends holomorphically near $I_{ij}=(\zeta_j^\ell)^{-1}(V_i\cap \mc{C}_j^D)\subset \T$ in $\C$. 
Such parametrization induces an orientation $\sigma^\ell_j\in\{-1,+1\}$:  if the vector field
$v=\pl_\theta \zeta_j^{\ell}(e^{i\theta})$ induces an outgoing orientation, 
we say that $\zeta_j^\ell$ is \textbf{outgoing} and we set $\sigma^\ell_j:=-1$, conversely we say that $\zeta_j^\ell$ is \textbf{incoming} if $v$ is incoming and we set $\sigma^\ell_j:=1$.
Note that the inverse of a  parametrization $(\zeta^\ell_j)^{-1}:\mc{C}^D_j\to \T$ induces a holomorphic chart near $\mc{C}_j^D$ by holomorphic extension to an annular neighborhood $V_j$ mapping to $\A_\delta$ for some $\delta<1$ if $\zeta_j^\ell$ is outgoing and to $\A_{1,\delta^{-1}}$ if $\zeta_j^\ell$ is incoming; these charts can also be used in the atlas. Conversely, the chart $\omega_j$ in \eqref{omega_j} induces a (incoming) parametrization of $\mc{C}_j^D$; to make it outgoing it suffices to consider $1/\omega_j$.\\ 

Similarly, for each $j\in [1,b_h^D]$, we call \textbf{parametrization} of $\mc{B}_j^D$ any real analytic diffeomorphism $\zeta_j^{h}:\T^+\to \mc{B}_j^D$ such  that for each open set $V_i$ with $V_i\cap \mc{B}^D_j\not=\emptyset$ and $V_i\cap \mc{K}=\emptyset$, for $\omega_i:V_i\to \D^+$ a holomorphic chart and $\omega_i(V_i\cap \mc{B}^D_j)\subset (-1,1)$, the map $\zeta_j^{h} :I_{ij}\to \omega_i(V_i\cap \mc{B}_j^D)$ extends holomorphically to a neighborhood of $I_{ij}=(\zeta_j^h)^{-1}(V_i\cap \mc{B}^D_j)\subset \T^+$ in $\C$, and for each open set $V_i$ such that $V_i\cap \mc{B}_j^D\not=\emptyset$ with $V_i\cap \mc{K}\not=\emptyset$, for each holomorphic chart $\omega_i:V_i\to \D^{++}$, the map $\omega_i\circ \zeta_j^h :I_{ij}\to \omega_i(V_i\cap \mc{B}_j^D)$ extends holomorphically to a neighborhood $\Omega$ of $I_{ij}$ 
in $\C$ and maps $\Omega \cap \R$ to an interval in $[0,1)$. Such parametrization induces an orientation $\sigma^h_j\in\{-1,+1\}$:  if the vector field
$v=\pl_\theta \zeta_j^{h}(e^{i\theta})$ induces an outgoing orientation, 
we say that $\zeta_j^h$ is \textbf{outgoing} and we set $\sigma^h_j:=-1$, conversely we say that $\zeta_j^h$ is \textbf{incoming} if $v$ is incoming and we set $\sigma^h_j:=1$.
The inverse of an analytic parametrization $(\zeta^h_j)^{-1}:\mc{B}^D_j\to \T^+$ produces a holomorphic chart 
$\omega_j^h$ by holomorphically extending $(\zeta^h_j)^{-1}$ to a half annular neighborhood of $\T^+\subset \C$  and mapping to 
$\A_\delta^+$ for some $\delta<1$ if $\zeta^h_j$ is outgoing and to $\A_{1,\delta^{-1}}^+$ if it is incoming. 
We will need to prove that parametrization of half-circles do exist, this will be obtained using the double of $\Sigma$ in the next section.

\begin{definition}[Riemann surface $\Sigma$ with corners and parametrized boundary]
A surface $\Sigma$ with corners and parametrized boundary is a surface with corners equipped with a parametrization of its boundary components as above.
\end{definition}

\subsection{Gluing of Riemann surfaces}\label{sub:gluing}
Consider a Riemann surface with corners as above (the marking plays no role for what follows beyond the fact that only Dirichlet boundary components carry a parametrization). First we explain the gluing of boundary circles. The inverse of the parametrizations of the boundary circles and half-circles induce 
holomorphic charts using holomorphic extension, as explained before. We call 
$\omega_j^\ell:V_j^\ell\to \A_{1,\delta^{-1}}$ the holomorphic extension of $(\zeta_j^{\ell})^{-1}$ near  $\mc{C}_j^D$ if $\zeta_j^\ell$ is incoming (resp. $(\zeta_j^{\ell})^{-1}:V_j^\ell\to \A_\delta$ if $\zeta_j^{\ell}$ is outgoing); note that 
$\omega_j^\ell$ is a half-annular neighborhood.
If $\mc{C}^D_j$ is outgoing and $\mc{C}^D_k$ is incoming, we can define a new Riemann surface $\Sigma_{j\#_\ell k}$ with $(b_\ell-2)$-boundary circles by gluing/identifying $\mc{C}^D_j$ with $\mc{C}^D_k$ as follows: 
we identify $\mc{C}^D_j$ with $\mc{C}^D_k$ by identifying the points $\zeta^\ell_j(e^{i\theta})$ with $\zeta^\ell_k(e^{i\theta})$. 
A neighborhood in $\Sigma_{j\#_\ell k}$ of the identified circle $\mc{C}_j\simeq \mc{C}_k$ is given by $(V^\ell_j\cup V^\ell_k)/\sim $ (where $\sim$ means the identification $\mc{C}_j\sim\mc{C}_k$) which identifies with the annulus $\mathbb{A}_{\delta,\delta^{-1}}=\{|z|\in [\delta,\delta^{-1}]\}$ by the chart
\[\omega_{jk}: z\in V^\ell_j \mapsto \omega^\ell_j(z)\in \mathbb{A}_{\delta}, \quad z\in V_k^\ell \mapsto \omega^\ell_k(z)\in \mathbb{A}_{1,\delta^{-1}}.\]
Note that on $\T$, the map is 
\begin{equation}\label{gluingonboundary}
\omega_{jk}(\zeta^\ell_j(e^{i\theta}))=\omega^\ell_j(\zeta^\ell_j(e^{i\theta}))=
e^{i\theta}=\omega^\ell_k(\zeta^\ell_k(e^{i\theta}))=\omega_{jk}(\zeta^\ell_k(e^{i\theta})).
\end{equation}
Since $\zeta_j^\ell,\zeta_k^\ell$ are analytic, this procedure produces a complex structure on the glued surface.
The resulting Riemann surface depends on the choice of parametrization $\zeta^\ell_j,\zeta^\ell_k$. We notice that if $\mc{C}_j^D$ and $\mc{C}_k^D$ belong to different connected components of $\Sigma$, $b_0(\Sigma_{j\#_\ell k})=b_0(\Sigma)-1$  if $b_0$ denotes the $0$-th Betti number; otherwise $b_0(\Sigma_{j\#_\ell k})=b_0(\Sigma)$.

 \begin{figure}[h] 
\begin{tikzpicture}[scale=1]
\draw[very thick] (-6,0) arc (0:180:1);
\draw[very thick] (-5.5,0) arc (0:180:1.5);
\draw[very thick] (-8.5,0) -- (-8,0);
\draw[very thick] (-5.5,0) -- (-6,0);
 \draw (-7.5,1.5) node[above]{$\textrm{ In }$} ;
  \draw (-7,0.5) node[above]{$\textrm{ Out }$} ;
 \draw (-3.3,-0.5) node[above]{$N$} ;
  \draw (-0.8,-0.5) node[above]{$N$} ;
  \draw (-3.8,-0.2) node[above]{$-\frac{1}{\delta}$} ;  
  \draw (-2.7,-0.2) node[above]{$-1$} ;
   \draw (-0.3,-0.2) node[above]{$\frac{1}{\delta}$} ;  
  \draw (-1.2,-0.2) node[above]{$1$} ;
    \draw (-2,2) node[above]{$\A_{1,\delta^{-1}}^+=\omega_k^h(V_k^h)$} ;
    \draw[very thick] (-1,0) arc (0:180:1);
\draw[very thick] (-0.5,0) arc (0:180:1.5);
\draw[very thick] (-3.5,0) -- (-3,0);
\draw[very thick] (-0.5,0) -- (-1,0);
 \draw (-2.5,1.5) node[above]{$\textrm{ Out }$} ;
  \draw (-2,0.5) node[above]{$\textrm{ In }$} ;
 \draw (-8.3,-0.5) node[above]{$N$} ;
  \draw (-5.8,-0.5) node[above]{$N$} ;
  \draw (-8.8,-0.2) node[above]{$-1$} ;  
  \draw (-7.7,-0.2) node[above]{$-\delta$} ;
   \draw (-5.3,-0.2) node[above]{$1$} ;  
  \draw (-6.2,-0.2) node[above]{$\delta$} ;
    \draw (-7,2) node[above]{$\A_\delta^+=\omega_j^h(V_j^h)$} ;
      \draw (-4.6,0.5) node[above]{$z\mapsto -1/z$} ;
\end{tikzpicture}
\caption{Gluing two neighborhoods of half-circles, each with D marking}\label{pic1}
\end{figure}
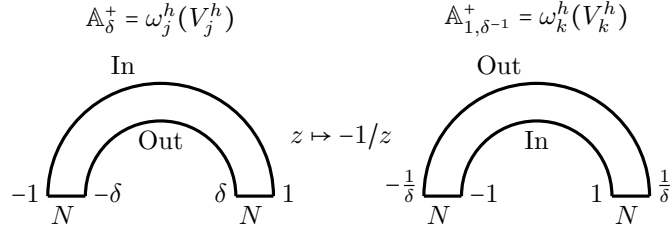
 \begin{figure}[h] 
\begin{tikzpicture}[scale=1]
\draw[very thick] (1,0) arc (0:180:1);
\draw[very thick] (2,0) arc (0:180:2);
\draw[dashed] (1.5,0) arc (0:180:1.5);
\draw[very thick] (1,0) -- (2,0);
\draw[very thick] (-1,0) -- (-2,0);
 \draw (0,2) node[above]{$\textrm{ In }$} ;
  \draw (0,0.5) node[above]{$\textrm{ Out }$} ;
    \draw (1.5,-0.5) node[above]{$N$} ;
    \draw (-1.5,-0.5) node[above]{$N$} ;
      \draw (0.8,-0.2) node[above]{$\delta$} ;  
      \draw (2.2,-0.2) node[above]{$\frac{1}{\delta}$} ;  
        \draw (-0.7,-0.2) node[above]{$-\delta$} ;  
      \draw (-2.3,-0.2) node[above]{$-\frac{1}{\delta}$} ;  
   \end{tikzpicture}
\end{figure}

Next we explain the gluing of boundary half-circles and we assume that the surface has $b_\ell$ boundary circles and $b_h$ boundary half-circles (half of them being marked $D$). If $\mc{B}_j$ is outgoing and $\mc{B}_k$ is incoming, we can define a new Riemann surface $\Sigma_{j\#_h k}$ by gluing/identifying $\mc{B}_j$ with $\mc{B}_k$ as follows: 
we identify $\mc{B}_j$ with $\mc{B}_k$ by identifying $\zeta^h_j(e^{i\theta})$ with $\zeta^h_k(e^{i(\pi-\theta)})$.  A neighborhood in $\Sigma_{j\#_h k}$ of the identified circle $\mc{B}_j\simeq \mc{B}_k$ is given by $(V^h_j\cup V^h_k)/\sim $ (where $\sim$ means the identification $\mc{B}_j\sim\mc{B}_k$ as above) which identifies with the half-annulus $\mathbb{A}^+_{\delta,\delta^{-1}}$ by the chart (see Figure \ref{pic1} and Figure \ref{glue1})
\[\omega_{jk}: z\in V^h_j \mapsto  \omega^h_j(z)  \in \mathbb{A}^+_{\delta}, \quad z\in V_k\mapsto -\frac{1}{\omega^h_k(z)}\in \mathbb{A}^+_{1,\delta^{-1}}.\]
 
Note that 
\begin{equation}\label{gluingonhalf}
\omega_{jk}(\zeta^h_j(e^{i\theta}))= \omega^h_j(\zeta^h_j(e^{i\theta})) =
e^{i\theta}=-\frac{1}{ e^{i(\pi-\theta)}}= -\frac{1}{\omega^h_k(\zeta^h_k(e^{i\theta}))}=\omega_{jk}(\zeta^h_k(e^{i\theta})).
\end{equation}
Also, we stress that, after the gluing procedure, the number of corner points has decreased by $4$ units as the former corner points involved in the gluing become regular (i.e. locally charted by the half disk). The resulting Riemann surface depends on $\zeta^h_j,\zeta^h_k$.  The marking plays no role in the gluing just described but in the following we will only glue half circles with the same marking $D$. 
 \begin{figure}[h]
  \begin{tikzpicture}
    \node[inner sep=0pt] (pant) at (0,0)
{ \includegraphics[scale=0.4]{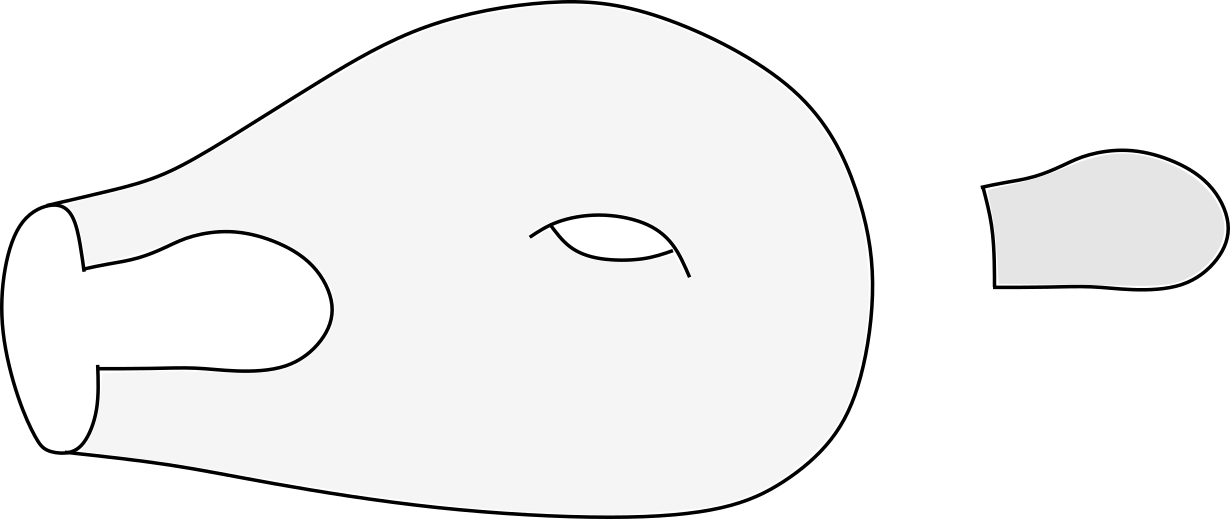}};
  \node[inner sep=0pt] (pant) at (8,0)
{ \includegraphics[scale=0.4]{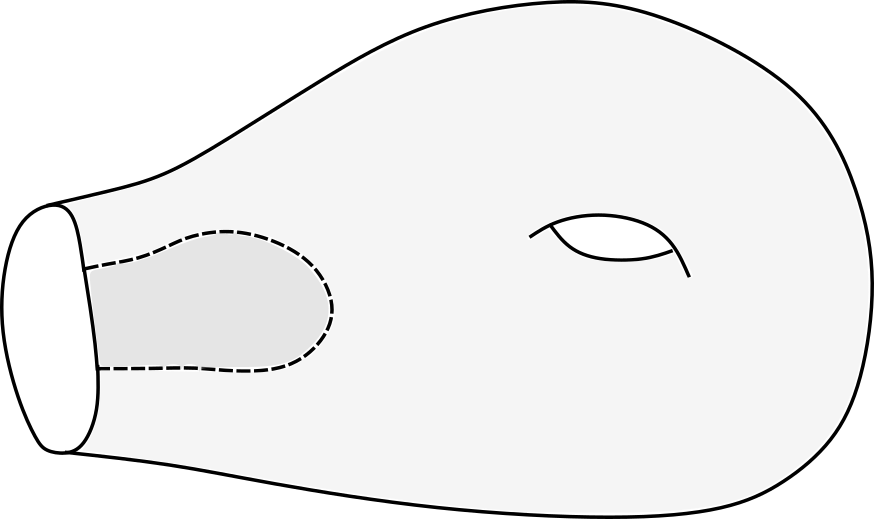}};
  \draw (7.5,0) node[right,black]{$\Sigma$} ;
    \draw (-1,1) node[right,black]{$\Sigma_1$} ; 
     \draw (2.5,0.2) node[right,black]{$\Sigma_2$} ; 
      \draw (2.6,-0.4) node[right,black]{$D$} ; 
       \draw (1.45,0) node[right,black]{$N$} ; 
        \draw (-3.7,-0.7) node[right,black]{$N$} ; 
         \draw (-2.1,-0.2) node[right,black]{$D$} ;
          \draw (5.2,-0.2) node[right,black]{$N$} ; 
\end{tikzpicture}
\caption{Two surfaces with corners $\Sigma_1,\Sigma_2$ glued along half-circles to produce the surface $\Sigma$ with N-marked boundary.}
\label{glue1}
\end{figure}
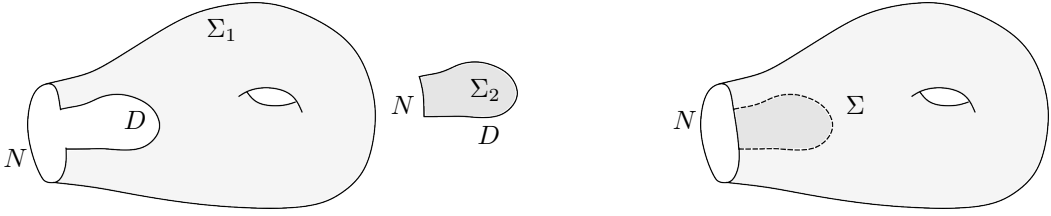

\subsection{Neumann double}\label{double_surface}
The \textbf{Neumann double} of a Riemann surface $\Sigma$ with smooth boundary is constructed by considering a copy $\bbar{\Sigma}$ of $\Sigma$ with opposite orientation and we denote by $\tau_\Sigma: \Sigma\to \bbar{\Sigma}$ the identity map. We equip $\bbar{\Sigma}$ with the complex structure given by the charts $(\bbar{\omega}_j\circ \tau^{-1}_\Sigma)_j$ if $\omega_j$ are the charts of $\Sigma$ and $\bbar{\omega}_j$ is the complex conjugate of $\omega_j$.
 The Neumann double surface is then the compact Riemann surface (with boundary if $\Sigma$ has Dirichlet boundary components, or closed otherwise)
$\Sigma^{\#2}:=\Sigma \sqcup\bbar{\Sigma}/\sim$, where $\sim$ means that boundary points in $\pl \Sigma_N$ are canonically identified, i.e. $\tau_\Sigma(x)=x$ if $x\in \pl \Sigma_N$.  The double  $\Sigma^{\#2}$ carries an involution, that we still denote by $\tau_\Sigma:\Sigma^{\#2}\to \Sigma^{\#2}$: a point $x\in \Sigma$ is  mapped to $\tau_\Sigma(x)\in \bbar{\Sigma}$ and $x\in \bbar{\Sigma}$ is mapped to $\tau_\Sigma^{-1}(x)\in \Sigma$. 
We consider the holomorphic charts $\omega_j$ on the interior of $\Sigma$ and $\bbar{\omega}_j\circ \tau_\Sigma^{-1}$ on the interior of $\bbar{\Sigma}=\tau_\Sigma(\Sigma)$ to induce charts on $\Sigma^{\#2}\setminus \pl \Sigma_N$.
Using the holomorphic charts $\omega_j:V_j\to \D^+$ covering boundary points in $\Sigma$ (on the Neumann boundary), and $\bbar{\omega}_j\circ \tau_\Sigma^{-1}$ for $\bbar{V}_j=\tau_\Sigma(V_j)\subset \bbar{\Sigma}$, one obtains a chart 
\[ \omega_j^{\#2}: (V_j\cup \bbar{V}_j)/\sim\,\,  \to \D, \quad \omega_j^{\#2}|_{V_j}:=\omega_j, \quad \omega_j^{\#2}|_{\bbar{V}_j}:=\bbar{\omega}_j\circ \tau_\Sigma^{-1}.\] 
The involution $\tau_\Sigma$ in this chart is given by 
\[ \tau_\Sigma(z)=\bar{z}\]
and $\tau_\Sigma$ is globally antiholomorphic. 

If $\Sigma$ equipped with parametrizations, then $\Sigma^{\#2}$ has induced parametrizations by just using the parametrization $\zeta_j^D$ of the Dirichlet boundary circles $\mc{C}_j^D$ in $\Sigma$ and their copy 
$\tau_\Sigma(\zeta_j^\ell)$ in $\bbar{\Sigma}$.\\
 
We can also double a Riemann surface  with corners along its Neumann boundary components $\pl \Sigma_N$ in the same way as above.
This can be done by setting $\Sigma^{\#2}:=(\Sigma \sqcup \bbar{\Sigma})/\sim$ where $\sim$ is an equivalence relation identifying a point $p\in \pl \Sigma_N$ to the same point $p\in \pl\bbar{\Sigma}_N$ on the second copy. 
The gluing of complex structure near the boundary circles $\mc{C}_j^N$ of $\pl_j\Sigma_N$ is done similarly as for the case with boundary. To define the complex structure near the glued half-circles $\mc{B}_j^N$ we can use the charts $\omega_j:V_j\to \D^+$ of $\omega_j:V_j\to \D^{++}$ covering a neighborhood of $\mc{B}_j^D$
 in $\Sigma$, with the Neumann boundary in $V_j$ corresponding to $\D^{++}\cap \R$ via the chart,
and their copy $\bbar{\omega}_j\circ \tau_\Sigma :\bbar{V}_j\to \D^{+-}$ (with $\D^{+-}:=\{z\in \D\,|\, {\rm Re}(z)\geq 0, {\rm Im}(z)\leq 0\}$) in $\bbar{\Sigma}$. These can be glued to a chart 
\[\omega_j^{\#2}: (V_j\cup \bbar{V}_j)/\sim\,\,  \to \D, \quad \omega_j^{\#2}|_{V_j}:=\omega_j, \quad \omega_j^{\#2}|_{\bbar{V}_j}:=\bbar{\omega}_j\circ \tau_\Sigma^{-1}\] 
and the involution in these charts is $z\mapsto \bar{z}$. 
We also see that the resulting surface is a Riemann surface with boundary (thus no corner), will all boundary components being marked with Dirichlet marking. The fact that the corners have angle $\pi/2$ is crucial for this property to hold. 
If $\zeta_j^\ell,\zeta_j^h$ denote the parametrization of the $j$-th Dirichlet loop and half-circles, 
the Neumann double has induced parametrization of its boundary by setting $\tau_\Sigma\circ \zeta_j^\ell$ for the reflected boundary $j$-th loop in $\bbar{\Sigma}$, and 
\[  e^{i\theta}\in \T^+ \mapsto \zeta_j^h(e^{i\theta}), \quad e^{i\theta}\in \T^- \mapsto \tau_\Sigma\circ \zeta_j^h(e^{-i\theta})\]
for the loop obtained as the gluing $\mc{B}_j\cup \tau_\Sigma(\mc{B}_j)$ of two half-circles (here $\T^-=\T\setminus \T^+$).


\subsection{Half-annular holomorphic charts}\label{sec:half_annular_holo_charts}

We have seen in \eqref{omega_j} that neighborhood of boundary circles $\mc{C}_j$ have holomorphic 
charts mapping to annuli $\A_\delta$ for some $\delta<1$. Here we claim that a similar result holds for boundary 
half-circles.
\begin{lemma}\label{model_half_annular}
Let $\Sigma$ be a Riemann surface with corners in the sense of Definition \ref{def:mfd_with_corners}. For each boundary half-circle $\mc{B}_j$ with D marking, 
there is a biholomorphism $\omega_j:V_j\to \A_\delta^+$ for some $\delta<1$ and $V_j$ a neighborhood of $\mc{B}_j$ in $\Sigma$, analytic up to the boundary in the sense that its extends holomorphically in each chart near $\mc{B}_j$ and near $\pl \Sigma_N\cap V_j$. We call $\omega_j$ a \textbf{half-annular holomorphic chart} near a boundary half-circle.
\end{lemma}
\begin{proof} Consider an annular neighborhood $V_j'$ of $\mc{B}_j$ and glue a half disk $D_j^+$ to $\mc{B}_j$ in a way that 
$\mc{D}_j^+:=\mc{B}_j\# D_j^+$ is equipped with a complex structure that extends smoothly that of $D_j^+$, and call the segment of $\pl \mc{D}_j^+$ extending $\pl \Sigma_N\cap V_j'$ the Neumann boundary of $\pl \mc{D}_j^+$. Next we double 
$\mc{D}_j^+$ along its Neumann boundary as in Section \ref{double_surface} and obtain a disk $\mc{D}_j$ with a complex structure and analytic boundary circle (the Dirichlet boundary). It comes equipped with an antiholomorphic involution $\tau_{\mc{D}_j}$ fixing pointwise the Neumann boundary $\pl \mc{D}_j^+\subset \mc{D}_j$. Consider a uniformizing map 
$\phi : \mc{D}_j\to \D$, which in turn is analytic up to $\pl \mc{D}_j$, then $\phi\circ \tau_{\mc{D}_j}\circ \phi^{-1}$ is an anti-holomorphic involution of $\D$, and its complex conjugate must then be a M\"obius transform preserving $\D$, i.e. $\phi\circ \tau_{\mc{D}_j}\circ \phi^{-1}(z)=\bbar{B(z)}$ for some M\"obius map $B(z)$ in $\D$ which satisfies $\bbar{B}(\bar{z})=B^{-1}(z)$. We check that $B$ must be of the form $B(z)=(az+b)/(-bz+a)$ with $\bar{b}=-b$ and $|a|^2-|b|^2=1$, and the set of fixed points of $\bbar{B}$ is given by the circle $|z+\bar{a}/\bar{b}|^2=|a/b|^2-1$, which is exactly the isometric circle $C_B$ of the M\"obius transform $B$. In particular, we can choose a M\"obius transform $H$ preserving $\D$ and mapping $C_B\cap \D$ to  $[-1,1]$ and we then have that $H\circ \phi \circ \tau_{\mc{D}_j}\circ \phi^{-1}\circ H^{-1}$ is the complex conjugation $z\mapsto \bar{z}$. The map $H\circ \phi$   thus extends  as a  biholomorphism $V_j\to \A_\delta^+$, with $V_j\subset V_j'$  an annular neighborhood   of $\mc{B}_j$    and for some $\delta<1$, and is analytic up to $\mc{B}_j\cup (V_j\cap \pl \Sigma_N)$ with respect to the complex structure of $\Sigma$. 
\end{proof}

\begin{remark}
Notice that by construction in the proof of Lemma \ref{model_half_annular}, the charts $\omega_j: V_j\to \A^+_\delta$ can be symmetrized to $\omega_j^{\#2}:V_j\cup \tau_\Sigma(V_j)\subset \Sigma^{\#2}\to \A_\delta$,  and the involution 
$\tau_\Sigma$ conjugated by $\omega_j^{\#2}$ becomes $z\mapsto \bar{z}$ in $\A_\delta$. 
\end{remark}

 \begin{figure}[h] 
\begin{tikzpicture}[scale=1]
\draw[very thick] (1,0) arc (0:180:1);
\draw[very thick] (2,0) arc (0:180:2);
\draw[very thick] (-2,0) -- (-1,0);
\draw[very thick] (1,0) -- (2,0);
\coordinate (Z1) at (1.5,0) ;
\coordinate (Z2) at (-1,0) ;
\coordinate (Z3) at (-0.8,1.3) ;
\coordinate (Z4) at (-0.7,-0.9) ;
 \draw (0,2) node[above]{$D$} ;
 \draw (-1.5,-0.5) node[above]{$N$} ;
 \draw (1.5,-0.5) node[above]{$N$} ;
  \draw (-2.3,-0.2) node[above]{$-1$} ;  
  \draw (-0.7,-0.2) node[above]{$-\delta$} ;
   \draw (2.2,-0.2) node[above]{$1$} ;  
  \draw (0.8,-0.2) node[above]{$\delta$} ;
\end{tikzpicture}
\caption{Half-annularholomoprhic chart $\mathbb{A}_\delta^+$ with its D/N boundary markings}\label{pic2}
\end{figure}
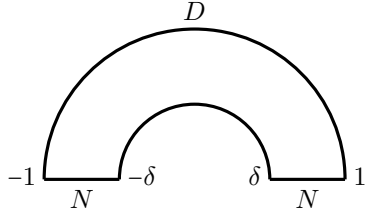

A variation of the proof above also shows the following existence of boundary half-circles (thus analytic) arbitrarily close to smooth curves with 
\begin{lemma}\label{boundarycut} 
Let $\zeta_0:\T^+ \to \mc{B}_0\subset  \Sigma$ be a smooth embedded curve with $\zeta_0(0)\in \pl \Sigma_N$ and 
$\zeta_0(\pi)\in \pl \Sigma_N$, which extends symmetrically with respect to $\tau_\Sigma$ on $\Sigma^{\#2}$ in a  smooth fashion. Then for each small neighbohood $V$ of $\mc{B}_0$ in $\Sigma$ there is an analytic embedded curve $\zeta:\T^+\to \mc{B}$ with $\zeta(0)\in \pl \Sigma_N$,
$\zeta(\pi)\in \pl \Sigma_N$ and $\mc{B}\subset V$, which extends symmetrically with respect to $\tau_\Sigma$ on $\Sigma^{\#2}$ in an analytic fashion.
\end{lemma}
\begin{proof}
Let $p_1,p_2$ be two points in $\pl \Sigma_N\notin \mc{K}$, which are not corner points, and let $\zeta_0:\T^+ \to \mc{B}_0\subset  \Sigma$ be a smooth embedded curve with $\zeta_0(0)=p_1$ and $\zeta_0(\pi)=p_2$, which extends symmetrically with respect to $\tau_\Sigma$ on $\Sigma^{\#2}$ in a  smooth fashion; we denote by $\zeta_0^{\#2}:\T \to \mc{B}_0^{\#2}$ the symmetric smooth parametrization. In a holomorphic charts $\omega_i:V(p_i)\to \D^+$ for $V(p_i)$ a small neighborhood of $p_i$ in $\Sigma$, the involution becomes $z\mapsto \bar{z}$. We also choose some metric $g_0$ compatible with the complex structure and which extends in a symmetric way to $\Sigma^{\#2}$ with respect to $\tau_\Sigma$.
For $\eps>0$ small, let $V_\eps$ be an $\eps-$tubular neighborhood of $\mc{B}_0^{\#2}$ with respect to $g_0$ on $\Sigma^{\#2}$: it has smooth boundary and is symmetric with respect to $\tau_\Sigma$, thus producing a smooth half-annular neighborhood of $\mc{B}_0$. 
We can then use the  same argument as in the proof of Lemma \ref{model_half_annular}: consider $V_\eps$ and glue a disk $\D$ to one of the boundary circles of the annulus $V_\eps$, and we equip the disk $\mc{D}_\eps :=V_\eps\cup \D$ with a smooth complex structure (or conformal structure) extending the complex structure of $V_\eps$. The proof of Lemma \ref{model_half_annular} then applies: it shows that there is a biholomorphism $\omega: V_\eps\to \A_{\delta}$ for some $\delta=\delta(\eps)<1$ with $\delta(\eps)\to 1$ when $\eps\to 0$, and 
$\omega \circ \tau_\Sigma \circ \omega^{-1}(z)=\bar{z}$. The only difference with Lemma \ref{model_half_annular} is that $\omega$ is not analytic but only smooth up to $\pl \mc{D}_\eps$ (it is of course analytic in $\mathring{\mc{D}}_\eps$). 
Now consider the half circle $B \subset \A^+_\delta$ with endpoints $\omega(p_1)\in (-1,-1+\delta)$ and 
$\omega(p_2)\in (1-\delta,1)$, and $\xi:\T^+\to B$ be  its canonical parametrization. It can be reflected smoothly to $\D$ by $z\mapsto \bar{z}$, thus $\mc{B}:=\omega^{-1}(B)$ can be reflected in a smooth way on $\Sigma$ symmetrically with respect to $\tau_\Sigma$. The parametrization $\omega^{-1}\circ \xi$ of $\mc{B}$ is analytic and $\mc{B}$ 
lies in an $\delta(\eps)$ neighborhood of $\zeta_0(\T^+)$ in $\Sigma$.
\end{proof}

\subsection{Admissible metrics}\label{sec:admissible_metrics}

We need to consider metrics $g$ on Riemann surfaces with corners that behave well under gluing in the sense that such metrics will form a smooth metric on the glued surface.

\begin{definition}[\textbf{Admissible and Neumann extendible metrics}]\label{admissibleN}
A smooth metric $g$ on a Riemann surface $\Sigma$ with corners will be said \textbf{admissible} if it  is compatible with the complex structure and if 
\begin{enumerate}
\item each boundary circle $\mc{C}_i$ with Dirichlet  marking has a neighborhood isometric to the flat annulus $(\A_\delta,|dz|^2/|z|^2)$ for some $\delta>0$,
\item each  boundary half-circle $\mc{B}_j$ with Dirichlet marking  has a neighborhood isometric to the flat half-annulus $(\A_\delta^+,|dz|^2/|z|^2)$ for some $\delta>0$.
\item the metric $g$ extends smoothly on the Neumann double $\Sigma^{\#2}$ in a way that it is symmetric under the involution $\tau_\Sigma:\Sigma^{\#2}\to \Sigma^{\#2}$.
\end{enumerate}
If the compatible metric $g$ satisfies only 3), we say that it is \textbf{Neumann extendible}. Note in particular that $\pl\Sigma_N$ must be geodesic for $g$.
A pair $(\Sigma,g)$ where $\Sigma$ is a Riemann surface with corners and $g$ is an admissible metric $g$ on $\Sigma$, will be called an admissible surface with corners.
\end{definition}

\begin{remark}\label{remark_metric}
The existence of Neumann extendible metric is clear as it suffices to consider any smooth metric $g$ on the Neumann double and set $\tilde{g}:=g+(\tau_\Sigma)^*g$, and then consider the restriction of this metric to $\Sigma$. 
For each corner $p_j\in \mc{K}$, using the map $z \in \A_\delta^+ \mapsto -\log(\bar{z})\in \H$ we see that there is a neighborhood $V_j$ such that an admissible metric $g$ on $V_j$  is isometric to $(\eps \D^{++},|dz|^2)$ for some $\eps>0$.
\end{remark}
Let us give a proof of the existence of admissible metrics. We first set it to be $g_\A:=|dz|^2/|z|^2$ in the model annulus charts $\omega_j: V_j\to \A_\delta$ near Dirichlet boundary circles $\mc{C}_j^D$, to be $g_\A= |dz|^2/|z|^2$  in the model charts near Dirichlet half-circles $\mc{B}_j^D$. Since $g_\A$ is symmetric under the map  $z\mapsto \bar{z}$ and smooth in $\A_\delta$, it remains to extend $g$ in an open set $\mc{O}$ of $\Sigma$ which does not intersect $\pl\Sigma_D$, in a way that it doubles in $\mc{O}\cup \tau_\Sigma(\mc{O})$ smoothly and symmetrically with respect to $\tau_\Sigma$. The main thing to check is the behaviour near $\pl \Sigma_N\cap \mc{O}$ and by assumption there are charts $\omega_j:V_j\to \D^+$ covering this region, and the doubled charts $\omega_j^{\#2}$ maps to $\D$. Now using a partition of unity which in the charts $\D$ is smooth and symmetric  under $z\mapsto \bar{z}$, we see that this becomes a local problem and one can simply take $g=e^{\rho(z)}|dz|^2$ in $\D$ with $\rho$ smooth and symmetric under $z\mapsto \bar{z}$. 

\subsection{Laplacian and spectral function}
On a closed Riemannian surface $(\Sigma,g)$, we consider the Laplacian $\Delta_{\Sigma,g}$. On a Riemannian surface $(\Sigma,g)$ with boundary or corners, we impose Dirichlet boundary condition on $\partial \Sigma_D$  and Neumann boundary condition on $\partial \Sigma_N$.  We denote $\Delta_{\Sigma,g,D}$, $\Delta_{\Sigma, g,N}$, $\Delta_{\Sigma,g,m}$ the self-adjoint extensions of Laplacian  with respectively Dirichlet (if $\partial \Sigma_N=\emptyset$), Neumann (if $\partial \Sigma_D=\emptyset$) and mixed boundary condition (otherwise). Notice that  all these operators are self-adjoint with discrete spectrum.  When there is no possible confusion about the underlying surface $\Sigma$ or the metric $g$, we shall sometime use the shortcut notations  $\Delta_{g}$, $\Delta_{g,D}$, $\Delta_{g,N}$, $\Delta_{g,m}$ and $\Delta_{\Sigma}$, $\Delta_{\Sigma,D}$, $\Delta_{\Sigma,N}$, $\Delta_{\Sigma,m}$ instead of $\Delta_{\Sigma,g},\Delta_{\Sigma,g,D}$, $\Delta_{\Sigma, g,N}$, $\Delta_{\Sigma,g,m}$, and we shall proceed similarly for the spectral zeta functions, the Green functions, the resolvent operators...  
  
The spectrum for the Laplacian $\Delta_{\Sigma, g,m}$ (resp. $\Delta_{\Sigma, g,N}$) on a Riemann surface $\Sigma$  with corners (resp. with boundary) and with Neumann extendible metric $g$ is in one-to-one correspondence with the spectrum of the Dirichlet Laplacian $\Delta_{\Sigma^{\#2},g^{\#2},D}$ (resp. the Laplacian $\Delta_{\Sigma^{\#2},g^{\#2}}$)  on the Neumann double  restricted to even functions with respect to $\tau_\Sigma$. The normalized eigenfunctions $e_j$ of $\Delta_{\Sigma,g,m}$ are exactly $\sqrt{2}$ times the restriction to $\Sigma$ of the normalized eigenfunctions $e_j^{{\rm even}}$ of $\Delta_{\Sigma^{\#2},g^{\#2},D}$ to $\Sigma$ that are even with respect to $\tau_\Sigma$. 
In particular, the heat kernel of $\Delta_{\Sigma,g,m}$ satisfies for $x,x'\in \Sigma$
\begin{equation}\label{heatkernel} 
e^{-t\Delta_{\Sigma,g,m}}(x,x')=e^{-t\Delta_{\Sigma^{\#2},g^{\#2},D}}(x,x')+e^{-t\Delta_{\Sigma^{\#2},g^{\#2},D}}(\tau_\Sigma x,x').
\end{equation}
The same holds when $\Delta_{\Sigma,g,N}$ replaces $\Delta_{\Sigma,g,m}$ and $\Delta_{\Sigma^{\#2},g^{\#2}}$ replaces $\Delta_{\Sigma^{\#2},g^{\#2},D}$.
The trace is then given for $t>0$ by
\[ {\rm Tr}(e^{-t\Delta_{\Sigma,g,m}})=\int_{\Sigma}(e^{-t\Delta_{\Sigma^{\#2},g^{\#2},D}}(x,x)+e^{-t\Delta_{\Sigma^{\#2},g^{\#2},D}}(\tau_\Sigma(x),x)){\rm dv}_g(x)\]
and it has an asymptotic expansion as $t\to 0$ given by Lemma \ref{asympt_exp_heat}. In particular, by the method of \cite{Ray-Singer},  the spectral zeta function of $\Delta_{\Sigma,g,*}$ for $*\in \{\emptyset,D,N,m\}$ is defined by
\[\zeta_{\Sigma,g,*}(s):={\rm Tr}(\Delta_{\Sigma,g,*}^{-s}(1-\Pi_{\ker \Delta_{\Sigma,g,*}}))\]
for ${\rm Re}(s)\gg 1$ where $\Pi_{\ker \Delta_{\Sigma,g,*}}$ is the orthogonal projection the $L^2$-kernel of $ \Delta_{\Sigma,g,*}$. The zeta function then extends meromorphically to $s\in \C$ and is analytic near $s=0$.
%

 Then we can define the determinant of Laplacian corresponding to closed Riemann surfaces and Riemann surfaces with boundary or corners, with a Neumann extendible metrics:  
 \[\begin{gathered}
 {\det}'(\Delta_{\Sigma,g})=\exp(-\pl_s\zeta_{\Sigma,g}(s)|_{s=0}),\quad  {\det}(\Delta_{\Sigma,g,D})=\exp(-\pl_s\zeta_{\Sigma,D}(s)|_{s=0}),\\
 \quad  {\det} '(\Delta_{\Sigma,g,N})=\exp(-\pl_s\zeta_{\Sigma,N}(s)|_{s=0}), \quad {\det} (\Delta_{\Sigma,g,m})=\exp(-\pl_s\zeta_{\Sigma,m}(s)|_{s=0}).
 \end{gathered}\]  
The index $'$ here emphasizes that the cases where the Laplacian has non trivial kernel.

For a Riemann surface $\Sigma$ with corners (resp. with boundary) and Neumann admissible metric $g$, as for Riemannian surfaces with boundary, let us denote by $\Delta_{\Sigma,g,D}$ the Laplacian with Dirichlet boundary condition on the whole boundary $\pl \Sigma$.  The spectrum for the Laplacian $\Delta_{\Sigma,g,D}$ is in one-to-one correspondence with the spectrum of the Dirichlet Laplacian $\Delta_{\Sigma^{\#2},g^{\#2},D}$ (resp. $\Delta_{\Sigma^{\#2},g^{\#2}}$)  on the Neumann double $(\Sigma^{\#2},g^{\#2})$ restricted to odd functions with respect to $\tau_\Sigma$. This directly gives the identity 
\begin{equation}\label{identity_det}
\det(\Delta_{\Sigma,g,m})\det(\Delta_{\Sigma,g,D})=\det(\Delta_{\Sigma^{\#2},g^{\#2},D}), \quad (\textrm{resp.} {\det}'(\Delta_{\Sigma,g,N})\det(\Delta_{\Sigma,g,D})={\det}'(\Delta_{\Sigma^{\#2},g^{\#2}})).
\end{equation}

\subsection{Green functions}
The Green function $G_{\Sigma,g}$ on a closed Riemann surface $\Sigma$ equipped with a compatible metric $g$ is defined 
to be the integral kernel of the resolvent operator $R_{\Sigma,g}:L^2(\Sigma,{\rm dv}_g)\to L^2(\Sigma,{\rm dv}_g)$ satisfying $\Delta_{\Sigma,g}R_{\Sigma,g}=2\pi ({\rm Id}-\Pi_{\ker \Delta_{\Sigma,g}})$, $R_{\Sigma,g}^*=R_{\Sigma,g}$ and $R_{\Sigma,g}1=0$. By integral kernel, we mean that for each $f\in L^2(\Sigma,{\rm dv}_g)$
\[ R_{\Sigma,g}f(x)=\int_{\Sigma} G_{\Sigma,g}(x,x')f(x'){\rm dv}_g(x').\] 
The Laplacian $\Delta_{\Sigma,g}$ has an orthonormal basis of real valued eigenfunctions $(e_j)_{j\in \N_0}$ in $L^2(\Sigma,{\rm dv}_g)$ with associated eigenvalues $\la_j\geq 0$; we set $\la_0=0$ and $e_0=({\rm v}_g(\Sigma))^{-1/2}$.  The Green function then admits the following Mercer's representation in $L^2(\Sigma\times\Sigma, {\rm dv}_g \otimes {\rm dv}_g)$
\begin{equation}
G_{\Sigma,g}(x,x')=2\pi \sum_{j\geq 1}\frac{1}{\lambda_j}e_j(x)e_j(x').
\end{equation}
Similarly, on a  compact Riemann surface with boundary equipped with a Neumann extendible metric, we will consider the Green function with Dirichlet boundary conditions $G_{\Sigma,g,D}$ associated to the Laplacian $\Delta_{\Sigma,g,D}$ or 
the Neumann Green function $G_{\Sigma,g,N}$   defined 
to be the integral kernel of the resolvent operator $\Delta_{\Sigma,g,N}R_{\Sigma,g,N}=2\pi ({\rm Id}-\Pi_{\ker_{\Delta,g,N}})$, $\partial_\nu R_{\Sigma,g,N}=0$, $R_{\Sigma,g,N}^*=R_{\Sigma,g,N}$ and $R_{\Sigma,g,N}1=0$, where $\Pi_{\ker_{\Delta,g,N}}$ is the orthogonal projection in $L^2(\Sigma,{\rm dv}_g)$ on $\ker \Delta_{\Sigma,g,N}$ (the constants).

Finally, on a Riemann surface with corners with Neumann extendible metric, we consider the Green function $G_{\Sigma,g,m}$ with mixed boundary conditions, of Dirichlet type on $\partial\Sigma_D$ and of Neumann type on $\partial\Sigma_N$. 
  
Just as for the heat operator in \eqref{heatkernel}, we can compare the Green function with Neumann or mixed condition on a Riemann surface with corners and Neumann extendible metric in terms of the Green function on the Neumann double $\Sigma^{\#2}$ (with doubled metric $g^{\#2}$):  this gives for all $x\not= x'$ in $\Sigma$
\begin{equation}\label{Green_symmetric} 
G_{\Sigma, g,m}(x,x')=G_{\Sigma^{\#2},g^{\#2},D}(x,x')+G_{\Sigma^{\#2},g^{\#2},D}(\tau_\Sigma(x),x')=G_{\Sigma^{\#2},g^{\#2},D}(x,x')+G_{\Sigma^{\#2},g^{\#2},D}(x,\tau_\Sigma(x')),
\end{equation}
and the same holds when $G_{\Sigma,g,N}$ replaces $G_{\Sigma,g,m}$ and $G_{\Sigma^{\#2},g^{\#2}}$ replaces $G_{\Sigma^{\#2},g^{\#2},D}$.
As for the Laplacian, we shall regularly use the shortcut notations  $G_{g}$, $G_{g,D}$, $G_{g,N}$, $G_{g,m}$ or  $G_{\Sigma}$, $G_{\Sigma,D}$, $G_{\Sigma,N}$, $G_{\Sigma,m}$ instead of $G_{\Sigma,g},G_{\Sigma,g,D}$, $G_{\Sigma, g,N}$, $G_{\Sigma,g,m}$, and similarly for the resolvent operator. 

\subsection{Gaussian Free Fields} 
On a closed Riemann surface $\Sigma$ with compatible metric $g$, the Gaussian Free Field (GFF in short) is defined as follows. Let $(a_j)_j$ be a sequence of i.i.d. real Gaussians   $\mc{N}(0,1)$ with mean $0$ and variance $1$, defined on some probability space   $(\Omega,\mc{F},\mathbb{P})$,  and define  the Gaussian Free Field with vanishing mean in the metric $g$ by the random functions
\begin{equation}\label{GFFong}
X_g:= \sqrt{2\pi}\sum_{j\geq 1}a_j\frac{e_j}{\sqrt{\la_j}} 
\end{equation}
 where  the sum converges almost surely in the Sobolev space  $H^{s}(\Sigma)$ for $s<0$ defined by
\[ H^{s}(\Sigma):=\{f=\sum_{j\geq 0}f_je_j\,|\, \|f\|_{H^s}^2:=|f_0|^2+\sum_{j\geq 1}\lambda_j^{s}|f_j|^2<+\infty\}.\] 
This Hilbert space is independent of $g$, only its norm depends on a choice of $g$.
The covariance of $X_g$  is then the Green function when viewed as a distribution, which we will write with a slight abuse of notation
\[\mathbb{E}[X_g(x)X_g(x')]= \,G_g(x,x').\]
In this case, we will denote the \textbf{Liouville field} by 
\[\phi_g:=c+X_g\] 
where $c\in\R$ is a constant that stands for the constant mode of the field.\\

We define in the same manner, on a  Riemann surface with boundary/corners and equipped with a Neumann extendible metric,  the GFF  with Dirichlet   boundary condition $X_{g,D}$ (with covariance $G_{g,D}$), the GFF  with Neumann    boundary condition $X_{g,N}$ (with covariance $G_{g,N}$) and the GFF  with mixed   boundary condition $X_{g,m}$ (with covariance $G_{g,m}$) 
\begin{equation}\label{GFFond}
X_{g,D}:=  \sqrt{2\pi}\sum_{j\geq 0}a_j\frac{e_{D,j}}{\sqrt{\lambda_{D,j}}} ,\quad    X_{g,N}:= \sqrt{2\pi}\sum_{j\geq 1}a_j\frac{e_{N,j}}{\sqrt{\lambda_{N,j}}}\quad \text{ and }\quad  X_{g,m}:= \sqrt{2\pi}\sum_{j\geq 0}a_j\frac{e_{m,j}}{\sqrt{\lambda_{m,j}}}
\end{equation}
where $(a_j)_j$ is a sequence of i.i.d. real Gaussians   $\mc{N}(0,1)$ and $(e_{D,j})_{j\geq 0}$, $(e_{N,j})_{j\geq 0}$, $(e_{m,j})_{j\geq 0}$ are orthonormal basis of the Laplacian with corresponding boundary conditions  and associated respective eigenvalues $(\lambda_{D,j})_{j\geq 0}$, $(\lambda_{N,j})_{j\geq 0}$, $(\lambda_{m,j})_{j\geq 0}$ (with   $\lambda_{N,0}=0$).  They converge almost surely  respectively in the Sobolev spaces    defined by (for all $s\in (-1/2,0)$) 
\begin{align*}
H^{s}_D(\Sigma,g):=&\{f=\sum_{j\geq 0}f_je_{D,j}\,|\, \|f\|_{H^s_D}^2:=\sum_{j\geq 0}\lambda_{D,j}^{s}|f_j|^2<+\infty\}
\\
H^{s}_N(\Sigma,g):=&\{f=\sum_{j\geq 0}f_je_{N,j}\,|\, \|f\|_{H^s_N}^2:=|f_0|^2+\sum_{j\geq 1}\lambda_{N,j}^{s}|f_j|^2<+\infty\}
\\
H^{s}_m(\Sigma,g):=&\{f=\sum_{j\geq 0}f_je_{m,j}\,|\, \|f\|_{H^s_m}^2:=\sum_{j\geq 0}\lambda_{m,j}^{s}|f_j|^2<+\infty\}.
\end{align*} 
In the case of Neumann b.c. (i.e. $\partial\Sigma_D=\emptyset$),  we will denote the Liouville field by $\phi_g:=c+X_{g,N}$ where $c\in\R$ is a constant that stands again for the constant mode of the field.
When $\partial\Sigma_D$ is non empty, $P \boldsymbol{\tilde \varphi}$ will stand for the harmonic extension to $\Sigma$ of a collection of boundary fields $\boldsymbol{\tilde \varphi}$ imposed on $\partial\Sigma_D$ with the condition $\partial_nP \boldsymbol{\tilde \varphi}=0$ on $\partial\Sigma_N$. The precise definition of these boundary fields will be given later (see Definition \ref{def_harmonicext}). For now, we just mention that the Liouville field will then be defined in the case  of Dirichlet boundary condition (resp. in the case of mixed boundary condition) by  
\begin{equation}\label{phi_g_dirichlet}
\phi_g:= X_{g,D}+P \boldsymbol{\tilde \varphi} \qquad (\textrm{resp. }\phi_g:= X_{g,m}+P \boldsymbol{\tilde \varphi}).
\end{equation} 
Hence $\phi_g$ will depend on boundary data in these cases. \\

As Gaussian Free Fields  are not functions but distributions, we will need to regularize them. We will call {\it g-regularization} of the GFF   its regularization along the geodesic circle of radius $\epsilon$ in the metric $g$ (we refer the reader to \cite[Section 3.2]{Guillarmou2019} and \cite{Wu1} for details). Let $x\in \Sigma$ and  $C_g(x,\epsilon):=\{z\in\Sigma\mid d_g(z,x)=\epsilon\}$ (for $x\in\partial\Sigma_N$ and $\epsilon$ small enough, $C_g(x,\epsilon)$ is an arc, with its two ends lying on $\partial\Sigma_N$). Let $(f_{x,\epsilon}^n)_{n\in \N}\in C^{\infty}(\Sigma)$ be a sequence with $||f_{x,\epsilon}^n||_{L^1}=1$ which is given by $f_{x,\eps}^n=\theta^n(d_g(x,\cdot)/\eps)$ where $\theta^n(r)\in C_c^\infty((0,2))$ is non-negative, 
supported near $r=1$ and such that $f^n_{x,\eps}{\rm dv}_g$ 
converges in $\mc{D}'(\Sigma)$ to the uniform probability measure 
$\mu_{x,\eps}$
on $C_g(x,\eps)$ as $n\to \infty$.  If the pairing $\langle h, f_{x,\eps}^n\rangle$ converges almost surely towards a random variable $h_\epsilon(x)$ that has a modification which is continuous in the parameters $(x,\epsilon)$, we will say that $h$ admits a $g$-regularization $(h_\epsilon)_\epsilon$. This is the case for the GFF $h=X_g$, $h=X_{g,D}$, $h=X_{g,N}$ and $h=X_{g,m}$, the proof is similar  to \cite[Lemma 3.2]{Guillarmou2019}.  
We will denote $X_{g,N,\epsilon}$, $X_{g,D,\epsilon}$  and $X_{g,m,\epsilon}$ the  respective $g$-regularizations of the various GFF introduced above
and $\phi_{g,\epsilon}$ the $g$-regularization of the Liouville field. 
An alternative way to define $X_{g,m}$ and $X_{g,m,\eps}$ when $g$ is Neumann extendible is to consider the Dirichlet GFF (or its regularization) on the doubled surface $\Sigma^{\#2}$ and take its even part with respect to the involution $\tau_\Sigma$:
\begin{equation}
X_{\Sigma,g,m}\stackrel{{\rm law}}{=}\frac{X_{\Sigma^{\#2},g^{\#2},D}+X_{\Sigma^{\#2},g^{\#2},D}\circ\tau_\Sigma}{\sqrt{2}},\qquad X_{\Sigma,g,N}\stackrel{{\rm law}}{=}\frac{X_{\Sigma^{\#2},g^{\#2}}+X_{\Sigma^{\#2},g^{\#2}}\circ\tau_\Sigma}{\sqrt{2}}.
\end{equation}

\subsection{Gaussian Multiplicative Chaos and vertex operators} 
For $\gamma\in\R$ and $h$ a random distribution admitting a  $g$-regularization $(h_\epsilon)_\epsilon$, we define the measures (the second definition makes sense only in the case of non empty Neumann boundary $\partial\Sigma_N$)
\begin{equation}\label{GMCg}
M^{g,\eps}_{\gamma}(h,\dd x):= \eps^{\frac{\gamma^2}{2}}e^{\gamma h_{ \eps}(x)}{\rm dv}_g(x),\quad \text{and }\quad M^{g,\eps}_{\gamma,\partial}(h,\dd s):= \eps^{\frac{\gamma^2}{4}}e^{\frac{\gamma}{2} h_{ \eps}(s)}{\rm d}\ell_g(s).
\end{equation}
Of particular interest for us is the case when $h=X_g,X_{g,D},X_{g,N}$ or $X_{g,m}$ (and consequently $h=\phi_g$ too). 

For the case with no boundary, when $\gamma\in (0,2)$, the random measure $M^{g,\eps}_\gamma(X_g,\dd x)$ converges  as $\eps\to 0$ in probability and weakly in the space of Radon measures towards a  non trivial random measure on $\Sigma$, see \cite[Section 3.2]{Guillarmou2019} and \cite{rhodes2014_gmcReview}. 

When $\Sigma$ has a non-empty boundary with Dirichlet marking, the same proof as in the closed case works, for the boundary  does not create any problem since the Green function vanishes at $\pl \Sigma$; this yields the convergence of 
$M^{g,\eps,D}_{\gamma}(h,\dd x)$ to some non-trivial random measure $M^{g,D}_{\gamma}(h,\dd x)$.

When $\Sigma$ has a non-empty boundary with Neumann marking and $g$ is Neumann extendible, 
on each set $\{ d_g(\cdot,\pl \Sigma)>\delta\}$ with $\delta>0$,  
the measure $M^{g,\eps}_{\gamma}(h,\dd x)$ converges in probability and weakly in the space of Radon measures towards a  non trivial random measure $M^g_{\gamma}(X_{g,N},\dd x)$ by the same argument as in the closed case. Moreover,
by \cite[Proposition 2.3]{huang2018}, one has $  M^g_{\gamma}(X_{g,N},\Sigma)<+\infty $ almost surely. Yet note that the expectation is finite iff $\gamma^2<2$. For the boundary measure, the fact that the measures $M^{g,\eps}_{\gamma,\partial}(X_{g,N},\dd s)$ or $M^{g,\eps}_\gamma(X_{g,m,\eps},\dd x)$ converge   as $\eps\to 0$ to some non-trivial random measure $M^{g}_{\gamma,\partial}(X_{g,N},\dd s)$ or $M^{g}_{\gamma,\partial}(X_{g,m},\dd s)$  follows from \cite[Proposition 2.1 and 2.2]{huang2018}. These measures give finite mass to the boundary of $\Sigma$. Indeed, the part $\partial\Sigma_D$ does not raise any problem as both fields $X_{g,N}$ or $X_{g,m}$ vanish there, so the only thing to check is that they have finite mass almost surely on $\pl \Sigma_N$, and actually the mass could only blow up at the corner points.  But the associated Green function on the Neumann double  has Dirichlet condition at $\pl \Sigma_D$ (thus vanishes at the corner points $\pl \Sigma_N\cap \pl\Sigma_D$).


We will occasionally use the notation $m(f):=\int f \dd m$ when $m$ is a measure and $f$ is $m$-integrable (hence the use of $ M^g_\gamma(X_{g,D},f)$, $M^{g}_{\gamma,\partial}(X_{g,N},f) $, ...). 

Using the notation \eqref{phi_g_dirichlet} and \eqref{GFFond} we have the following relations when $\Sigma$ satisfies
\begin{description}
\item[$\partial\Sigma=\emptyset$] $M^g_\gamma(\phi_g,\dd x)=e^{\gamma c}M^g_\gamma(X_{g},\dd x)$,
\item[$\partial\Sigma\not=\emptyset$ and $\partial\Sigma_D=\emptyset$] $M^g_\gamma(\phi_g,\dd x)=e^{\gamma c}M^g_\gamma(X_{g,N},\dd x)$ and $M^{g}_{\gamma,\partial}(\phi_g,\dd s)=e^{\frac{\gamma}{2}c}M^{g}_{\gamma,\partial}(X_{g,N},\dd s)$,
\item[$\partial\Sigma\not=\emptyset$ and $\partial\Sigma_N=\emptyset$] $M^g_\gamma(\phi_g,\dd x)=e^{\gamma P  \tilde{\boldsymbol{\varphi}} (x)}M^g_\gamma(X_{g,D},\dd x)$ ,
\item[$\partial\Sigma_D\not=\emptyset$ and $\partial\Sigma_N\not=\emptyset$] $M^g_\gamma(\phi_g,\dd x)=e^{\gamma P  \tilde{\boldsymbol{\varphi}} (x)}M^g_\gamma(X_{g,m},\dd x)$ and $M^{g}_{\gamma,\partial}(\phi_g,\dd s)=e^{\frac{\gamma}{2}P  \tilde{\boldsymbol{\varphi}} (x)}M^{g}_{\gamma,\partial}(X_{g,m},\dd s)$
\end{description}
 
Also, from  \cite[Lemma 3.2]{Guillarmou2019}  we recall that there exist a function $W_g  \in C^\infty(\Sigma)$, called Robin constant, such that 
\begin{equation}\label{varXg}
\lim_{\eps \to 0}\E[X^2_{g,\eps}(x)]+\log\eps=W_g(x).
\end{equation}
For a Riemann surface with boundary (with Neumann marking and Neumann extendible metric $g$) or with corners (and then with admissible metric $g$), we have the uniform estimate as $\eps\to 0$ for $x$ in compact sets of $\mathring{\Sigma}$
\begin{align*} 
& \E[X_{g,N,\eps}(x)^2]=\int_{C_g(x,\eps)}\int_{C_g(x,\eps)} G_{g,N}(y,y')d\ell_g(y)d\ell_g(y')=-\log\eps+W_{g,N}(x)+o(1)\\
& \E[X_{g,m,\eps}(x)^2]=\int_{C_g(x,\eps)}\int_{C_g(x,\eps)} G_{g,m}(y,y')d\ell_g(y)d\ell_g(y')=-\log\eps+W_{g,m}(x)+o(1)
\end{align*}
where $W_{g,N}(x):=W_{\Sigma^{\#2},g^{\#2}}(x)+ G_{\Sigma^{\#2},g^{\#2}}(x,\tau_\Sigma(x))$ and $W_{g,m}(x):=W_{\Sigma^{\#2},g^{\#2},D}(x)+G_{\Sigma^{\#2},g^{\#2},D}(x,\tau_\Sigma(x))$ by using again \eqref{Green_symmetric}.
On the boundary, we can again use the Neumann double to deduce that 
\begin{align}\label{boundvarYg}
\lim_{\eps \to 0}\E[X^2_{g,N,\eps}(s)]-2\log\eps=W_{g,N,\pl}(s),& &\lim_{\eps \to 0}\E[X^2_{g,m,\eps}(s)]+2\log\eps=&W_{g,m,\pl}(s)
\end{align}
uniformly over, respectively,  $\partial\Sigma$ and the compact subsets of  the interior of $\partial\Sigma_N$ (seen as a submanifold of $\partial\Sigma$), with $W_{g,N,\pl}(s):=2W_{\Sigma^{\#2},g^{\#2}}(s)$ and $W_{g,m,\pl}(s):=2W_{\Sigma^{\#2},g^{\#2},D}(s)$.\\

In the following, the behavior of these constants under conformal change of metrics will be important. Some simple rules can be obtained in the case when $\partial\Sigma_D\not=\emptyset$: if $g'=e^{\omega}g$ for some $\omega \in C^\infty(\Sigma)$  and $g$ is admissible (observe that the Dirichlet GFF or the mixed b.c. GFF remain the same in law for $g$ and $g'$, for the associated Green function are the same): 
\begin{align}\label{varYgconformal1}
& \text{ in }\mathring{\Sigma}: W_{g',D}(x)= W_{g,D}(x)+\frac{\omega(x)}{2} \quad \text{ and }\quad W_{g',m}(x)= W_{g,m}(x)+\frac{\omega(x)}{2} \\
\label{varYgconformal2}
& \text{ on }\partial\Sigma_N \text{ if non empty:}\,\,  W_{g',m,\pl }(s)= W_{g,m,\pl}(s)+ \omega(s) 
\end{align}
which leads directly to the scaling relations
\begin{align}\label{scalingmeasure1}
& M^{g'}_\gamma(X_{g',D},\dd x)= e^{\frac{\gamma}{2}Q \omega(x)} M^g_\gamma(X_{g,D},\dd x) \quad \text{ and }\quad M^{g'}_\gamma(X_{g',m},\dd x)= e^{\frac{\gamma}{2}Q \omega(x)} M^g_\gamma(X_{g,m},\dd x)\\
\label{scalingmeasure2}
& M^{g'}_{\gamma,\partial}(X_{g',m},\dd s)= e^{\frac{\gamma}{4}Q \omega(s)} M^g_{\gamma,\partial}(X_{g,m},\dd s)\quad \text{ on }\partial\Sigma_N.
\end{align}
And if $\partial \Sigma_D=\emptyset$, (recall $g$ is assumed Neumann extendible) we have 
\begin{align}\label{nodirichletgmc}
M_{\gamma}^{g'}(X_{g'},\dd x){=} e^{\gamma(\frac{Q}{2}\omega-m_{g'}(X_{g}))}M_{\gamma}^g(X_g,\dd x)
\end{align}
\[
M_{\gamma}^{g'}(X_{g',N},\dd x){=} e^{\gamma(\frac{Q}{2}\omega-m_{g'}(X_{g,N}))}M_{\gamma}^g(X_{g,N},\dd x), \quad
M^{g'}_{\gamma, \partial}(X_{g',N},\dd x){=} e^{\frac{\gamma}{2}(\frac{Q}{2}\varphi-m_{g'}(X_{g,N}))}M^{g}_{\gamma, \partial}(X_{g,N},\dd x) 
\]
where we use the notation $m_g(f):=\frac{1}{\rm{v}_g(\Sigma)}\int_{\Sigma}f(x)\dd {\rm v}_g(x)$.\\

\textbf{Vertex operators.} On a Riemann surface $\Sigma$ (closed, with boundary or with corner), we introduce the regularized vertex operators, for fixed $\alpha\in\R$ (called {\it weight}) and $x\in \mathring{\Sigma}$ (recall \eqref{phi_g_dirichlet} and $ \phi_{g,\epsilon}$ the $\eps$-regularization)
\begin{equation*}
V_{\alpha,g,\eps}(x)=\eps^{\alpha^2/2}  e^{\alpha  \phi_{g,\epsilon}(x) } .
\end{equation*}
When $\Sigma$ has a boundary and $\partial\Sigma_N\not=\emptyset$, we also define the boundary vertex operators, for fixed $\beta\in\R$ (called {\it weight}) and $s\in \partial\Sigma_N$,  
\begin{equation*}
V_{\frac{\beta}{2},g,\eps}(s)=\eps^{\frac{\beta^2}{4}}  e^{\frac{\beta}{2}  \phi_{g,\epsilon}(s) } .
\end{equation*}
Notice that in the case of a surface with $\partial\Sigma_D\not=\emptyset$  and if $g'=e^{\omega}g$, then the relations \eqref{varYgconformal1} or \eqref{varYgconformal2} give 
\begin{equation}\label{scalingvertex}
V_{\alpha,g',\eps}(x)=(1+o(1)) e^{\frac{\alpha^2}{4} w(x)} V_{\alpha,g,\eps}(x)\quad \text{ and }\quad V_{\beta,g',\eps}(s)=(1+o(1)) e^{\frac{\beta^2}{8} \omega(s)} V_{\beta,g,\eps}(s)
\end{equation}
when $\eps$ goes to $0$. We will use this fact later at the level of correlation functions.

\section{Background on Liouville CFT for closed Riemann surfaces}\label{sec:backclosed}
  
 In this section, we review the needed result on  Liouville CFT on closed Riemann surfaces as was established in \cite{DKRV16, Guillarmou2019}.
 
 \subsubsection{LCFT path integral}
  Let  $\Sigma$ be a closed Riemann surface of genus $g(\Sigma)$, $g$ be a fixed metric on $\Sigma$ with scalar curvature denoted by $K_g$,  let us also consider parameters
 \[\gamma\in(0,2),\quad  \mu>0 ,\quad Q=\frac{\gamma}{2}+\frac{2}{\gamma}.\] 
 Recall that the Liouville field is $\phi_g=c+X_g$ with $c\in\R$. For $F:  H^{s}(\Sigma)\to\R$ (with $s<0$) a bounded continuous functional, we set 
\begin{align}\label{def:pathintegralF}
\cjg F(\phi)\cjd_{\Sigma, g}:=& \big(\frac{{\rm v}_{g}(\Sigma)}{{\det}'(\Delta_{g})}\big)^\hf  \int_\R  \E\Big[ F( \phi_g) \exp\Big( -\frac{Q}{4\pi}\int_{\Sigma}K_{g}\phi_g\,{\rm dv}_{g} - \mu   M^g_\gamma(\phi_g,\Sigma)  \Big) \Big]\,\dd c . 
\end{align}
 By \cite[Proposition 4.1]{Guillarmou2019}, this quantity  defines a measure on $H^s(\Sigma)$ and moreover the partition function defined as  the total mass of this measure, i.e $\cjg 1 \cjd_ {\Sigma, g}$, is finite iff the genus $g(\Sigma) \geq 2$.

 Next, if the surface $\Sigma$ is closed, the correlation functions are defined by the limit
\begin{equation}\label{defcorrelg}
\cjg \prod_i V_{\alpha_i,g}(x_i) \cjd_ {\Sigma, g}:=\lim_{\eps \to 0} \: \cjg \prod_i V_{\alpha_i,g,\eps}(x_i) \cjd_ {\Sigma, g}
\end{equation}
where we have fixed $m$ distinct points $x_1,\dots,x_m$  on $\Sigma$ with respective associated weights $\alpha_1,\dots,\alpha_m\in\R$. Non triviality of correlation functions are then summarized in the following proposition (see \cite[Prop 4.4]{Guillarmou2019}):

 \begin{proposition}\label{limitcorel} Let $ x_1,\dots,x_m$ be distinct points on a closed surface $\Sigma$ and $ (\alpha_1,\dots,\alpha_m)\in\R^m$.
The limit \eqref{defcorrelg} exists and is non zero if and only if the weights $(\alpha_1,\dots,\alpha_m)$ satisfy the \textbf{Seiberg bounds}
 \begin{align}\label{seiberg1}
 & \sum_{i}\alpha_i + 2 Q (g(\Sigma)-1)>0,\\ 
 &\forall i,\quad \alpha_i<Q\label{seiberg2}.
 \end{align}
 \end{proposition}

The correlation function then obeys the following transformation laws\footnote{The reader may compare \eqref{confan} with the general axiomatic of CFTs exposed in \cite{Gawedzki96_CFT}.} (see \cite[Prop 4.6]{Guillarmou2019}):
 \begin{proposition}\label{covconf2}{\bf (Conformal anomaly and diffeomorphism invariance)}
Let $g,g'$ be two conformal metrics on a closed surface $\Sigma$ with $g'=e^{\omega}g$ for some $\omega\in C^\infty(\Sigma)$, and let $ x_1,\dots,x_m$ be distinct points on $\Sigma$ and $ (\alpha_1,\dots,\alpha_m)\in\R^m$ satisfying the Seiberg bounds \eqref{seiberg1} and \eqref{seiberg2}. Then the following holds true
\begin{equation}\label{confan} 
\log\frac{\cjg \prod_i V_{\alpha_i,g'}(x_i) \cjd_ {\Sigma, g'}}{\cjg \prod_i V_{\alpha_i,g}(x_i) \cjd_ {\Sigma, g}}= 
\frac{1+6Q^2}{96\pi}\int_{\Sigma}(|d\omega|_g^2+2K_g\omega) {\rm dv}_g-\sum_i\Delta_{\alpha_i}\omega(x_i)
\end{equation}
where the real numbers $\Delta_{\alpha_i}$, called {\it conformal weights}, are defined by the relation for $\alpha\in\R$
 \begin{align}\label{deltaalphadef}
\Delta_\alpha=\frac{\alpha}{2}(Q-\frac{\alpha}{2}).
\end{align} 
Let $\psi:\Sigma\to \Sigma$ be an orientation preserving diffeomorphism. Then  
\[ \cjg \prod_i V_{\alpha_i,\psi^*g}(x_i) \cjd_ {\Sigma, \psi^*g}=\cjg \prod_i V_{\alpha_i,g}(\psi(x_i)) \cjd_ {\Sigma, g}  .\]
\end{proposition}

\section{Liouville Conformal Field Theory on compact Riemann surfaces with Neumann boundary}\label{sec:open}
In this Section we recall the definition of correlation functions, and more generally the path integral, for Liouville conformal field theory on Riemann surfaces with Neumann boundary. We summarize here the relevant material, details can be found in 
\cite{huang2018,Wu1}. 

Let $\Sigma$ be a compact Riemann surface with boundary and $g$ be a compatible metric. Recall from last section that $g$ is  conformal to a Neumann extendible metric, but we do not need to assume that $g$ is Neumann extendible for this Section.
Correlation functions will always be considered in the case with $b$ boundary circles and Neumann boundary condition (i.e. $\partial\Sigma_D=\emptyset$). We will denote by $\partial_j\Sigma$  ($j=1,\dots,b$) the analytic boundary circles. We fix parameters $\gamma\in(0,2)$, $\mu\geq 0$ and $Q=\frac{\gamma}{2}+\frac{2}{\gamma}$. We consider a piecewise constant function $\mu_B:\partial\Sigma\to\R_+$.
We assume 
\begin{equation}\label{hyp:cosmo}
\text{either }\quad \mu>0\quad \text{ or }\quad \mu_B \quad \text{does not identically vanish.}
\end{equation}
 For $F:  H^{s}(\Sigma)\to\R$ (with $s<0$) a bounded continuous functional, we set (recall the definition of the Liouville field   $\phi_g=c+X_{g,N}$ in the case of Neumann b.c.)
\begin{align}\label{def:pathintegralFb}
&\cjg F(\phi)\cjd_{\Sigma, g}:=  \big(\frac{{\rm v}_{g}(\Sigma)}{{\det}'(\Delta_{g,N})}\big)^\hf
\\
&  \times \int_\R  \E\Big[ F( \phi_g) \exp\Big( -\frac{Q}{4\pi}\int_{\Sigma}K_{g}\phi_g\,{\rm dv}_{g}  -\frac{Q}{2\pi}\int_{\partial\Sigma}k_{g}\phi_g\,{\rm d}\ell_{g} - \mu   M^g_\gamma(\phi_g,\Sigma)-     M^g_{\gamma,\partial}(\phi_g,\mu_B\mathbf{1}_{\partial\Sigma})  \Big) \Big]\,\dd c . \nonumber
\end{align}
Here $k_g$ and ${\rm d}\ell_g$ are respectively the geodesic curvature and the Riemannian measure
  on $\pl \Sigma$. Again, this quantity  defines a measure on $H^{s}(\Sigma)$. Then we consider the limit
\begin{equation}\label{defcorrelgbound}
\cjg F(\phi_g)\prod_i V_{\alpha_i,g}(x_i)\prod_j V_{\frac{\beta_j}{2},g }(s_j)  \cjd_ {\Sigma, g}:=\lim_{\eps \to 0} \: \cjg F(\phi_g)\prod_i V_{\alpha_i,g,\eps}(x_i)\prod_j V_{\frac{\beta_j}{2},g,\eps}(s_j) \cjd_ {\Sigma, g}
\end{equation}
where we have fixed $m$   distinct bulk points $x_1,\dots,x_m$  on $\mathring{\Sigma}$ with respective associated weights $\alpha_1,\dots,\alpha_m\in\R$ and $m_B$ distinct boundary points $s_1,\dots,s_{m_B}$ on $\partial\Sigma$. Correlation functions are defined by taking $F=1$. Non triviality of correlation functions are then summarized in the following proposition (see \cite{Wu1}):
\begin{proposition}\label{limitcorelbound}
 Let $ x_1,\dots,x_m$ be distinct points  in $\mathring{\Sigma}$ and $s_1,\dots,s_{m_B}$ be distinct points  on $\partial\Sigma$ with $ (\alpha_1,\dots,\alpha_m,\beta_1\dots,\beta_{m_B})\in\R^{m+m_B}$.
The limit \eqref{defcorrelgbound} exists and is non zero if and only if the weights  obey the Seiberg bounds
 \begin{align}\label{seiberg1bound}
 & \sum_{i}\alpha_i + \sum_j \frac{\beta_j}{2}>Q\chi(\Sigma) ,\\ 
 &\forall i,\quad \alpha_i<Q\quad \quad \forall j,\quad \beta_j<Q\label{seiberg2bound}
 \end{align}
 where  $\chi(\Sigma)=2-2g(\Sigma)-b$ ($b$ is the number of boundary components).
 \end{proposition}
 Now, similarly to Proposition \ref{covconf2}, we have (\cite[Proposition 5.5 and 5.6]{Wu1}):
 \begin{proposition}\label{covconf3}{\bf (Conformal anomaly and diffeomorphism invariance)}
Let $g,g'$ be two compatible metrics on a compact Riemann surface $\Sigma$ with Neumann boundary, with $g'=e^{\omega}g$ for some $\omega\in C^\infty(\Sigma)$, and let $ x_1,\dots,x_m$ and $s_1,\dots,s_{m_B}$ be distinct points respectively on $\mathring{\Sigma}$ and $\partial\Sigma$ with respective weights  $ (\alpha_1,\dots,\alpha_m,\beta_1\dots,\beta_{m_B})\in\R^{m+m_B}$, and obeying the Seiberg bounds of Proposition \ref{limitcorelbound}. Then we have
\begin{align}\label{confan3} 
\log\left(\frac{\cjg \prod_i V_{\alpha_i,g'}(x_i)\prod_j V_{\frac{\beta_j}{2},g' }(s_j)  \cjd_ {\Sigma, g'}}{\cjg \prod_i V_{\alpha_i,g}(x_i)\prod_j V_{\frac{\beta_j}{2},g }(s_j)  \cjd_ {\Sigma, g}}\right)=& 
\frac{1+6Q^2}{96\pi}\Big(\int_{\Sigma}(|d\omega|_g^2+2K_g\omega) {\rm dv}_g+4\int_{\partial\Sigma}k_g\omega \dd\ell_g\Big)
\\
&-\sum_i\Delta_{\alpha_i}\omega(x_i)-\frac{1}{2}\sum_j\Delta_{\beta_j}\omega(s_j)
\end{align}
where the conformal weights $\Delta_{\alpha_i},\Delta_{\beta_j}$ are again given by \eqref{deltaalphadef}.\\
Let $\psi:\Sigma\to \Sigma'$ be an orientation preserving diffeomorphism. Then  
\[ \cjg \prod_i V_{\alpha_i,\psi^*g}(x_i) \prod_j V_{\frac{\beta_j}{2},\psi^*g }(s_j) \cjd_ {\Sigma, \psi^*g}=\cjg \prod_i V_{\alpha_i,g}(\psi(x_i))\prod_j V_{\frac{\beta_j}{2},g }(\psi(s_j))  \cjd_ {\Sigma', g}  .\]
\end{proposition}

\section{Amplitudes}\label{sec:decomp}
The main goal of this section is to introduce amplitudes. Amplitudes are path integrals defined as \eqref{def:pathintegralFb} but with a twist that the integrated variable (called field) in the path integrals are conditioned to taking prescribed values on the Dirichlet boundary $\pl \Sigma_D$ of the Riemann surface $\Sigma$ they are defined on. They represent integral kernels of operators on a Hilbert space given by an $L^2$ space over a measured space of fields on the unit circle $\T$ or the half circle $\T^+$. The operation of gluing Riemann surfaces along Dirichlet boundary components $\pl \Sigma_D$ then becomes a composition of the corresponding amplitudes when viewed as operators on the Hilbert space.  This will be instrumental in decomposing correlation functions into elementary pieces, producing structure constants of the theory and 
the so-called conformal blocks, which are universal functions over the moduli space of Riemann surfaces. Similarly, in the case of Riemann surfaces with corners, there are two types of gluing/cutting: bulk cuts (gluing along Dirichlet circles) and boundary cuts (gluing along Dirichlet half-circles), each of them being associated to a different Hilbert space. We start by describing these Hilbert spaces. Then we will introduce some material related to Dirichlet-to-Neumann operators (DN map for short) in order to define the amplitudes. Finally we describe the gluing of amplitudes.  In the case of closed Riemann surfaces, amplitudes were defined in \cite{GKRV21_Segal}, where it was also proven how they glue along bulk cuts. The novelty here is to treat the case of Riemann surfaces with a boundary/corners and the gluing along half-circles cuts, which intersect the Neumann boundary.

\subsection{Hilbert spaces of LCFT and boundary fields}\label{sub:hilbert} 
 
 The construction of the (bulk) Hilbert space of LCFT relies on the real-valued  Fourier series defined on  the unit circle $\T=\{z\in \C\,|\, |z|=1\}$. Such a generic Fourier series $\tilde\varphi^\ell$ ($\ell$ stands for loops) will be decomposed into its constant mode $c$ and orthogonal part $\varphi^\ell$ as
\begin{equation}\label{GFFcircle0}
\tilde\varphi^\ell=c+\varphi^\ell,\quad \text{and }\quad \forall \theta\in\R, \quad\varphi^\ell(e^{i\theta})=\sum_{n\not=0}\varphi^\ell_ne^{in\theta}
\end{equation}
with $(\varphi^\ell_n)_{n\not=0}$  its non-constant Fourier coefficients, which will be themselves parametrized by $\varphi_n^\ell=\frac{x^\ell_n+iy^\ell_n}{2\sqrt{n}}$  and $\varphi^\ell_{-n}=\frac{x^\ell_n-iy^\ell_n}{2\sqrt{n}}$  for $n>0$ with $x^\ell_n,y^\ell_n\in\R$. Next we consider the Sobolev space  $H^{s}(\T)$ defined as the set of distributions s.t.
\begin{equation}\label{outline:ws}
\|\tilde\varphi^\ell\|_{H^s(\T)}^2:=c^2+\sum_{n\not=0}|\varphi^\ell_n|^2(|n|+1)^{2s} <\infty.
\end{equation}
We will view the  series $\tilde{\varphi}^\ell:=c+\varphi^\ell$ as the coordinate map of the space $\R\times    \Omega_\T$, where 
\begin{align}\label{omegat}
  \Omega_\T=(\R^{2})^{\N^*}
\end{align}
  is equipped with the cylinder sigma-algebra $
 \Sigma_\T=\mathcal{B}(\R^2)^{\otimes \N^*}$ ($\mathcal{B}(E)$  stands for the Borel sigma-algebra on the topological space $E$) and the Gaussian  product measure 
 \begin{align}\label{Pdefin}
 \P_\T:=\bigotimes_{n\geq 1}\frac{1}{2\pi}e^{-\frac{1}{2}((x_n^\ell)^2+(y^\ell_n)^2)}\dd x^\ell_n\dd y^\ell_n.
\end{align}
Here $ \P_\T$ is supported on $H^s(\T)$ for any $s<0$ in the sense that $ \P_\T(\varphi^\ell\in H^s(\T))=1$. The bulk Hilbert space, denoted by $\mc{H}$, is then  
\[\mc{H}:=L^2(\R\times \Omega_{\T},\mu_0), \textrm{ with  }
\mu_0:=\dd c\otimes  \P_{\T}\] 
and Hermitian product denoted by $\langle\cdot,\cdot\rangle_{\mc{H}}$. We notice that, under $\P_\T$, the   series \eqref{GFFcircle0} is random (actually Gaussian) and arises naturally when considering the restriction of the whole plane GFF to the unit circle: the covariance of the random variable $\varphi^{\ell}(\theta)$ is 
\[ \E[ \varphi^{\ell}(\theta)\varphi^{\ell}(\theta')]=-\log|e^{i\theta}-e^{i\theta'}|.\]

The boundary Hilbert space has a somewhat similar structure. It relies on the space of series defined on the half unit circle $\T^+=\{e^{i\theta}\,|\, \theta\in [0,\pi]\}$. It will be more convenient, for technical reasons related to the Neumann doubles, to consider them as even trigonometric series on the circle $\T$ (below $h$ stands for half-circle)
\begin{equation}\label{GFFcircle+}
\tilde\varphi^h=c+\varphi^h,\quad \text{and }\quad \forall \theta\in\R, \quad   \varphi^h(e^{i\theta})= \sum_{n\not=0}\varphi^h_ne^{in\theta} 
\end{equation}
with $\varphi^h_{n}=\varphi^h_{-n}\in\R $ for $n>0$, and the $\varphi^h_{n}$'s are themselves parametrized by $\varphi^h_{n}=\frac{\sqrt{2}x^h_n}{2\sqrt{|n|}}$ where $x_n^h\in\R$.
Even trigonometric  series form a subspace of $H^s(\T)$ denoted by  $H^s_{\rm even}(\T)$. That is to say, all fields supported on Dirichlet boundary half-circles will be later considered as even fields on the doubled half-circle, which will be viewed as 
a Dirichlet boundary circle on the Neumann double $\Sigma^{\#2}$  that is even with respect to the natural involutive symmetry $\tau_\Sigma$.

We will view the  series $\tilde{\varphi}^h:=c+\varphi^h$ as the coordinate map of the space $\R\times    \Omega_{\T^+}$, where 
\begin{align}\label{omegat+}
  \Omega_{\T^+}=\R^{\N^*}
\end{align}
  is equipped with the cylinder sigma-algebra $ \Sigma_{\T^+}=\mathcal{B}(\R)^{\otimes \N^*}$   and the Gaussian  product measure 
 $$
 \mathbb{P}_{\mathbb{T}^{+}}:=\bigotimes_{n\geq 1} \frac{1}{\sqrt{2\pi}}e^{-\frac{(x^h_n)^2}{2}}\dd x^h_n .
 $$
The boundary Hilbert space is then $\mc{H}_+:=L^2(\R\times \Omega_{\T^+})$ with underlying measure 
\[\mu_0^+:=\dd c\otimes  \P_{\T^+}\] 
and Hermitian product denoted by $\langle\cdot,\cdot\rangle_{\mc{H}_+}$.

 In what follows, we will consider a family of such fields $\boldsymbol{\tilde \varphi}=(\boldsymbol{\tilde \varphi}^\ell,\boldsymbol{\tilde \varphi}^h) =(\tilde\varphi^\ell_1,\dots,\tilde\varphi^\ell_{b_\ell},\tilde\varphi^h_1,\dots,\tilde\varphi^h_{b_h})\in (H^{s}(\T))^{b_\ell}\times  (H^{s}_{\rm even}(\T))^{b_h} $, in which case the previous notations referring to the $j$-th field will be augmented with an index $j$, namely $c_j$, $\varphi^\ell_{j,n}$ or $\varphi^h_{j,n}$,  $x^\ell_{j,n}$ or $x^h_{j,n}$, $y^\ell_{j,n}$.  By an abuse of notations, we will also denote by $(\cdot,\cdot)_{2,\ell}$ the pairing between $(H^{s}(\T))^{b}$ and   $(H^{-s}(\T))^{b}$, namely
\[
(\boldsymbol{\tilde \varphi}^\ell,\boldsymbol{\tilde \phi}^\ell)_{2,\ell}:=\frac{1}{2\pi}\sum_{j=1}^b\int_0^{2\pi}  \tilde\varphi^\ell_j(e^{i\theta})\tilde\phi^\ell_j (e^{i\theta}) \dd \theta,
\]
  and by $(\cdot,\cdot)_{2,h}$ the pairing between $(H_{\rm even}^{s}(\T))^{b}$ and $(H_{\rm even}^{-s}(\T))^{b}$, namely
\[
(\boldsymbol{\tilde \varphi}^h,\boldsymbol{\tilde \phi}^h)_{2,h}:=\frac{1}{\pi}\sum_{j=1}^b\int_0^\pi  \tilde\varphi^h_j(e^{i\theta})\tilde\phi^h_j (e^{i\theta}) \dd\theta.
\]
Finally,   we will denote by $(\cdot,\cdot)_{2}$ the total  pairing
\begin{equation}\label{unpairing}
(\boldsymbol{\tilde \varphi},\boldsymbol{\tilde \phi})_{2}=  (\boldsymbol{\tilde \varphi}^\ell,\boldsymbol{\tilde \phi}^\ell)_{2,\ell}+\tfrac{1}{2}(\boldsymbol{\tilde \varphi}^h,\boldsymbol{\tilde \phi}^h)_{2,h}.
\end{equation}
  Such families of fields will serve as boundary values in the forthcoming definition of the amplitudes.

\subsection{Dirichlet-to-Neumann map}\label{subDTN}
Let  $\Sigma$ be a compact  Riemann surface with corners,  with  parametrized boundary $\partial\Sigma=\partial\Sigma_N\cup \partial\Sigma_D$, the parametrization being $\zeta_j^\ell$ and $\zeta_j^h$ for the D-marked circles and half-circles, and with 
\begin{equation}\label{boundary_Sigma}
\begin{gathered}
\partial\Sigma_N=\mc{C}^N\bigcup \mc{B}^N, \quad \mc{C}^N=\bigcup_{j=1}^{b^N_\ell}\caC_j^N, \quad 
\mc{B}^N:=\bigcup_{j=1}^{b^N_h}\caB_j^N\\
\partial\Sigma_D=\mc{C}^D\bigcup \mc{B}^D, \quad \mc{C}^D=\bigcup_{j=1}^{b^D_\ell}\caC_j^D, \quad \mc{B}^D=\bigcup_{j=1}^{b^D_h}\caB^D_j
\end{gathered}
\end{equation}
consisting of $b_\ell^N+b_\ell^D$ analytic  circles  and $b^N_h+b^D_h$ analytic half circles (here   $b^N_\ell$, $b^D_\ell$, $b^N_h$ or $b^D_h$ could possibly be equal to $0$). Recall that the  subscript $\ell$ stands for "loop" whereas the subscript "h" stands for "half-circle".\\

\textbf{Assumption on the metric $g$.}  In this section, we always consider a metric $g$ on $\Sigma$ that is Neumann extendible, but except when mentioned, $g$ is not assumed to be admissible.   

\begin{definition}[\textbf{Harmonic extension}]\label{def_harmonicext}
Given  a boundary field  $\boldsymbol{\tilde \varphi}=(\tilde\varphi_1^\ell,\dots,\tilde\varphi^\ell_{b^D_\ell},\tilde\varphi_1^h,\dots,\tilde\varphi^h_{b^D_h})=:(\boldsymbol{\tilde \varphi}^\ell,\boldsymbol{\tilde \varphi}^h )\in (H^{s}(\T))^{b^D_\ell}\times  (H^{s}_{\rm even}(\T))^{b^D_h}$ with $s\in\R$, we will write   $P\tilde{\boldsymbol{\varphi}}$ for  the harmonic extension  of $\tilde{\boldsymbol{\varphi}}$, that is $\Delta_g P\tilde{\boldsymbol{\varphi}}=0$ on $\mathring{\Sigma}$ with boundary values  $P\tilde{\boldsymbol{\varphi}}_{|\mc{C}^D_j}=\tilde\varphi_j^\ell\circ (\zeta_j^\ell)^{-1} $  for $j=1,\dots,b^D_\ell$, $P\tilde{\boldsymbol{\varphi}}_{|\mc{B}^D_j}=\tilde\varphi_j^h\circ (\zeta_j^h)^{-1} $  for $j=1,\dots,b^D_h$ and Neumann boundary condition   $\partial_\nu^gP \boldsymbol{\tilde \varphi}=0$ on $\partial\Sigma_N$. 
\end{definition} 
The boundary value has to be understood in the following weak sense: for all $u\in C^\infty(\T)$ and $v\in C^\infty(\T^+)$, if $\zeta^\ell_j,\zeta^h_j$ are the (analytic extensions of) the parametrizations of $\mc{C}^D_j, \mc{B}^D_j$
\begin{equation}\label{limit_weak}
\lim_{r\to 1^-}\int_{0}^{2\pi}P\tilde{\boldsymbol{\varphi}}(\zeta_j^\ell(re^{i\theta}))\bbar{u(e^{i\theta})}d\theta =2\pi (\tilde\varphi_j^\ell,u)_{2,\ell},\qquad \lim_{r\to 1^-}\int_{0}^{\pi}P\tilde{\boldsymbol{\varphi}}(\zeta_j^h(re^{i\theta}))\bbar{v(e^{i\theta})}d\theta =\pi (\tilde\varphi_j^\ell,v)_{2,h}.\end{equation}
 We shall sometimes add the $\Sigma$ subscript $P_\Sigma \tilde{\boldsymbol{\varphi}}$ to emphasize the space on which the harmonic extension lives.

The definition of the amplitudes will involve   the Dirichlet-to-Neumann operator (DN map for short) on a Riemann surface $\Sigma$ with corners. The DN map for $\Sigma$
\[\mathbf{D}_\Sigma:C^\infty(\T)^{b^D_\ell}\times C^\infty(\T^+)^{b^D_h}\to  C^\infty(\T)^{b^D_\ell}\times C^\infty(\T^+)^{b^D_h}\] is defined as follows: for 
 $\boldsymbol{\tilde \varphi}=(\tilde\varphi_1^\ell,\dots,\tilde\varphi^\ell_{b^D_\ell},\tilde\varphi_1^h,\dots,\tilde\varphi^h_{b^D_h})\in C^\infty(\T)^{b^D_\ell}\times C^\infty(\T^+)^{b^D_h}$ 
 \begin{align}\label{DNmap_def}
&\mathbf{D}_{\Sigma}\tilde{\boldsymbol{\varphi}}=(\mathbf{D}^\ell_{\Sigma}\tilde{\boldsymbol{\varphi}},\mathbf{D}^h_{\Sigma}\tilde{\boldsymbol{\varphi}}),
\\
&\mathbf{D}^\ell_{\Sigma}\tilde{\boldsymbol{\varphi}}:=(-\partial_{\nu } P\tilde{\boldsymbol{\varphi}}_{|\mc{C}^D_j}\circ\zeta^\ell_j)_{j=1,\dots,b^D_\ell},\nonumber
\\
&\mathbf{D}^h_{\Sigma}\tilde{\boldsymbol{\varphi}}:=(-\partial_{\nu } P\tilde{\boldsymbol{\varphi}}_{|\mc{B}^D_j}\circ\zeta^h_j)_{j=1,\dots,b^D_h}\nonumber
\end{align}
where $\nu$ denotes the inward unit  normal vector to $\mc{C}_j^D$ and to $\mc{B}_j^D$, and $ P\tilde{\boldsymbol{\varphi}}$ is the harmonic extension defined in Definition \ref{def_harmonicext}.
We shall see below in Lemma \ref{prop:DNrestriction} that $\mathbf{D}_\Sigma$ can indeed be viewed as an operator ($C^\infty_{\rm even}(\T)$ means smooth and invariant by $e^{i\theta}\mapsto e^{-i\theta}$)
\begin{equation}\label{DNnew}
\mathbf{D}_\Sigma: C^\infty(\T)^{b^D_\ell}\times C^\infty_{\rm even}(\T)^{b^D_h}\to C^\infty(\T)^{b^D_\ell}\times C^\infty_{\rm even}(\T)^{b^D_h}
\end{equation} 
and we will then consider it that way.
Note that, by the Green formula 
\begin{equation}\label{Greenformula}
\int_{\Sigma} |dP\tilde{\boldsymbol{\varphi}}|_g^2{\rm dv}_g = 2 \pi (  \boldsymbol{\tilde \varphi} ,\mathbf{D}_\Sigma \boldsymbol{\tilde \varphi} )_2 .
\end{equation}
By formula \eqref{Greenformula}, $\mathbf{D}_\Sigma$ is a non-negative symmetric operator, with kernel $\ker {\bf D}_{\Sigma}=\R \tilde{1}$ where $\tilde{1}= (1, \dots, 1)$ if $ \partial\Sigma_D\not=\emptyset$.

We will also consider the following variant of the DN map on $\Sigma$, a compact  Riemann surface with corners and with  parametrized boundary $\partial\Sigma=\partial\Sigma_N\cup \partial\Sigma_D$.  Here  $\partial\Sigma$ could possibly be empty. Next we need to specify two ways of cutting the surface:
\begin{definition}{\bf (Interior cut)}
We will say that a curve $\mc{C}'\subset \mathring{\Sigma}$ is an interior (or bulk) cut if it is an analytic parametrized embedded circle: ${\zeta'}^\ell:\T \to \mc{C}'$. In particular, its inverse $({\zeta'}^\ell)^{-1}$ extends holomorphically as a holomorphic chart $\omega'^\ell:V'^\ell\to\A_{\delta,\delta^{-1}}$ mapping some annular neighborhood $V'^\ell$ of $\mathcal{C}'$ to  the annulus $\A_{\delta,\delta^{-1}}$ for some $\delta<1$. \end{definition}

\begin{definition}{\bf (Boundary cut)}
We will say that a curve $\mc{B}'\subset \Sigma$ is a  boundary cut if it is an analytic parametrized embedded 
curve: ${\zeta'}^{h}:\T^+ \to \mc{B}'$ with endpoint ${\zeta'}^{h}(0),{\zeta'}^{h}(\pi)\in \pl \Sigma_N$ on the Neumann boundary, and such that ${\zeta'}^h$ extends analytically to ${\zeta'}^h:\T\to \mc{B}'\cup \tau_\Sigma(\mc{B}')\subset \Sigma^{\#2}$ in a way that ${\zeta'}^h(e^{-i\theta})=\tau_\Sigma({\zeta'}^h(e^{i\theta}))$.
In particular $({\zeta'}^h)^{-1}$ extends a holomorphic chart $\omega'^h:V'^h \to\A^+_{\delta,\delta^{-1}}$ mapping some half-annular neighborhood $V'^h$ of $\mathcal{B}'$ to  the half-annulus $\A^+_{\delta,\delta^{-1}}$  for some $\delta<1$.
\end{definition}

The existence of boundary cuts is proved in Lemma \ref{boundarycut}.
Notice that, if for such a cut, $\Sigma\setminus \mc{C}'$ (resp. $\Sigma\setminus \mc{B}'$) has two connected components, each component is a Riemann surface with corners in the sense of  Definition \ref{def:mfd_with_corners}, and the cuts become new Dirichlet boundary components (circles or half-circles).\\

\textbf{Additional assumption on $g$ near the cuts.} 
In the remaining part of this Section \ref{subDTN} we make the following assumption on the metric near the  cuts. We assume that we are given a collection of  non overlapping cuts made up of $b'_\ell$ interior cuts $\mc{C}':=\sqcup_{j=1}^{b'_\ell}\mathcal{C}'_j\subset \mathring{\Sigma}$, with parametrizations ${\zeta'_j}^\ell$ and annular  neighborhoods ${V_j'}^\ell$,  and $b'_h$ boundary cuts $\mc{B}':=\sqcup_{j=1}^{b'_h}\mathcal{B}'_j$ with parametrizations ${\zeta'_j}^h$ and half-annular neighborhoods ${V_j'}^h$. We choose a metric $g$ such that there exists $f^\ell_j\in C^\infty(V^\ell_j)$, $f_j'^\ell\in C^\infty(V_j'^\ell)$,   $f_j^h\in C^\infty(V_j^h)$ and $f_j'^h\in C^\infty(V_j'^h)$ with $f^\ell_j|_{\mc{C}^D_j}=0$, $f_j'^\ell|_{\mc{C}_j'}=0$,  $f_j^h|_{\mc{B}^D_j}=0$ and $f_j'^h|_{\mc{B}'_j}=0$ such that  
\begin{equation} \label{metric_near_bdry} 
(\omega_j^\ell)^*\Big(\frac{|dz|^2}{|z|^2}\Big)=e^{f_j^\ell}g, \qquad ({\omega'_j}^\ell)^*\Big(\frac{|dz|^2}{|z|^2}\Big)=e^{f'^\ell_j}g, , \qquad (\omega_j^h)^*\Big(\frac{|dz|^2}{|z|^2}\Big)=e^{f^h_j}g,\qquad ({\omega'_j}^h)^*\Big(\frac{|dz|^2}{|z|^2}\Big)=e^{f'^h_j}g
\end{equation}
(recall that $g$ extends smoothly by $g+\tau_\Sigma^*g$ on the Neumann double $\Sigma^{\#2}_N$ by assumption). Notice that the curves $\mc{C}_j^{D},\mc{C}'_j$ and $\mc{B}_j^D,\mc{B}'_j$ are geodesics of respective lengths $2\pi$ and $\pi$. In the special case where $f^\ell_j=f_j'^\ell=f_j^h=f_j'^h=0$, we say that $g$ is admissible near the cuts.

\begin{definition}[\textbf{Harmonic extension}]\label{def_harmoniccut}
Given  a boundary field  
\[\boldsymbol{\tilde \varphi}=(\tilde\varphi_1^\ell,\dots,\tilde\varphi^\ell_{b'_\ell},\tilde\varphi_1^h,\dots,\tilde\varphi^h_{b'_h})=:(\boldsymbol{\tilde \varphi}^\ell,\boldsymbol{\tilde \varphi}^h)\in (H^{s}(\T))^{b'_\ell}\times  (H^{s}_{\rm even}(\T))^{b'_h}\]  
with $s\in\R$, we will write   $P_{\mc{C}',\mc{B}'}\tilde{\boldsymbol{\varphi}}$ for  the harmonic extension  of $\tilde{\boldsymbol{\varphi}}$, that is $\Delta_g P_{\mc{C}',\mc{B}'}\tilde{\boldsymbol{\varphi}}=0$ on the interior of
$ \Sigma\setminus( \mc{C}'\cup \mc{B}')$,  with boundary value $0$ on $\partial\Sigma_D$, Neumann boundary condition $\partial_\nu P_{\mc{C}',\mc{B}'}\tilde{\boldsymbol{\varphi}}=0$ on $\partial\Sigma_N$,  and equal to $\tilde\varphi_j^\ell\circ \omega'^\ell_j $ on $\mc{C}'_j$ for $j=1,\cdots,b'_\ell$ and equal to $\tilde\varphi_j^h\circ \omega'^h_j $ on $\mc{B}'_j$ for $j=1,\cdots,b'_h$.
\end{definition}

 The DN map $\mathbf{D}_{\Sigma,\mc{C}',\mc{B}'}:C^\infty(\T)^{b'_\ell}\times C^\infty_{\rm even}(\T)^{b'_h}\to C^\infty(\T)^{b'_\ell}\times C^\infty(\T)^{b'_h}$ associated to $\mathcal{C}',\mc{B}'$ is then defined as the jump at $\mc{C}'$ and $\mc{B}' $ of the harmonic extension outside of the cuts:    for $\tilde{\boldsymbol{\varphi}}\in C^\infty(\T)^{b'_\ell}\times C^\infty_{\rm even}(\T)^{b'_h}$ 
\begin{align}\label{defDSigmaC}
&\mathbf{D}_{\Sigma,{\mc{C}',\mc{B'}}}\tilde{\boldsymbol{\varphi}}=(\mathbf{D}^\ell_{\Sigma,{\mc{C}',\mc{B}'}}\tilde{\boldsymbol{\varphi}},\mathbf{D}^h_{\Sigma,{\mc{C}',\mc{B}'}}\tilde{\boldsymbol{\varphi}}),
\\
&\mathbf{D}^\ell_{\Sigma,{\mc{C}',\mc{B}'}}\tilde{\boldsymbol{\varphi}}:=-((\partial_{\nu_-} P_{\mc{C}',\mc{B}'}\tilde{\boldsymbol{\varphi}})|_{\mc{C}'_j}\circ {\zeta'_j}^{\ell}+(\partial_{\nu_+} P_{\mc{C}',\mc{B}'}\tilde{\boldsymbol{\varphi}})|_{\mc{C}'_j}\circ {\zeta'_j}^{\ell})_{j=1,\dots,b'_\ell},\nonumber
\\
&\mathbf{D}^h_{\Sigma,{\mc{C}',\mc{B}'}}\tilde{\boldsymbol{\varphi}}:=-((\partial_{\nu_-} P_{\mc{C}',\mc{B}'}\tilde{\boldsymbol{\varphi}})|_{\mc{B}'_j}\circ {\zeta'_j}^{h}+(\partial_{\nu_+} P_{\mc{C}',\mc{B}'}\tilde{\boldsymbol{\varphi}})|_{\mc{B}'_j}\circ {\zeta'_j}^{h})_{j=1,\dots,b'_h}.\nonumber
\end{align}
where $\nu_\pm$ are the opposite two unit normal vector field at $\mc{C}'_j$ and $\mc{B}'_j$, 
pointing inside when approaching the curves from one side of the annular neighorbhoods ${V'_j}^{\ell}$ and ${V'_j}^{h}$.

We consider the double $\Sigma^{\#2}$ of $\Sigma$ with the doubled metric $g^{\#2}$ (which is smooth by assumption).
Now each curve $\mc{C}_j'$ or $\mc{B}'_j$ is then doubled to $\Sigma^{\#2}$, the half-circle becoming full geodesic circles $\mc{B}_j'\cup \tau_\Sigma(\mc{B}_j')$, and we denote $\mc{C}''=\cup_{j} (\mc{C}_j'\cup \tau_\Sigma(\mc{C}'_j))$ and $\mc{B}''=\cup_{j} (\mc{B}_j'\cup \tau_\Sigma(\mc{B}_j'))$.
The symmetry $\tau_\Sigma$ on the Neumann double $\Sigma^{\#2}$ induces a symmetry on 
$\pl(\Sigma^{\#2})$, and using the parametrization of the Dirichlet boundary circles in $\Sigma^{\#2}$  induced by those 
of $\Sigma$, the symmetry $\tau_\Sigma$ induces a symmetry on $\T^{2b_{\ell}^D}\times \T^{b_h^D}$, where the 2 sets of $b_\ell^D$ circles are the circles $\mc{C}_j^D$ and $\tau_{\Sigma}(\mc{C}_j^D)$ and the last $b_h^D$ circles are the glued half-circles $\mc{B}_j^D\# \tau_\Sigma(\mc{B}_j^D)$. This symmetry will still be denoted by $\tau_\Sigma$, by an abuse of notation; notice that $\tau_\Sigma(z)=\bar{z}$ on the $b_h^D$ copies of $\T$. For $\boldsymbol{\tilde \varphi}\in (H^{s}(\T))^{b'_\ell}\times  (H^{s}_{\rm even}(\T))^{b'_h}$, we define $\tilde{\boldsymbol{\varphi}}+\tau_\Sigma^*\tilde{\boldsymbol{\varphi}}:=(\boldsymbol{\tilde \varphi}^\ell,\boldsymbol{\tilde \varphi}^{h,\#},\boldsymbol{\tilde \varphi}^\ell\circ\tau_\Sigma)\in H^{s}(\T)^{2b'_\ell+b'_h}$, where $\boldsymbol{\tilde \varphi}^{h,\#}=\boldsymbol{\tilde \varphi}^h$ on $\mc{B}^D$ and $\boldsymbol{\tilde \varphi}^{h,\#}=\boldsymbol{\tilde \varphi}^h\circ\tau_\Sigma$ on $\tau_\Sigma(\mc{B}^D)$.\\  

Then we have the following:
\begin{lemma}\label{prop:DNrestriction}
Let $\Sigma^{\#2}$ be the Neumann double of $\Sigma$. Then  

\noindent 1) In the space $ (H^{s}(\T))^{b'_\ell}\times  (H^{s}_{\rm even}(\T))^{b'_h}$, we have 
\begin{equation}\label{DN_int_bdouble}
    \mathbf{D}_{\Sigma, \mc{C}',\mc{B}'}\tilde{\boldsymbol{\varphi}} =\mathbf{D}_{\Sigma^{\#2}, \mc{C}'',\mc{B}''}(\tilde{\boldsymbol{\varphi}}+\tau_\Sigma^*\tilde{\boldsymbol{\varphi}})|_{\mc{C}',\mc{B}'} \in (H^{s-1}(\T))^{b'_\ell}\times  (H^{s-1}_{\rm even}(\T))^{b'_h}
\end{equation}
2) In the space $(H^{s}(\T))^{b_\ell^D}\times  (H^{s}_{\rm even}(\T))^{b_h^D}$, we have 
\begin{equation}\label{DN_double}
    \mathbf{D}_{\Sigma}\tilde{\boldsymbol{\varphi}} =\mathbf{D}_{\Sigma^{\#2}}(\tilde{\boldsymbol{\varphi}}+\tau_\Sigma^*\tilde{\boldsymbol{\varphi}})|_{\pl \Sigma_D} \in (H^{s-1}(\T))^{b_\ell^D}\times  (H^{s-1}_{\rm even}(\T))^{b_h^D}
\end{equation}
\end{lemma}
\begin{proof}
It suffices to consider the case of $\tilde{\boldsymbol{\varphi}}$ smooth with even Fourier modes on the half circles. In particular, 
$\tilde{\boldsymbol{\varphi}}$ can be viewed as a smooth $\tau_\Sigma$-symmetric function on each boundary circle. We then deduce that $P_{\Sigma}\tilde{\boldsymbol{\varphi}}+\tau_\Sigma^*P_{\Sigma}\tilde{\boldsymbol{\varphi}}$ is the harmonic extension to $\Sigma^{\#2}$ of $\tilde{\boldsymbol{\varphi}}+\tau_\Sigma^*\tilde{\boldsymbol{\varphi}}$ and this implies the claim. For the regularity property, we use the fact that the DN map of $\Sigma^{\#2}$ maps $H^s(\pl \Sigma^{\#2})$ to $H^{s-1}(\pl \Sigma^{\#2})$ (see \cite[Section 4.2]{GKRV21_Segal}).
\end{proof}
From now on, we shall then always consider the DN maps as acting on $(H^{s}(\T))^{b_\ell^D}\times  (H^{s}_{\rm even}(\T))^{b_h^D}$ and $ (H^{s}(\T))^{b'_\ell}\times  (H^{s}_{\rm even}(\T))^{b'_h}$ (or $H^s$ replaced by $C^\infty$ if $s=\infty$), i.e. 
we view the fields on the half-circle $\T^+$ fields on $\T$ that are symmetric by the involution $e^{i\theta}\mapsto e^{-i\theta}$ (corresponding to $\tau_\Sigma$ on the Neumann double).

Next we express the relation between DN map and the Green function restricted on the cutting circle, which plays an important rule in the gluing lemma.
Recall that $G_{g,m}$ denotes the Green function on $(\Sigma,g)$ with Dirichlet condition at $\pl \Sigma_D$ and Neumann condition on $\partial\Sigma_N$. The following proposition holds:
\begin{lemma}\label{rest_DNmap}:
 1) If $\partial\Sigma_D\neq \emptyset$, then $\mathbf{D}_{\Sigma,{\mc{C}'},\mc{B}'}$ is invertible, and its integral kernel is given by
 \begin{align}
     \mathbf{D}_{\Sigma,{\mc{C}',\mc{B}'}}^{-1}(y,y')=\frac{1}{2\pi}G_{g,m}(y,y'), \quad y\not=y' \in \mc{C}'\cup\mc{B}'.
 \end{align}
2) If $\partial\Sigma_D=\emptyset$, then on the space $C_0^{\infty}(\T)^{b_\ell'}\times C_0^{\infty}(\T^+)^{b_h'}$, we have
 \begin{align}
     \mathbf{D}_{\Sigma,{\mc{C}',\mc{B}'}}G_{g,m}=2\pi {\rm Id}
 \end{align}
 where $C_{0}^{\infty}(\T)=\left\{\tilde{\varphi} \in C^{\infty}(\T)| \int_{\T}\tilde{\varphi}(e^{i\theta}) d \theta=0\right\}$ and $C_{{\rm even},0}^{\infty}(\T)=C_{0}^{\infty}(\T)\cap C_{{\rm even}}^{\infty}(\T)$.
\end{lemma}
\begin{proof} We consider the Neumann double $\Sigma^{\#2}$. If $\pl\Sigma_D\neq\emptyset$, we have by \cite[Eq. (4.8)]{GKRV21_Segal} that $\mathbf{D}_{\Sigma^{\#2}, \mc{C}'',\mc{B}''}$ is invertible with inverse having integral kernel given by the Dirichlet Green function $(2\pi)^{-1}G_{\Sigma^{\#2},g,D}$ on the double surface.
We can then use Lemma \ref{prop:DNrestriction} to deduce 
\begin{align*}
&  \mathbf{D}_{\Sigma,{\mc{C}',\mc{B}'}}\Big( [\mathbf{D}^{-1}_{\Sigma^{\#2},\mc{C}'',\mc{B}''})(\tilde{\boldsymbol{\varphi}}+\tau_\Sigma^*\tilde{\boldsymbol{\varphi}})]|_{\mc{C'},\mc{B}'})\\
& = \mathbf{D}_{\Sigma^{\#2},\mc{C}'',\mc{B}''}\Big(( [\mathbf{D}^{-1}_{\Sigma^{\#2},\mc{C}'',\mc{B}''})(\tilde{\boldsymbol{\varphi}}+\tau_\Sigma^*\tilde{\boldsymbol{\varphi}})]|_{\mc{C'},\mc{B}'})+\tau_{\Sigma}^*[\mathbf{D}^{-1}_{\Sigma^{\#2},\mc{C}'',\mc{B}''})(\tilde{\boldsymbol{\varphi}}+\tau_\Sigma^*\tilde{\boldsymbol{\varphi}})]|_{\mc{C'},\mc{B}'}\Big)\Big|_{\mc{C}',\mc{B}'}\\
& = \mathbf{D}_{\Sigma^{\#2},\mc{C}'',\mc{B}''} \mathbf{D}^{-1}_{\Sigma^{\#2},\mc{C}'',\mc{B}''}(\tilde{\boldsymbol{\varphi}}+\tau_\Sigma^*\tilde{\boldsymbol{\varphi}})|_{\mc{C}',\mc{B}'}=\tilde{\boldsymbol{\varphi}}.
\end{align*}
Thus  $\mathbf{D}_{\Sigma,{\mc{C}',\mc{B}'}}$ is invertible and its integral kernel is for $y,y'\in \mc{C}'\cup \mc{B}'$
\[\mathbf{D}_{\Sigma,{\mc{C}',\mc{B}'}}^{-1}(y,y')=(2\pi)^{-1}(G_{\Sigma^{\#2},g,m}(y,y')+G_{\Sigma^{\#2},g,m}(y,\tau_\Sigma(y')))=(2\pi)^{-1}G_{g,m}(y,y')\]
where we used \eqref{Green_symmetric}.
The same argument applies for 2) as well.
\end{proof}

For $\boldsymbol{\tilde \varphi}=(\tilde\varphi_1^\ell,\dots,\tilde\varphi^\ell_{b^D_\ell},\tilde\varphi_1^h,\dots,\tilde\varphi^h_{b^D_h})=:(\boldsymbol{\tilde \varphi}^\ell,\boldsymbol{\tilde \varphi}^h )\in (H^{s}(\T))^{b^D_\ell}\times  (H^{s}_{\rm even}(\T))^{b^D_h}$ real valued with $s\in\R$, we will write 
\[ \tilde{\varphi}_j^{\ell}(\theta)=\sum_{n\in \Z} \varphi^\ell_{j,n}e^{in\theta}, \quad \tilde{\varphi}_j^{h}(\theta)=\sum_{n\in \Z} \varphi_{j,n}^he^{in\theta}.\]
  Let us further introduce  the  operator $\mathbf{D}=(\mathbf{D}^\ell,\mathbf{D}^h)$  by 
  \begin{align}
& \forall \tilde{\boldsymbol{\varphi}}^\ell \in C^\infty(\T;\R)^{b^D_\ell}, \quad
\mathbf{D}^\ell\tilde{\boldsymbol{\varphi}}^\ell=  \sum_{n\in \Z} |n| \varphi^\ell_{j,n}e^{in\theta},  \label{defmathbfD}\\
& \forall \tilde{\boldsymbol{\varphi}}^h\in C_{\rm even}^\infty(\T;\R)^{b^D_h}, \quad 
\mathbf{D}^h\tilde{\boldsymbol{\varphi}}^h= \sum_{n\in \Z} |n| \varphi_{j,n}^he^{in\theta}. 
\label{defmathbfDhalf}
\end{align}
Finally we consider the operators on $C^\infty(\T)^{b^D_\ell}\times C_{\rm even}^\infty(\T)^{b^D_h}$ 
\begin{align}
\label{tildeD} 
& \widetilde{\mathbf{D}}_{\Sigma} :=\mathbf{D}_{\Sigma}-\mathbf{D},
\\
&  \Pi_0(\tilde{\boldsymbol{\varphi}}):=  (( \tilde\varphi_1^\ell,1)_{2,\ell},\dots,( \tilde\varphi_{b^D_\ell}^\ell,1)_{2,\ell},( \tilde\varphi_1^h,1)_{2,h},\dots,( \tilde\varphi^h_{b^D_h},1)_{2,h}) \label{pi0}
\end{align}
and on $ C^\infty(\T)^{b'_\ell}\times C_{\rm even}^\infty(\T)^{b'_h}$ 
\begin{align}
& \widetilde{\mathbf{D}}_{\Sigma,{\mc{C}',\mc{B}'}} :=\mathbf{D}_{\Sigma,{\mc{C}',\mc{B}'}}-2\mathbf{D},\\
&  \Pi'_0(\tilde{\boldsymbol{\varphi}}):=  (( \tilde\varphi_1^\ell,1)_{2,\ell},\dots,( \tilde\varphi_{b'_\ell}^\ell,1)_{2,\ell},( \tilde\varphi_1^h,1)_{2,h},\dots,(\tilde\varphi^h_{b'_h},1 )_{2,h}).\label{Pi0'}
\end{align}
When $\mc{C}'=\emptyset$ (resp. $\mc{B}'=\emptyset$), we shall use the notation $\widetilde{\mathbf{D}}_{\Sigma,\mc{C}'}$ (resp.$\widetilde{\mathbf{D}}_{\Sigma,\mc{B}'}$) for $\widetilde{\mathbf{D}}_{\Sigma,\mc{C}',\mc{B}'}$.

By \eqref{DN_double} and using that the operator $\mathbf{D}$ acts diagonally on  $C^\infty(\T)^{b^D_\ell}\times C_{\rm even}^\infty(\T)^{b^D_h}$, we see that $\widetilde{\mathbf{D}}_{\Sigma}$ preserves $C^\infty(\T)^{b^D_\ell}\times C_{\rm even}^\infty(\T)^{b^D_h}$ and similarly for $\widetilde{\mathbf{D}}_{\Sigma,{\mc{C}',\mc{B}'}}$.

\begin{lemma}\label{lemmaDSigma-D}
The operators $\widetilde{\mathbf{D}}_{\Sigma}$, $\mathbf{D}_{\Sigma}(\Pi_0+{\bf D})^{-1}-{\rm Id}$ are smoothing operators in the sense that they are operators with smooth Schwartz kernel that are bounded for all $s,s'\in \R$ as maps 
 \begin{align*}
    &(H^{s}(\T))^{b_\ell^D}\times(H^{s}_{\rm even}(\T))^{b_h^D} \to (H^{s}(\T))^{b_\ell^D}\times(H^{s'}_{\rm even}(\T))^{b_h^D}.
 \end{align*}
The operators $\widetilde{\mathbf{D}}_{\Sigma,{\mc{C}'},\mc{B}'}$,  
 $\mathbf{D}_{\Sigma,\mc{C}',\mc{B}'}(2\Pi_0'+2{\bf D})^{-1}-{\rm Id}$,
 are smoothing operators in the sense that they are operators with smooth Schwartz kernel that are bounded for all $s,s'\in \R$ as maps 
 \begin{align*}
    &(H^{s}(\T))^{b'_\ell}\times(H^{s}_{\rm even}(\T))^{b'_h} \to (H^{s'}(\T))^{b'_\ell}\times(H^{s'}_{\rm even}(\T))^{b'_h}. 
 \end{align*}
In particular, all these operators are trace class on $L^2(\T)^{b_\ell}\times L^2_{\rm even}(\T)^{b_h}$, $L^2(\T)^{b_\ell'}\times L^2_{\rm even}(\T)^{b_h'}$ and the Fredholm determinant  $\det_{\rm Fr}(\mathbf{D}_{\Sigma,\mc{C}',\mc{B}'}(2\Pi_0'+2{\bf D})^{-1})$ is well-defined.
\end{lemma}
\begin{proof} Using Proposition \ref{prop:DNrestriction}, we shall reduce to the case of a manifold without corners by considering the Neumann double $\Sigma^{\#2}$ by gluing the Neumann boundary circles and half-circles of $\Sigma$ with the corresponding circles/half-circles of $\bbar{\Sigma}$. In \cite[Lemma 4.1]{GKRV21_Segal}, it is proved that $\mathbf{D}_{\Sigma^{\#2}}-
\mathbf{D}$  is an operator with a smooth integral kernel $K(x,x')$ acting   on $L^2(\T)^{2b^D_\ell+b_h^D}$ where $\mathbf{D}$ is the Fourier multiplier by $\sqrt{\Delta_\T}$ acting on $L^2(\T)^{2b^D_\ell+b_h^D}$ as \eqref{defmathbfD}. By \eqref{DN_double} and using that the operator 
$\mathbf{D}$ acts diagonally on  $C^\infty(\T)^{2b^D_\ell+b_h^D}$  and that it preserves even functions with respect to $z\mapsto \bar{z}$ on $\T$, for $x\in\pl \Sigma_D$ and $\tilde{\boldsymbol{\varphi}}\in C^\infty(\T)^{b^D_\ell}\times C^\infty_{\rm even}(\T)^{b_h^D}$
\begin{equation}\label{DSigma-Ddouble}
\begin{split} 
(\mathbf{D}_{\Sigma}-\mathbf{D})\tilde{\boldsymbol{\varphi}}(x)& =((\mathbf{D}_{\Sigma^{\#2}}-\mathbf{D})(\tilde{\boldsymbol{\varphi}}+\tau_\Sigma^*\tilde{\boldsymbol{\varphi}}))(x)\\
& =\int_{\T^{2b^D_\ell}\times \T^{b_h^D}} K(x,x')(\tilde{\boldsymbol{\varphi}}(x')+\tilde{\boldsymbol{\varphi}}(\tau_{\Sigma}(x'))\dd\ell_g(x')\\
&= \int_{\T^{b_\ell^D}\times \T^{b^D_h}}(K(x,x')+K(\tau_\Sigma(x),x'))\tilde{\boldsymbol{\varphi}}(x')\dd\ell_g(x').\end{split}
\end{equation}
This shows that $\tilde{\mathbf{D}}_{\Sigma}$ is smoothing since $K(x,x')+K(\tau_\Sigma(x),x')$ is a smooth function of $x,x'$. The same argument applies to 
$\widetilde{\mathbf{D}}_{\Sigma,{\mc{C}'},\mc{B}'}$. To deal with 
$\mathbf{D}_{\Sigma}(\Pi_0+{\bf D})^{-1}-{\rm Id}$, we simply write it as 
\[ \mathbf{D}_{\Sigma}(\Pi_0+{\bf D})^{-1}-{\rm Id}=(\tilde{\mathbf{D}}_{\Sigma}-\Pi_0)(\Pi_0+{\bf D})^{-1}\]
with $(\Pi_0+{\bf D})^{-1}:C^{-\infty}(\T^{b_\ell^D+b_h^D})\to C^{-\infty}(\T^{b_\ell^D+b_h^D})$ continuous while 
$(\tilde{\mathbf{D}}_{\Sigma}-\Pi_0): C^{-\infty}(\T^{b_\ell^D+b_h^D})\to C^{\infty}(\T^{b_\ell^D+b_h^D})$ is continuous, thus 
the composition is smoothing ($C^{-\infty}$ denotes the space of distributions). The same argument applies to $\mathbf{D}_{\Sigma,\mc{C}',\mc{B}'}(2\Pi_0'+2{\bf D})^{-1}-{\rm Id}$.
\end{proof}
\subsection{Amplitudes}

Let  $\Sigma$ be a Riemann surface with corners, with non-empty topological boundary given by \eqref{boundary_Sigma}  (here  $b^N_\ell$ or $b^D_\ell$ could possibly be equal to $0$), and $b^N_h+b^D_h$ analytic half circles (recall  that $b^N_h=b^D_h$ but they could be both equal to $0$).  All of these  curves do not intersect each other (except of course that each curve of the type $\caB_j^N$ is tied at its ends to the end of a curve of the type  $\caB^D_j$ and vice versa).\\

\textbf{Assumption on the metric $g$.}  In this section, we always consider a metric $g$ on $\Sigma$ that is Neumann extendible. Notice that  $g$ is always conformally equivalent to an admissible metric $g_0$, namely of the form $g={e^{\omega}}g_0$ for some smooth function $\omega:\Sigma\to\R$, and we will require in addition  that $\omega|_{\pl \Sigma_D}=0$.\\

We consider a piecewise constant function $\mu_B:\partial\Sigma_N\to\R_+$ and the condition
\begin{equation}\label{hyp:cosmoamp}
\text{either }\quad \mu>0\quad \text{ or }\quad \mu_B \quad \text{does not identically vanish.}
\end{equation}
Given marked points ${\bf x}:=(x_1,\dots,x_m)$ in  the interior $\Sigma^\circ$ of the surface with weights $\boldsymbol{\alpha}:=(\alpha_1,\dots,\alpha_m)\in\R^m$, and marked points ${\bf s}=(s_1,\dots,s_{m_B})$ on $\partial\Sigma_N$ with  weights $\boldsymbol{\beta}:=(\beta_1,...,\beta_{m_B})\in\R^{m_B}$, the  amplitude $\caA_{\Sigma,g,{\bf x},\boldsymbol{\alpha},{\bf s},\boldsymbol{\beta}, \boldsymbol{\zeta}}$ is defined as follows: 

 \begin{definition}{\textbf{Amplitudes of surfaces with boundary/corners.}}\label{def:amp}
 
 We suppose that the second Seiberg bound \eqref{seiberg2bound} holds, i.e. $\alpha_i<Q$, $i=1,\dots,m$ and $\beta_i<Q$, $i=1,\dots,m_B$.   
 
\noindent {\bf (A)}
Let $\partial\Sigma_D=\emptyset$. Assume \eqref{hyp:cosmoamp} and \eqref{seiberg1bound}.  For $F$ continuous  nonnegative function on $H^{s}(\Sigma)$ for some $s<0$ we define 
\begin{equation}\label{defampzerobound} 
\caA_{\Sigma,g,{\bf x},\boldsymbol{\alpha},{\bf s},\boldsymbol{\beta}}(F):=\lim_{\epsilon\to 0}\langle F(\phi_g)\prod_{i=1}^m V_{\alpha_i,g,\epsilon}(x_i)\prod_{i=1}^{m_B} V_{\frac{\beta_i}{2},g,\epsilon}(s_i) \rangle_{\Sigma,g} .
\end{equation}
using \eqref{defcorrelgbound}.
 If $F=1$ then the amplitude is just the LCFT correlation function  for  surfaces with boundary and will be simply denoted by $\caA_{\Sigma,g,{\bf x},\boldsymbol{\alpha},{\bf s},\boldsymbol{\beta}}$.
 \vskip 3mm
 
\noindent {\bf (B)}  If  $\partial\Sigma_D$ has $b^D_\ell+b^D_h>0$ boundary  components, we assume $\mu,\mu_B\geq 0$. The amplitude $\caA_{\Sigma,g,{\bf x},\boldsymbol{\alpha},{\bf s},\boldsymbol{\beta},\boldsymbol{\zeta}}$ is
  a function $(F,\tilde{\boldsymbol{\varphi}})\mapsto \caA_{\Sigma,g,{\bf x},\boldsymbol{\alpha},{\bf s},\boldsymbol{\beta},\boldsymbol{\zeta} }(F,\tilde{\boldsymbol{\varphi}})$ of  the boundary fields $\tilde{\boldsymbol{\varphi}}:=(\boldsymbol{\tilde \varphi}^\ell,\boldsymbol{\tilde \varphi}^h)\in (H^{s}(\T))^{b_\ell^D}\times(H^{s}_{\rm even}(\T))^{b_h^D}$ with $s<0$  and of continuous nonnegative functions $F$ defined on  $H^{s}(\Sigma)$ for $s\in (-1/2,0)$.  
 It is  defined by (recall that   $\phi_g= X_{g,m}+P \boldsymbol{\tilde \varphi}$)
\begin{align}\label{amplitude}
 &\caA_{\Sigma,g,{\bf x},\boldsymbol{\alpha},{\bf s},\boldsymbol{\beta},\boldsymbol{\zeta}}(F,\tilde{\boldsymbol{\varphi}})= \lim_{\eps\to 0}Z_{\Sigma,g,m}\caA^0_{\Sigma,g}(\tilde{\boldsymbol{\varphi}})\nonumber
 \\
 &
\times \E \big[F(\phi_g)\prod_{i=1}^m V_{\alpha_i,g,\epsilon}(x_i)\prod_{i=1}^{m_B} V_{\frac{\beta_i}{2},g,\epsilon}(s_i)e^{-\frac{Q}{4\pi}\int_\Sigma K_g\phi_g\dd {\rm v}_g-\frac{Q}{2\pi}\int_{\partial \Sigma}k_g\phi_g\dd \ell_g -\mu M_\gamma^g (\phi_g,\Sigma)-  M^g_{\gamma,\partial}(\phi_g,\mu_B\mathbf{1}_{\partial\Sigma_N}) }\big]
\end{align}
where 
the expectation $\E$ is over the  GFF $X_{g,m}$, $M_\gamma^g $ and $M^g_{\gamma,\partial}$  are defined as in \eqref{GMCg} 
and $Z_{\Sigma,g,m}$ is  a normalization constant given by 
 \begin{align}\label{znormal}
 Z_{\Sigma,g,m}=\det (\Delta_{g,m})^{-\hf}\exp\big(\frac{1}{8\pi}\int_{\partial\Sigma}k_g\,\dd\ell_g\big)
 \end{align}
with $k_g$ the geodesic curvature of $\pl \Sigma$, and $\caA^0_{\Sigma,g}(\tilde{\boldsymbol{\varphi}})$ is the free field amplitude defined as  
\begin{align}\label{amplifree}
\caA^0_{\Sigma,g}(\tilde{\boldsymbol{\varphi}})=e^{-\frac{1}{2}( \tilde{\boldsymbol{\varphi}}, (\mathbf{D}_\Sigma-\mathbf{D})  \tilde{\boldsymbol{\varphi}})_2}.
\end{align}
When $F=1$, we will simply write $ \caA_{\Sigma,g,{\bf x},\boldsymbol{\alpha},{\bf s},\boldsymbol{\beta},\boldsymbol{\zeta}}( \tilde{\boldsymbol{\varphi}})$.
 \end{definition}
 
 The case (A) is already treated in \cite{Wu1}. In case (B), note that the existence of the limit  results from the Girsanov argument as in \cite[Section 3]{DKRV16}. The definitions above trivially extend to the situation when $F$ is no more assumed to be nonnegative but with the further requirement that $\caA_{\Sigma,g,{\bf x},\boldsymbol{\alpha},{\bf s},\boldsymbol{\beta},\boldsymbol{\zeta}}(|F|)<\infty$ in the case $\partial\Sigma_D=\emptyset$ and $ \caA_{\Sigma,g,{\bf x},\boldsymbol{\alpha},{\bf s},\boldsymbol{\beta},\boldsymbol{\zeta}}(|F|,\tilde{\boldsymbol{\varphi}})<\infty$ $(\mu_0)^{\otimes b_\ell^D}\otimes(\mu_0^+)^{\otimes b_h^D}$ almost everywhere in the case $\partial\Sigma_D\not=\emptyset$.

  We complete this section by describing the way the amplitudes change under conformal changes of metrics and under the action of diffeomorphisms:  
 \begin{proposition}\label{Weyl} The amplitudes obey the following transformation rules: 
 \\
{\bf 1) Weyl covariance:} Assume $\partial\Sigma_D\not=\emptyset$. Let $g$ be admissible on $\Sigma$ and $g'$ another metric on $\Sigma$, Neumann extendible and conformal to $g$, i.e. of the form $g'=e^{\omega}g$ for some smooth map $\omega: \Sigma\to\R$.  Assume $\omega(x)=0$ on $ \partial\Sigma_D$. Then for $F$  measurable  nonnegative function on $H^{-s}(\Sigma)$  
\begin{align*}
\caA_{\Sigma, g,{\bf x},\boldsymbol{\alpha},{\bf s},\boldsymbol{\beta},\boldsymbol{\zeta}}(F,\tilde{\boldsymbol{\varphi}})= \caA_{\Sigma,g,{\bf x},\boldsymbol{\alpha},{\bf s},\boldsymbol{\beta},\boldsymbol{\zeta}}\big(F(\cdot-\frac{Q}{2}\omega ),\tilde{\boldsymbol{\varphi}}\big)\exp\Big(\frac{1+6Q^2}{96\pi}S_{\rm L}^0(\Sigma,g,\omega)-\sum_{j=1}^m\Delta_{\alpha_j}\omega(x_j) -\sum_{j=1}^{m_B}\frac{1}{2}\Delta_{\beta_j}\omega(s_j) \Big)
\end{align*}
where 
$$S_{\rm L}^0(\Sigma,g,\omega)=\int_{\Sigma}(|d\omega|_{g}^2+2K_{g}\omega) {\rm dv}_{g}.$$

{\bf 2) Diffeomorphism invariance:} Let $\Sigma,\Sigma'$ be two Riemann surfaces with corners and  $\psi:\Sigma\to\Sigma'$ be an  orientation preserving diffeomorphism, obeying  $\psi(\partial\Sigma_D)=\partial\Sigma_D'$ and $\psi(\partial\Sigma_N)=\partial\Sigma_N'$. Let  $F$ be a  measurable  nonnegative function on $H^{-s}(\Sigma)$. Setting $F_\psi(\phi):=F(\phi\circ\psi)$,  we have
 \begin{align*} 
\caA_{\Sigma,  \psi^*g,{\bf x},\boldsymbol{\alpha},{\bf s},\boldsymbol{\beta},\boldsymbol{\zeta}}(F,\tilde{\boldsymbol{\varphi}})= \caA_{\Sigma', g,\psi({\bf x}),\boldsymbol{\alpha},{\bf s},\boldsymbol{\beta},\psi\circ\boldsymbol{\zeta}}\big(F_\psi ,\tilde{\boldsymbol{\varphi}}\big).
\end{align*}
\end{proposition}

 The proof of this proposition is exactly the same way as \cite[Proposition 4.7]{GKRV21_Segal} up to the following lemma related to the  Weyl anomaly for $Z_{\Sigma,g,m}$: 
 \begin{proposition}
For $g',g$   two metrics as in item 1) of Proposition \ref{Weyl} above, we have 
\[
 Z_{\Sigma,g',m}=Z_{\Sigma,g,m}\exp(\frac{1}{96\pi}\int_{\Sigma}(|d\omega|^2_{g}+2K_{g}\omega) {\rm dv}_{g}).
\]
\end{proposition}
\begin{proof}
If $b^D_h=0$, then the proof is the same as \cite{OsgoodPS88}. Now we focus on the case $b^D_h>0$. Since both $g,g'$ are Neumann extendible, we can consider the Neumann double $(\Sigma^{\#2},g^{\#2})$ and $(\Sigma^{\#2},{g'}^{\#2})$, which are surfaces with a boundary, marked with Dirichlet. 
By \eqref{identity_det} and the fact that $\pl \Sigma_N$ is geodesic, we have $Z_{\Sigma^{\#2},g^{\#2},D}=Z_{\Sigma,g,m}Z_{\Sigma,g,D}$ where 
\begin{align*}
& Z_{\Sigma,g,D}:=\det(\Delta_{\Sigma,g,D})^{-1/2}\exp(\frac{1}{8\pi}\int_{\partial\Sigma_D}k_g\,\dd\ell_g), \\ 
& Z_{\Sigma^{\#2},g^{\#2},D}:=\det(\Delta_{\Sigma^{\#2},g^{\#2},D})^{-1/2}\exp(\frac{1}{4\pi}\int_{\partial\Sigma_D}k_g\,\dd\ell_g).\end{align*} 
By  \cite[Proposition 4.7]{GKRV21_Segal} (see end of the proof) and the fact that $(\omega+\tau_\Sigma^*\omega)|_{\pl \Sigma^{\#2}}=0$, we have
\[ Z_{\Sigma^{\#2},{g'}^{\#2},D}=Z_{\Sigma^{\#2},g^{\#2},D}\exp(\frac{1}{48\pi}\int_{\Sigma}(|d\omega|^2_{g}+2K_{g}\omega) {\rm dv}_{g}).\]
Therefore, to derive the Polyakov formula for $Z_{\Sigma,g,m}$, we only need to deduce the Polyakov formula for $\det(\Delta_{\Sigma,g,D})$. In this case, the boundary becomes curvilinear, following the terminology in  \cite[Theorem 1.6]{Aldana-Kirsten-Rowlett},  and there are extra correction terms arising from the corner singularity. But in our case, $\omega$  vanishes on these singular points, so the ordinary Polyakov formula still holds
\[ Z_{\Sigma,g',D}=Z_{\Sigma,g,D}\exp(\frac{1}{96\pi}\int_{\Sigma}(|d\omega|^2_{g}+2K_{g}\omega) {\rm dv}_{g}).\qedhere\]
\end{proof}
 
\subsection{Regularity of amplitudes}
The goal of this subsection is to establish Proposition \ref{integrcf} below, which   will be useful to control the integrability properties of the amplitude. We start with the following auxiliary result:
\begin{proposition}\label{Euler_char} 
Consider  a Riemann surface $\Sigma$ with corners equipped with a smooth compatible metric $g$. Then
\[\frac{1}{4\pi}\int_\Sigma K_g\dd {\rm v}_g+\frac{1}{2\pi}\int_{\partial\Sigma}k_g\dd \ell_g=\chi(\Sigma)-\frac{b_h^D}{2}=\frac{1}{2}\chi(\Sigma^{\#2})\]
where $\Sigma^{\#2}$ is the Neumann double of $\Sigma$.
\end{proposition}
\begin{proof} Recall that a Riemann surface with corners obeys   $b^D_h=b^N_h$.  The Gauss-Bonnet formula entails that
\[ \frac{1}{2}\int_\Sigma K_g\dd {\rm v}_g+ \int_{\partial\Sigma}k_g\dd \ell_g+\sum_i (\pi-\alpha_i)=2\pi\chi(\Sigma) \]
where the $\alpha_i=\pi/2$ stand for the angles of the surface, labelled with $i$.  Also, we can assume that $g$ is Neumann extendible since the left hand side only depends on the conformal class of the metric, and the angles are fixed. Then we can connect the number of corners $\kappa$ to the number of Dirichlet boundary half-circles by the relation $\kappa=2b_h^D$. The first equality of our statement follows. On the other hand,  Gauss-Bonnet on $\Sigma^{\#2}$ (which has no corners and only Dirichlet boundary) yields 
\[2\pi \chi(\Sigma^{\#2})= \frac{1}{2}\int_{\Sigma^{\#2}} K_g\dd {\rm v}_g+ \int_{\partial\Sigma^{\#2}_D}k_g\dd \ell_g=
\int_\Sigma K_g\dd {\rm v}_g+ 2\int_{\partial\Sigma_D}k_g\dd \ell_g\]
so the second equality follows from this.
\end{proof}

For $\tilde{\boldsymbol{\varphi}}:=(\tilde\varphi^\ell_1,\dots,\tilde\varphi^\ell_{b_\ell^D},\tilde\varphi^h_1,\dots,\tilde\varphi^h_{b_h^D})\in (H^{s}(\T))^{b_\ell^D}\times(H^{s}_{\rm even}(\T))^{b_h^D}$, we denote by  $(c^\ell_j)_{j\leq b^D_\ell}$ and $(c^h_j)_{j\leq b^D_h}$    their constant modes, set $ \bar c=\tfrac{1}{2b^D_\ell+b^D_h}(2\sum_{j=1}^{b_\ell^D} c^\ell_j+\sum_{j=1}^{b_h^D} c^h_j)$ and $\boldsymbol{\varphi}$ the collection of the  centered fields, i.e. with all zero modes removed (recall that $\tilde\varphi_j^\ell =c^\ell_j+\varphi^\ell_j$ or $\tilde\varphi^h_j=c^h_j+\varphi^h_j$).  Also, we set   $c_+=c\mathbf{1}_{\{c>0\}}$ and $c_-=c\mathbf{1}_{\{c<0\}}$.   
  \begin{proposition}\label{integrcf}
Let $(\Sigma,g,{\bf x},\boldsymbol{\alpha},{\bf s},\boldsymbol{\beta})$ be a  surface  with corners with a compatible metric, with $\partial\Sigma_D\not=\emptyset$ having  $b^D=b_h^D+b_\ell^D$ boundary components. We assume  that \eqref{hyp:cosmoamp} holds.
Then, there exists some constant $a>0$ such that for any $R>0$, there exists $C_R>0$ such that  $\mu_0^{\otimes b_\ell^D}\otimes (\mu_0^+)^{\otimes b_h^D}$-almost everywhere in  $ \boldsymbol{\tilde \varphi}\in  (H^{s}(\T))^{b_\ell^D}\times(H^{s}(\T^+))^{b_h^D}$ 
\begin{align*}
 \caA_{\Sigma,g,{\bf x},\boldsymbol{\alpha}}( \tilde{\boldsymbol{\varphi}}) \leq C_R  e^{s \bar c_--R \bar c _+  -a(  \sum_{j=1}^{b^D_\ell}(\bar c-c^\ell_j)^2+ \sum_{j=1}^{b^D_h}(\bar c-c^\ell_j)^2)}
A(\boldsymbol{\varphi})
\end{align*}
  satisfying $\int A(\boldsymbol{\varphi}) ^2 \,(\P_\T)^{\otimes b_\ell^D}\times(\otimes \P_{\T^+})^{\otimes b_h^D}<\infty$, 
where 
\begin{align*}
s:=\sum_i\alpha_i+\sum_j\frac{\beta_j}{2}-\frac{Q}{2}\chi(\Sigma^{\#2})
\end{align*}
\end{proposition}
\begin{proof}
First we estimate the free amplitude. Consider the Neumann double $\Sigma^{\#2}$, which has $2b_\ell^D+b_h^D$ Dirichlet boundary circles. The boundary field $ \tilde{\boldsymbol{\varphi}}$ induces the field 
$\tilde{\boldsymbol{\varphi}}+\tau_\Sigma^* \tilde{\boldsymbol{\varphi}}$ on $\pl \Sigma^{\#2}$ and we have by \eqref{DSigma-Ddouble} 
\[  \mathcal{A}_\Sigma^0(\tilde{\boldsymbol{\varphi}})=\mathcal{A}_{\Sigma^{\#2}}^0(\tilde{\boldsymbol{\varphi}}+\tau_\Sigma^* \tilde{\boldsymbol{\varphi}}).\]
By \cite[Lemma 4.5]{GKRV21_Segal}, there is $a>0$ such that for all $N>0$, there is $k_N>0$ such that 
\begin{equation}\label{est_A0} 
\mathcal{A}_\Sigma^0(\tilde{\boldsymbol{\varphi}})\leq e^{-a(  \sum_{j=1}^{b^D_\ell}(\bar c-c^\ell_j)^2+ \sum_{j=1}^{b^D_h}(\bar c-c^h_j)^2)}e^{\sum_{n=1}^\infty \sum_{j=1}^{b_\ell^D}a_n((x_{j,n}^\ell)^2+(y_{j,n}^\ell)^2)+\sum_{j=1}^{b_h^D}a_n((x_{j,n}^h)^2+(y_{j,n}^h)^2)}
\end{equation}
where $a_n\leq 1/4-a$ for $n\geq k_N$ and $a_n\leq \frac{1}{4N^n}$ for $n>k_N$. To obtain the estimate on the Liouville amplitude, it remains to bound the term involving the curvature and the Gaussian multiplicative chaos in \eqref{amplitude}, and this is done by the same method as in \cite[Lemma 4.6]{GKRV21_Segal}: using the Girsanov transform, the $X_{g,m}$ expectation in \eqref{amplitude} is bounded by 
\begin{equation}\label{upp_bdd_Amp} 
Ce^{-\frac{Q}{4\pi}\int_\Sigma K_gP\tilde{\boldsymbol{\varphi}}{\rm dv}_g-\frac{Q}{2\pi}\int_{\pl\Sigma} k_gP\tilde{\boldsymbol{\varphi}}{\rm d}\ell_g}\prod_{j=1}^m e^{\alpha_jP\tilde{\boldsymbol{\varphi}}(x_j)}\prod_{i=1}^{m_B} e^{\frac{\beta_i}{2}P\tilde{\boldsymbol{\varphi}}(s_j)}\E[e^{-\mu \int_\Sigma e^{\gamma u}M^g_\gamma(\phi_g,dx)-\int_{\pl \Sigma}e^{\gamma u}M_{\gamma,\pl}^g(\phi_g,\mu_B1_{\pl \Sigma_N})}]
\end{equation}
for some $C>0$ independent of $\tilde{\boldsymbol{\varphi}}$ and some function $u$ such that 
\[u(x)-\sum_{j=1}^m \alpha_jG_{g,m}(x,x_j)-\sum_{i=1}^m \frac{\beta_i}{2}G_{g,m}(x,s_i)\in C^\infty(\Sigma) .\] 
We write $P\tilde{\boldsymbol{\varphi}}=\bar{c}+P\boldsymbol{\varphi}+f$ with $f$ the harmonic extension of $\sum_{j=1}^{b_\ell^D}(c^\ell_j-\bar{c})1_{\mc{C}_j^D}+\sum_{j=1}^{b_h^D}(c^h_j-\bar{c})1_{\mc{B}_j^D}$ and note that 
$|f|\leq \sum_{j=1}^{b_\ell^D}(c^\ell_j-\bar{c})+\sum_{j=1}^{b_h^D}(c^h_j-\bar{c})$ by the maximum principle. We bound for some uniform $C>0$
\begin{align}
& \prod_{j=1}^m e^{\alpha_jP\tilde{\boldsymbol{\varphi}}(x_j)}\prod_{i=1}^{m_B} e^{\frac{\beta_i}{2}P\tilde{\boldsymbol{\varphi}}(s_j)}\leq e^{(\sum_j\alpha_j+\sum_ j\frac{\beta_j}{2})\bar{c}+(\sum_j |\alpha_j|+\sum_j\frac{|\beta_j|}{2})( \sum_j |c_j^\ell-\bar{c}|+\sum_j |c_j^h-\bar{c}|)}e^{C(\sum_{j}P\boldsymbol{\varphi}(x_j)+\sum_{i}P\boldsymbol{\varphi}(s_i))}\label{bound_mp}\\
& e^{-\frac{Q}{4\pi}\int_\Sigma K_gP\tilde{\boldsymbol{\varphi}}{\rm dv}_g-\frac{Q}{2\pi}\int_{\pl\Sigma} k_gP\tilde{\boldsymbol{\varphi}}{\rm d}\ell_g}\leq e^{-\frac{Q}{2}\chi(\Sigma^{\#2})\bar{c}+C(\sum_j |c_j^\ell-\bar{c}|+\sum_j |c_j^h-\bar{c}|)+ \sum_{n\geq 1}d_n(\sum_{j}|\varphi_{j,n}^\ell |+\sum_{j}|\varphi_{j,n}^h |) }\label{bound_curv}
\end{align}
using Proposition \ref{Euler_char}, with $d_n=\mc{O}(|n|^{-\infty})$. 
When $\bar{c}\leq 0$, we bound the expectation in \eqref{upp_bdd_Amp} by $1$ and obtain the desired result from \eqref{bound_mp}, \eqref{bound_curv}, \eqref{est_A0} and the fact that $e^{\sup_{x\in K} P\tilde{\boldsymbol{\varphi}}(x)}$ has finite positive moments of all order if $K\cap \pl\Sigma_D=\emptyset$ with $K$ a compact set. When $\bar{c}>0$, we can bound for all $N>0$ (using $e^{-x}\leq C_Nx^{-N}$ for $x\geq 0$) above the expectation in  \eqref{upp_bdd_Amp}  by 
\[ C_Ne^{-\gamma N\inf_{K}P\tilde{\boldsymbol{\varphi}}-\gamma N\bar{c}+\gamma N(\sum_j|c^\ell_j-\bar{c}|+\sum_j|c^h_j-\bar{c}|)}\E[ (\mu M^g_\gamma(X_{g,m},K)+M_{\gamma,\pl}^g(X_{g,m},\mu_B1_{K\cap \pl \Sigma_N})^{-N}]\]
where $K$ is a small open set not intersecting ${\bf x}\cup {\bf s}$. If $\mu>0$ we can choose $K$ to be a small ball contained in the interior of $\Sigma$. If $\mu_B$ does not identically vanish, we can choose $K$ to be a small ball such that $K\cap \partial \Sigma_N$ is a non empty open set contained in the support of $\mu_B$. The random variable $e^{-\inf_{x\in K} P\tilde{\boldsymbol{\varphi}}(x)}$ has finite positive moments of all orders and we use the fact that $M^g_\gamma(X_{g,m},K)$ or $M_{\gamma,\pl}^g(X_{g,m},\mu_B1_{K\cap \pl \Sigma_N})$ have negative moments of all orders (depending on the choice of $K$)  by \cite[Theorem 2.12]{rhodes2014_gmcReview} to deduce the result when $\bar{c}\geq 0$ (by combining again with \eqref{bound_mp}, \eqref{bound_curv}, \eqref{est_A0}).
\end{proof}

\begin{remark}
We can also get bounds in the case when both $\mu$ and $\mu_B$ vanish, in which case they become
\begin{align*}
 \caA_{\Sigma,g,{\bf x},\boldsymbol{\alpha}}( \tilde{\boldsymbol{\varphi}}) \leq C_R  e^{s \bar c   -a(  \sum_{j=1}^{b^D_\ell}(\bar c-c^\ell_j)^2+ \sum_{j=1}^{b^D_h}(\bar c-c^\ell_j)^2)}
A(\boldsymbol{\varphi})
\end{align*}
with $\int A(\boldsymbol{\varphi}) ^2 \,(\P_\T)^{\otimes b_\ell^D}\times(\otimes \P_{\T^+})^{\otimes b_h^D}<\infty$. The point is that we will use these estimates to prove that the amplitudes are Hilbert-Schmidt operators, but in the case $\mu=\mu_B=0$, the exponential term $e^{s \bar c  }$ is not bounded and prevents the amplitude from being Hilbert-Schmidt. The potential serves to control this divergence in other cases.
\end{remark}
 
\section{Gluing of amplitudes}\label{sec:gluing}
  
  In this section, we prove that the amplitudes behave nicely, namely they compose, under the operation of gluing the surfaces they are attached to along boundary circles or half-circles. We will distinguish two cases: the gluing of two different Riemann surfaces with corners or the self-gluing of a Riemann surface with corners, namely the gluing of two boundary components of the same surface. 
%
\subsection{Gluing two different Riemann surfaces}

Let  $(\Sigma_1,g_1,\boldsymbol{\zeta}_1)$ and  $(\Sigma_2,g_2,\boldsymbol{\zeta}_2)$ be two surfaces with corners, the metrics being admissible, and parametrized boundary, with 
\[\partial\Sigma_{D,1}=\bigcup_{j=1}^{b^D_{\ell,1}}\caC_{j,1}^D\cup\bigcup_{j=1}^{b^D_{h,1}}\caB_{j,1}^D, \quad \partial\Sigma_{D,2}=\bigcup_{j=1}^{b^D_{\ell,2}}\caC_{j,2}^D\cup\bigcup_{j=1}^{b^D_{h,2}}\caB_{j,2}^D\] 
(decomposition of the Dirichlet boundary into circles and half-circles).  Now we select some Dirichlet boundary or half-circles in each surface along which we want  to glue the two surfaces. Let  $\caC^D_{j,1}\subset\partial\Sigma_{D,1}$, $j\in\mathcal{J}_\ell :=  \{1,\dots,k_\ell\}$ with $k_\ell\leq \min\{b^D_{\ell,1}, b^D_{\ell,2}\}$,  be outgoing circle boundary components and $\caC^D_{j,2}\subset\partial\Sigma_{D,2}$, $j\in\mathcal{J}_\ell$,  be incoming circle boundary components. Also, let  $\caB^D_{j,1}\subset\partial\Sigma_{D,1}$, $j\in\mathcal{J}_h :=  \{1,\dots,k_h\}$ with $k_h\leq \min\{b^D_{h,1}, b^D_{h,2}\}$,  be outgoing half-circle boundary components and $\caB^D_{j,2}\subset\partial\Sigma_{D,2}$, $j\in\mathcal{J}_h$,  be incoming half-circle boundary components. 
Then we can  glue $k_\ell$ circular boundaries   $\caC_{j,1}^D,\caC_{j,2}^D$, with $j=1,\dots, k_\ell$, and $k_h$ half-circle boundaries   $\caB_{j,1}^D,\caB_{j,2}^D$, , with $j=1,\dots, k_h$, as prescribed in Subsection \ref{sub:gluing} to get a new admissible surface with corners and parametrized boundary  $(\Sigma,g,\boldsymbol{\zeta}):=(\Sigma_1\#\Sigma_2,g_1\#g_2, \boldsymbol{\zeta}_1\#\boldsymbol{\zeta}_2 )$ with  $b^D_{\ell,1}+b^D_{\ell,2}-2k_\ell$ Dirichlet boundary circles and 
$b^D_{h,1}+b^D_{h,2}-2k_h$ Dirichlet boundary half-circles,  where  $\boldsymbol{\zeta}:=\boldsymbol{\zeta}_1\#\boldsymbol{\zeta}_2$ denotes  the collection of parametrizations of the Dirichlet boundaries $\mc{C}^D_{j,1}$ and $\mc{C}^2_{j,2}$, where  the index $j$ runs over all possible indices with $j>k_\ell$, and the collection of parametrizations of the Dirichlet half-boundaries $\mc{B}^D_{j,1}$ and $\mc{B}^D_{j,2}$, where  the index $j$ runs over all possible indices with $j>k_h$. 

We consider marked interior points $\mathbf{x}_1=\{x_{11},\dots,x_{1m_1}\}$ on $\mathring{\Sigma}_1$ with respective weights $\boldsymbol{\alpha}_1=(\alpha_{11},\dots,\alpha_{1m_1})$ and   $\mathbf{x}_2=\{x_{21},\dots,x_{2m_2}\}$ on $\mathring{\Sigma}_2$ with respective weights $\boldsymbol{\alpha}_2=(\alpha_{21},\dots,\alpha_{2m_2})$.  Also, 
we consider marked boundary points $\mathbf{s}_1=\{s_{11},\dots,s_{1m^1_B}\}$ on $\partial\Sigma_{1,N}$ with respective weights $\boldsymbol{\beta}_1=(\beta_{11},\dots,\beta_{1m^1_B})\in \R^{m^1_B}$ and   $\mathbf{s}_2=\{s_{21},\dots,s_{2m_2}\}$ on $\partial\Sigma_{2,N}$ with respective weights $\boldsymbol{\beta}_2=(\beta_{21},\dots,\beta_{2m_2})\in \R^{m^2_B}$.  Set $\mathbf{x}:=(\mathbf{x}_1,\mathbf{x}_2)$, $\boldsymbol{\alpha}:=(\boldsymbol{\alpha}_1,\boldsymbol{\alpha}_2)$, $\mathbf{s}:=(\mathbf{s}_1,\mathbf{s}_2)$ and $\boldsymbol{\beta}:=(\boldsymbol{\beta}_1,\boldsymbol{\beta}_2)$.

The surface $\Sigma$ thus has an analytic Dirichlet boundary made up of the curves $\caC^D_{j,1} $ for $j=k_\ell+1,\dots,b^D_{\ell,1}$ and the curves $\caC^D_{j,2}\subset $ for $j=k_\ell+1,\dots,b^D_{\ell,2} $, and of the half-circles  $\caB^D_{j,1} $ for $j=k_h+1,\dots,b^D_{h,1}$ and the half-circles $\caB^D_{j,2}\subset $ for $j=k_h+1,\dots,b^D_{h,2} $.
We denote by $\mc{C}_j=\mc{C}^D_{j,1}=\mc{C}_{j,2}^D$ the glued curves on $\Sigma$ for $j=1,\dots,k_\ell$, and  by $\mc{B}_j=\mc{B}^D_{j,1}=\mc{B}_{j,2}^D$ the glued half-circles on $\Sigma$ for $j=1,\dots,k_h$.
 Boundary conditions on $\Sigma$ will thus be written as couples $(\tilde{\boldsymbol{\varphi}}_1,\tilde{\boldsymbol{\varphi}}_2)$ with $\tilde{\boldsymbol{\varphi}}_i:=(\boldsymbol{\tilde \varphi}_{i}^\ell,\boldsymbol{\tilde \varphi}_{i}^h)\in (H^{s}(\T))^{b_{\ell,i}^D-k_\ell}\times(H^{s}_{\rm even}(\T))^{b_{h,i}^D-k_h}$ for $i=1,2$.
Similarly, for $i=1,2$, the surface $\Sigma_i$ had analytic Dirichlet boundary made up of the analytic circles  $\caC^D_{j,i} $ for $j=1,\dots,k_\ell$ and  $\caC_{j,i}^D $ for $j=k_\ell+1,\dots,b^D_{\ell,i}$ and of the analytic half-circles $\caB^D_{j,i} $ for $j=1,\dots,k_h$ and  $\caB_{j,i}^D $ for $j=k_h+1,\dots,b^D_{h,i}$. Boundary conditions on $\Sigma_i$ will thus be written as couples $(\tilde{\boldsymbol{\varphi}},\tilde{\boldsymbol{\varphi}}_i)$ with $\tilde{\boldsymbol{\varphi}}:=(\boldsymbol{\tilde \varphi}^\ell,\boldsymbol{\tilde \varphi}^h)\in (H^{s}(\T))^{k_\ell}\times(H^{s}_{\rm even}(\T))^{k_h}$ and $\tilde{\boldsymbol{\varphi}}_i:=(\boldsymbol{\tilde \varphi}_{i}^\ell,\boldsymbol{\tilde \varphi}_{i}^h)\in (H^{s}(\T))^{b_{\ell,i}^D-k_\ell}\times(H^{s}_{\rm even}(\T))^{b_{h,i}^D-k_h}$ for $i=1,2$.

 \begin{figure}[h]   
 	\begin{tikzpicture}
 		\node[inner sep=0pt] (pant) at (0,0)
 		{ \includegraphics[scale=0.4]{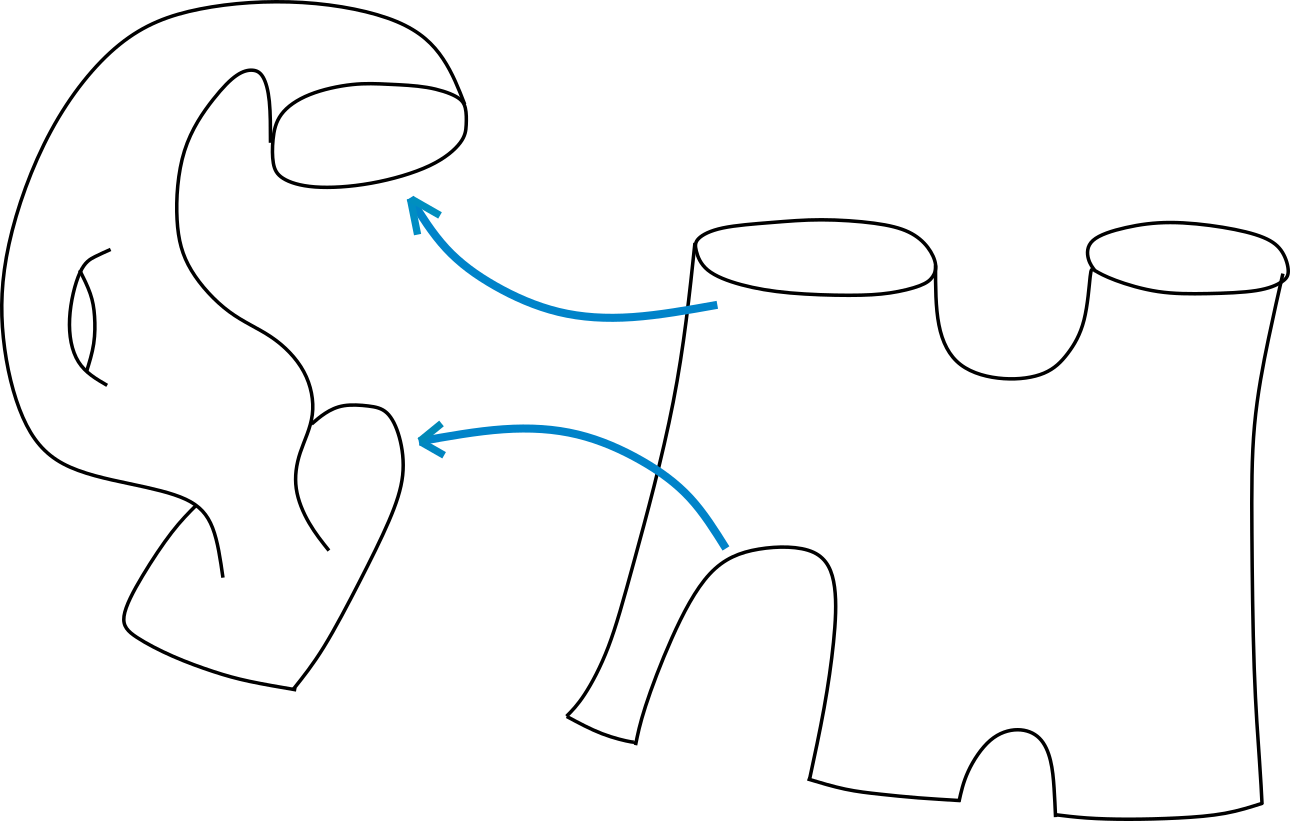}};
		\node[inner sep=0pt] (pant) at (8,0)
		 { \includegraphics[scale=0.4]{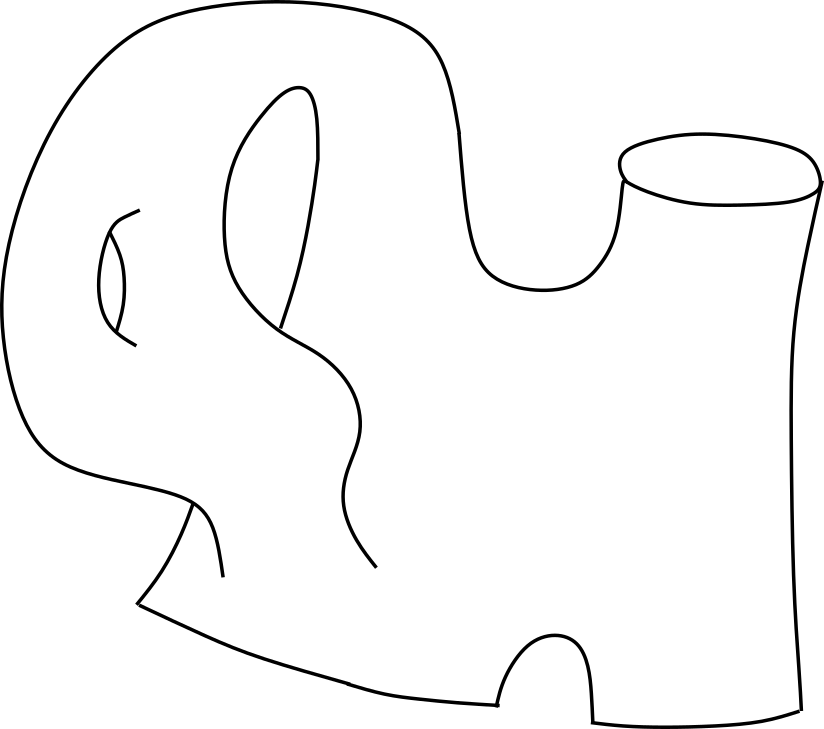}};
		 \draw (1.5,-0.5) node[right,black]{$\Sigma_2$} ; 
     \draw (-2.6,0.1) node[right,black]{$\Sigma_1$} ; 
      \draw (-1,1.6) node[right,black]{$\mc{C}_{11}$} ; 
      \draw (-1.6,-1) node[right,black]{$\mc{B}_{11}$} ; 
           \draw (0.2,-1.2) node[right,black]{$\mc{B}_{21}$} ; 
            \draw (1.5,-1.4) node[right,black]{$\mc{B}_{22}$} ; 
        \draw (0.2,1.3) node[right,black]{$\mc{C}_{21}$} ; 
            \draw (8,0) node[right,black]{$\Sigma$} ; 
 	\end{tikzpicture}
 	\caption{Gluing of two Riemann surfaces $\Sigma_1,\Sigma_2$ with corners along a circle and a half-circle, to produce the surface with corner $\Sigma$}\label{f:glue surfaces}
 \end{figure}
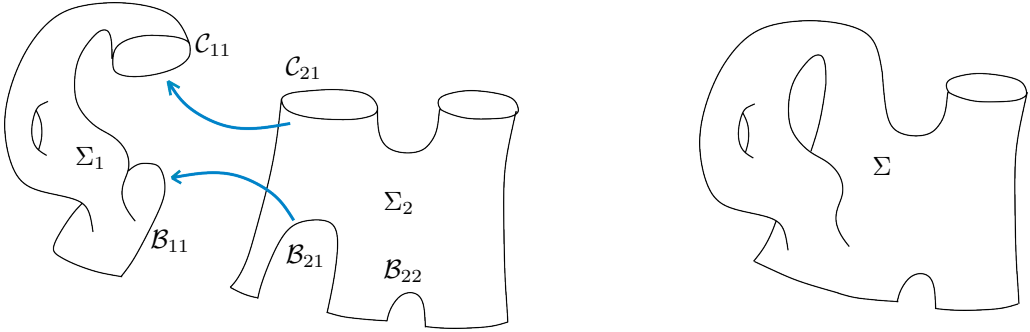
\begin{proposition}\label{glueampli}
Let $F_1,F_2$ be measurable  nonnegative functions respectively  on $H^{s}(\Sigma_1)$ and  $H^{s}(\Sigma_2)$ for $s\in(-1/2,0)$ and let us denote by $F_1\otimes F_2$ the functional on $H^{s}(\Sigma_1\#\Sigma_2)$ defined by  \[F_1\otimes F_2(\phi_{g_1\# g_2}):=F_1(\phi_{g_1\# g_2|\Sigma_1})F_2(\phi_{g_1\# g_2|\Sigma_2}).\] Then
 \begin{align*}
&\caA_{\Sigma_1\#\Sigma_2,g_1\# g_2, {\bf x},\boldsymbol{\alpha},{\bf s},\boldsymbol{\beta},\boldsymbol{\zeta}}
(F_1\otimes F_2,\tilde{\boldsymbol{\varphi}}_1,\tilde{\boldsymbol{\varphi}}_2)\\
&\quad =C\int \caA_{\Sigma_1,g_1,{\bf x}_1,\boldsymbol{\alpha}_1,{\bf s}_1,\boldsymbol{\beta}_1,\boldsymbol{\zeta}_1}(F_1,\tilde{\boldsymbol{\varphi}} ,\tilde{\boldsymbol{\varphi}}_1)\caA_{\Sigma_2,g_2,{\bf x}_2,\boldsymbol{\alpha}_2,{\bf s}_2,\boldsymbol{\beta}_2,\boldsymbol{\zeta}_2}(F_2,\tilde{\boldsymbol{\varphi}},\tilde{\boldsymbol{\varphi}}_2) \dd\mu_0^{\otimes_{k_\ell}} (\tilde{\boldsymbol{\varphi}}^\ell){\dd\mu_0^+}^{\otimes_{k_h}} (\tilde{\boldsymbol{\varphi}}^h).
\end{align*}
 where 
 \begin{align*}
 C= \frac{1}{2^{\frac{k_\ell}{2}+\frac{3 k_h}{4}} \pi^{k_\ell+\frac{3k_h}{4}}} \textrm{ if }\partial\Sigma_D \not =\emptyset,\qquad C= \frac{1}{2^{\frac{k_\ell}{2}+\frac{3 k_h}{4}-1} \pi^{k_\ell-1+\frac{3k_h}{4}}} \textrm{ if }\partial\Sigma_D =\emptyset.
 \end{align*}
  \end{proposition}

Next, we consider the case of self gluing. First we study the case of self-gluing of two Dirichlet boundary circles. 
 Let  $(\Sigma,g,\boldsymbol{\zeta})$  be a Riemann surface with corners the metric $g$ being admissible, with $b=b^D_\ell+b^D_h$ boundary components (here $b^D_\ell\geq 2$) such that the boundary contains an  outgoing boundary component  $\caC^D_j \subset\partial\Sigma_D$ and an incoming boundary component $\caC^D_k\subset\partial\Sigma_D$. Then we can glue the two boundary components as described in Subsection \ref{sub:gluing} to produce the surface $(\Sigma_{j\# k},g,\boldsymbol{\zeta}_\#)$. In this context, we will write the boundary component for $\Sigma$ as  $(\tilde\varphi^\ell_j,\tilde\varphi^\ell_k  ,\tilde{\boldsymbol{\varphi}}) \in   (H^{s}(\T)) \times (H^{s}(\T)) \times (H^{s}(\T))^{b_\ell^D-2}\times (H^{s}_{\rm even}(\T))^{b_h^D}$ where $\tilde\varphi_j$ corresponds to  $\caC^D_j$, $\tilde\varphi_k$ corresponds to  $\caC^D_k$ and $\tilde{\boldsymbol{\varphi}}'$ corresponds to the boundary components of $\Sigma_{j\# k}$. Now, we have

\begin{proposition}\label{selfglueampli}
Let  $F$ be a measurable  nonnegative function on $H^{s}(\Sigma)$ for $s<0$
 \[
\caA_{\Sigma_{j\# k},g,{\bf x}, \boldsymbol{\alpha},{\bf s},\boldsymbol{\beta},\boldsymbol{\zeta}}
(F,\tilde{\boldsymbol{\varphi}}')=C \int \caA_{\Sigma,g,{\bf x},\boldsymbol{\alpha},{\bf s},\boldsymbol{\beta},\boldsymbol{\zeta}}(F,\tilde\varphi , \tilde\varphi, \tilde{\boldsymbol{\varphi}}') \dd\mu_0(\tilde{\varphi}).
\]
where $C= \frac{1}{\sqrt{2} \pi }$ if $\partial\Sigma_D \not =\emptyset$ and $C= \sqrt{2}  $ if $\partial\Sigma_D =\emptyset$.
 \end{proposition}

In the third case, we consider the case of self-gluing of two half-circles. 
 Let  $(\Sigma,g,\boldsymbol{\zeta})$  be a Riemann surface with corners with metric $g$ being admissible, with $b=b^D_\ell+b^D_h$ boundary components ($b^D_h\geq2$) such that the boundary contains an  outgoing boundary component  $\caB^D_j \subset\partial\Sigma_D$ and an incoming boundary component $\caB^D_k\subset\partial\Sigma_D$. Then we can glue the two boundary components as described in Subsection \ref{sub:gluing} to produce the surface $(\Sigma_{jk},g,{\bf x},\boldsymbol{\alpha},{\bf s},\boldsymbol{\beta},\boldsymbol{\zeta})$. In this context, we will write the boundary component for $\Sigma$ as  $(\tilde\varphi^h_j,\tilde\varphi^h_k  ,\tilde{\boldsymbol{\varphi}}) \in   (H^{s}_{\rm even}(\T)) \times (H^{s}_{\rm even}(\T)) \times(H^{s}(\T))^{b_\ell^D} \times(H^{s}_{\rm even}(\T))^{b_h^D-2}   $ where $\tilde\varphi^h_j$ corresponds to  $\caB^D_j$, $\tilde\varphi^h_k$ corresponds to  $\caB^D_k$ and $\tilde{\boldsymbol{\varphi}}'$ corresponds to the boundary components of $\Sigma_{jk}$. Now, we have

\begin{proposition}\label{selfglueamplih}
Let  $F$ be a measurable  nonnegative function on $H^{s}(\Sigma)$ for $s<0$
 \[
\caA_{\Sigma_{jk},g, \boldsymbol{\alpha},{\bf s},\boldsymbol{\beta},\boldsymbol{\zeta}}
(F,\tilde{\boldsymbol{\varphi}}')=C \int \caA_{\Sigma,g,{\bf x},\boldsymbol{\alpha},{\bf s},\boldsymbol{\beta},\boldsymbol{\zeta}}(F,\tilde\varphi , \tilde\varphi, \tilde{\boldsymbol{\varphi}}') \dd\mu_0(\tilde{\varphi}).
\]
where $C= \frac{1}{(2 \pi)^{3/4}}$ if $\partial\Sigma_D \not =\emptyset$ and $C=(2\pi)^{1/4}$ if $\partial\Sigma_D =\emptyset$.
 \end{proposition}

  
 
 Now we prove these statements.
 
 \medskip
\noindent \emph{Proof of Propositions \ref{glueampli}, \ref{selfglueampli} and \ref{selfglueamplih}.}  Actually the proofs of the three propositions are similar, up to cosmetic changes, so that we only focus on Prop.  \ref{glueampli}.  The proof is similar to \cite[Prop. 5.2]{GKRV21_Segal}, but we write the details since the multiplicative constants $C$ are different. 
 For notational simplicity, we will write $\Sigma$ for $\Sigma_1\#\Sigma_2$, $g$ for $g_1\# g_2$, $m=m_1+m_2$ and $m_B=m_B^1+m_B^2$. We split the proof in two cases depending on whether the resulting surface $\Sigma $ has a non trivial Dirichlet boundary (case 1) or not (case 2).  

\medskip
$\bullet$ Assume first $\partial\Sigma_D \not= \emptyset$.  
Let us denote   the union of glued circle boundary components by $\mathcal{C}:=\bigcup_{j=1}^{k_\ell}\mc{C}_j$  and the union of glued half-circle boundary components by $\mathcal{B}:=\bigcup_{j=1}^{k_h}\mc{B}_j$.   Define next
\begin{equation}
\widetilde{\mathcal{A}}_{\Sigma,g}(F_1\otimes F_2,\tilde{\boldsymbol{\varphi}}_1,\tilde{\boldsymbol{\varphi}}_2):=\E \big[F_1\otimes F_2(\phi_g) \prod_{j=1}^{m} V_{\alpha_j,g}(x_j)\prod_{j=1}^{m_B} V_{\frac{\beta_j}{2},g}(s_j)e^{-\frac{Q}{4\pi}\int_\Sigma K_g \phi_g\dd {\rm v}_g-\mu M^g_\gamma(\phi_g,\Sigma)-  M^g_{\gamma,\partial}(\phi_g,\mu_B\mathbf{1}_{\partial\Sigma_N})}\big]
\end{equation}
where $\phi_g=X_{g,m}+P(\tilde{\boldsymbol{\varphi}}_1,\tilde{\boldsymbol{\varphi}}_2)$, $X_{g,m}$ is the  GFF with mixed boundary conditions on $\Sigma$ and expectation $\E$ is taken over this  GFF, $P(\tilde{\boldsymbol{\varphi}}_1,\tilde{\boldsymbol{\varphi}}_2)$ stands for the harmonic extension to $\Sigma$ of the boundary fields $\tilde{\boldsymbol{\varphi}}_1,\tilde{\boldsymbol{\varphi}}_2$, which stand respectively for the  boundary conditions on the remaining (i.e. unglued) components of $\partial \Sigma_1$ and  $\partial \Sigma_2$, namely (with $\zeta_{i,j}$ the parametrization of $\mc{C}_{i,j}$)
\begin{align*}
\Delta_g P(\tilde{\boldsymbol{\varphi}}_1,\tilde{\boldsymbol{\varphi}}_2)=0\quad  \text{on }\Sigma ,&\quad  \text{  for }j>k_\ell\quad P(\tilde{\boldsymbol{\varphi}}_1,\tilde{\boldsymbol{\varphi}}_2)_{|\mc{C}_{1,j}}= \tilde{\varphi}_{1,j}\circ (\zeta_{1,j}^\ell)^{-1}  ,& P(\tilde{\boldsymbol{\varphi}}_1,\tilde{\boldsymbol{\varphi}}_2)_{|\mc{C}_{2,j}}= \tilde{\varphi}_{2,j}\circ (\zeta_{2,j}^\ell)^{-1}  \\
\partial_\nu P(\tilde{\boldsymbol{\varphi}}_1,\tilde{\boldsymbol{\varphi}}_2)=0\quad  \text{on }\partial\Sigma_N ,&\quad\text{  for }j>k_h\quad P(\tilde{\boldsymbol{\varphi}}_1,\tilde{\boldsymbol{\varphi}}_2)_{|\mc{B}_{1,j}}= \tilde{\varphi}_{1,j}\circ  (\zeta_{1,j}^h)^{-1}  ,& P(\tilde{\boldsymbol{\varphi}}_1,\tilde{\boldsymbol{\varphi}}_2)_{|\mc{B}_{2,j}}= \tilde{\varphi}_{2,j}\circ (\zeta_{2,j}^h)^{-1}  .
\end{align*}
Therefore
$$
\caA_{\Sigma,g, {\bf x},\boldsymbol{\alpha},{\bf s},\boldsymbol{\beta},\boldsymbol{\zeta}}(F_1\otimes F_2,\tilde{\boldsymbol{\varphi}}_1,\tilde{\boldsymbol{\varphi}}_2)
=
Z_{\Sigma,m,g}\caA^0_{\Sigma,g}(\tilde{\boldsymbol{\varphi}}_1,\tilde{\boldsymbol{\varphi}}_2)
\widetilde{\mathcal{A}}_{\Sigma,g}(F_1\otimes F_2,\tilde{\boldsymbol{\varphi}}_1,\tilde{\boldsymbol{\varphi}}_2).
$$
For $i=1,2$, let now $X_i:=X_{g_i,m}$   be a  GFF with mixed b.c. respectively on $\Sigma_i$   with Dirichlet boundary condition on $\partial\Sigma_{D,i}$  and Neumann  boundary condition on $\partial\Sigma_{N,i}$.  We assume $X_1$ and $X_2$ are independent. Then we have the following decomposition in law (see Proposition \ref{decompGFF} item 1.a)
\begin{align*}
X_{g,m}\stackrel{{\rm law}}=X_1+X_2+P{\bf Y}
\end{align*}
where ${\bf Y}:=({\bf Y}^\ell,{\bf Y}^h)$ is the restriction of $X_{g,m}$ to the  glued boundary components $\mathcal{C}\cup\mathcal{B}$ expressed in parametrized coordinates, i.e.   
$${\bf Y} ^\ell:=( X_{g,m|_{{\caC}_1}}\circ \zeta^\ell_{1,1},\dots,X_{g,m|_{{\caC}_{k_\ell}}}\circ \zeta^\ell_{1,k_\ell})\qquad {\bf Y} ^h:= (X_{g,m|_{{\caB}_1}}\circ \zeta^h_{1,1},\dots,X_{g,m|_{{\caB}_{k_h}}}\circ \zeta^h_{1,k_h}),$$
and  $P{\bf Y}$ is its harmonic extension to $ \Sigma $, which vanishes on $\partial\Sigma_D$ (non empty) and $\partial_\nu P{\bf Y}=0$ on $\partial\Sigma_N$. Let us just stress here that making sense of the restriction of $X_{g,m}$ is not completely straightforward as the GFF is only a distribution, but this can be done in the same way as the restriction of the GFF to a circle (or half-circle).  Finally we denote by $\mc{R}_{\mathcal{C}\cup\mathcal{B}}$ the restriction of  the harmonic function $P(\tilde{\boldsymbol{\varphi}}_1,\tilde{\boldsymbol{\varphi}}_2)$ to $\mathcal{C}\cup\mathcal{B}$ in parametrized coordinates
$$
\mc{R}_{\mathcal{C}\cup\mathcal{B}}
:=(\mc{R}_{\mathcal{C}\cup\mathcal{B}}^\ell,\mc{R}_{\mathcal{C}\cup\mathcal{B}}^h)
$$
with 
\begin{align*}
& \mc{R}_{\mathcal{C}\cup\mathcal{B}}^\ell:=(P(\tilde{\boldsymbol{\varphi}}_1,\tilde{\boldsymbol{\varphi}}_2)_{|\mc{C}_{1}}\circ\zeta^\ell_{1,1},\dots,P(\tilde{\boldsymbol{\varphi}}_1,\tilde{\boldsymbol{\varphi}}_2)_{|\mc{C}_{k_\ell}}\circ \zeta^\ell_{1,k_\ell})\\ & \mc{R}_{\mathcal{C}\cup\mathcal{B}}^h:=(P(\tilde{\boldsymbol{\varphi}}_1,\tilde{\boldsymbol{\varphi}}_2)_{|\mc{B}_{1}}\circ\zeta^h_{1,1},\dots,P(\tilde{\boldsymbol{\varphi}}_1,\tilde{\boldsymbol{\varphi}}_2)_{|\mc{B}_{k_h}}\circ \zeta^h_{1,k_h}).
\end{align*}
On $\Sigma_i$ ($i=1,2$), the function $P(\tilde{\boldsymbol{\varphi}}_1,\tilde{\boldsymbol{\varphi}}_2)+P{\bf Y}$ is harmonic with boundary values (expressed in parametrized coordinates  on $\Sigma_i$) $ \tilde\varphi^\ell_{i,j}$ on $\mathcal{C}_{i,j}$ for $j>k_\ell$, $ \tilde\varphi^h_{i,j}$ on $\mathcal{B}_{i,j}$ for $j>k_h$, $({\bf Y}^\ell+\mc{R}^\ell_{\mathcal{C}\cup\mathcal{B}})_j$ on $\mc{C}_{j}$ for $j=1,\dots,k_\ell$ and $({\bf Y}^h+\mc{R}^h_{\mathcal{C}\cup\mathcal{B}})_j$ on $\mc{B}_{j}$ for $j=1,\dots,k_h$.  Also, note  that $\mc{R}_{\mathcal{C}\cup\mathcal{B}}$ is a random function with values in $(H^N(\T))^{k_\ell}\times (H^N_{\rm even}(\T))^{k_h}$ for all $N>0$.

Thus, using the decomposition of the GFF, we get
\begin{equation}
\widetilde{\mathcal{A}}_{\Sigma,g}(F_1\otimes F_2,\tilde{\boldsymbol{\varphi}}_1,\tilde{\boldsymbol{\varphi}}_2)
=
 \int \widetilde{\mathcal{A}}_ {\Sigma_1,g_1}(F_1,\tilde{\boldsymbol{\varphi}} +\mc{R}_{\mathcal{C}\cup\mathcal{B}},\tilde{\boldsymbol{\varphi}}_1 )\widetilde{\mathcal{A}}_ {\Sigma_2,g_2}(F_2,\tilde{\boldsymbol{\varphi}} +\mc{R}_{\mathcal{C}\cup\mathcal{B}},\tilde{\boldsymbol{\varphi}}_2 )\dd \P_{{\bf Y}} (\tilde{\boldsymbol{\varphi}})
\end{equation}
with $\P_{{\bf Y}}$ the law of $ {\bf Y} $ and, for $i=1,2$,
\begin{equation}
\widetilde{\mathcal{A}}_{\Sigma_i,g_i}(F_i,\tilde{\boldsymbol{\varphi}} ,\tilde{\boldsymbol{\varphi}}_i):=\E \big[F_i( \phi_i)\prod_{q=1}^{m_i } V_{\alpha_{iq},g_i}(x_{iq})\prod_{j=1}^{m^i_B} V_{\frac{\beta_{ij}}{2},g_i}(s_{ij})e^{-\frac{Q}{4\pi}\int_{\Sigma_i} K_{g_i} \phi_i\dd {\rm v}_{g_i}-\mu M_\gamma(\phi_i,\Sigma_i)-  M_{\gamma,\partial}(\phi_i,\mu_B\mathbf{1}_{\partial\Sigma_{N,i}})}\big]
\end{equation}
where $\E$ is taken with respect to the mixed b.c. GFF  $X_i$  on $\Sigma_i$, $\phi_i=X_i+P(\tilde{\boldsymbol{\varphi}} ,\tilde{\boldsymbol{\varphi}}_i)$ and $P(\tilde{\boldsymbol{\varphi}} ,\tilde{\boldsymbol{\varphi}}_i)$ stands for the harmonic extension on $\Sigma_i$ of the boundary fields $\tilde{\boldsymbol{\varphi}} ,\tilde{\boldsymbol{\varphi}}_i$ respectively on $\mathcal{C}\cup \mathcal{B}$ and $\partial\Sigma_{D,i}\setminus (\mathcal{C}\cup \mathcal{B})$.

  Notice that the orientation of incoming boundary (half-)curves $\mc{C}_{2,j}$ for $j\leq k_\ell$ or $\mc{B}_{2,j}$ for $j\leq k_h$ on $\Sigma_2$ is opposite to their orientation as curves drawn on $\Sigma$   due to \eqref{gluingonboundary}, which is used to get  the relation above. Writing   $T_* \P_{\boldsymbol{Y}} $ for the pushforward of the measure $ \P_{\boldsymbol{Y}} $ by the map $T:\tilde{\boldsymbol{\varphi}}\mapsto \tilde{\boldsymbol{\varphi}}+\mc{R}_{\mathcal{C}\cup\mathcal{B}}$ and then writing the $\widetilde{\mathcal{A}}$'s in terms of the amplitudes on $\Sigma_1$ and $\Sigma_2$, we obtain
\begin{multline}\label{comput1}
\caA_{\Sigma,g, {\bf x},\boldsymbol{\alpha},{\bf s},\boldsymbol{\beta},\boldsymbol{\zeta}}(F_1\otimes  F_2,\tilde{\boldsymbol{\varphi}}_1,\tilde{\boldsymbol{\varphi}}_2)
=
 \frac{Z_{\Sigma,m,g}}{Z_{\Sigma_1,m,g_1}Z_{\Sigma_2,m,g_2}}\caA^0_{\Sigma,g}(\tilde{\boldsymbol{\varphi}}_1,\tilde{\boldsymbol{\varphi}}_2) 
 \\
  \times \int \frac{\mathcal{A}_{\Sigma_1,g_1,{\bf x}_1,\boldsymbol{\alpha}_1,{\bf s}_1,\boldsymbol{\beta}_1,\boldsymbol{\zeta}_1}(F_1,\tilde{\boldsymbol{\varphi}}  ,\tilde{\boldsymbol{\varphi}}_1 )}{\caA^0_{\Sigma_1,g_1}(\tilde{\boldsymbol{\varphi}} ,\tilde{\boldsymbol{\varphi}}_1)} \frac{\mathcal{A}_{\Sigma_2,g_2,{\bf x}_2,\boldsymbol{\alpha}_2,{\bf s}_2,\boldsymbol{\beta}_2,\boldsymbol{\zeta}_2}(F_2,\tilde{\boldsymbol{\varphi}}  ,\tilde{\boldsymbol{\varphi}}_2 )}{\caA^0_{\Sigma_2,g_2}(\tilde{\boldsymbol{\varphi}} ,\tilde{\boldsymbol{\varphi}}_2)} \,  \dd T_*\P_{\boldsymbol{Y}} (\tilde{\boldsymbol{\varphi}}) .
\end{multline}

The proof of the case  $\partial\Sigma  \not= \emptyset$ is completed with the following lemma combined with Proposition \ref{detformula}:

\begin{lemma}\label{density1}
The measure $T_*\P_{\boldsymbol{Y}} $  satisfies
\[T_*\P_{\boldsymbol{Y}} = 2^{-k_h/2}\pi^{-(k_\ell+k_h)/2 } {\rm det}_{\rm Fr}(\mathbf{D}_{\Sigma,{\mc{C}},\mc{B}}(2\mathbf{D}_0)^{-1})^{1/2} \frac{\caA^0_{\Sigma_1,g_1}(\tilde{\boldsymbol{\varphi}} ,\tilde{\boldsymbol{\varphi}}_1)\caA^0_{\Sigma_2,g_2}(\tilde{\boldsymbol{\varphi}} ,\tilde{\boldsymbol{\varphi}}_2)}{\caA^0_{\Sigma,g}(\tilde{\boldsymbol{\varphi}}_1,\tilde{\boldsymbol{\varphi}}_2)} \dd\mu_0^{\otimes_{k_\ell}} (\tilde{\boldsymbol{\varphi}}^\ell){\dd\mu_0^+}^{\otimes_{k_h}} (\tilde{\boldsymbol{\varphi}}^h)\]
with $\mathbf{D}_{\Sigma,\mc{C},\mc{B}}$ the Dirichlet-to-Neumann operator on $L^2(\mc{C},\mc{B}):=L^2(\mathbb{T})^{k_\ell}\times L^2_{\rm even}(\mathbb{T})^{k_h}$ defined as in \eqref{defDSigmaC}, $\mathbf{D}_0:=\Pi_0+\mathbf{D}$  on $L^2(\mc{C},\mc{B})$ defined using \eqref{defmathbfD}$+$\eqref{defmathbfDhalf} and \eqref{pi0}, and  $\det_{\rm Fr}(\mathbf{D}_{\Sigma,{\mc{C}},\mc{B}}(2\mathbf{D}_0)^{-1})$ is the  Fredholm determinant as in Lemma \ref{lemmaDSigma-D}. 
\end{lemma}
\begin{proof}

First we get rid of the shift.  Decompose  $ \mathbf{D}_{\Sigma,{\mc{C},\mc{B}}}\tilde{\boldsymbol{\varphi}}=( \mathbf{D}^\ell_{\Sigma,{\mc{C},\mc{B}}}\tilde{\boldsymbol{\varphi}}, \mathbf{D}^h_{\Sigma,{\mc{C},\mc{B}}}\tilde{\boldsymbol{\varphi}})\in H^s(\T)^{k_\ell}\times H^s_{\rm even}(\T)^{k_h}$.
  By Cameron-Martin, and using from Lemma \ref{rest_DNmap}  that $ \mathbf{D}_{\Sigma,{\mc{C},\mc{B}}}G_{g,m}=2\pi {\rm Id}
$ on $H^s(\T)^{k_\ell}\times H^s_{\rm even}(\T)^{k_h}$, we get 
\begin{equation}\label{CameronMartin}
T_*\P_{{\bf Y}}=
\exp\Big( ({\bf Y},{\bf D}_{\Sigma,{\mc{C}},\mc{B}} \mc{R}_{\mathcal{C}\cup\mathcal{B}})_{2}
-\frac{1}{2}(\mc{R}_{\mathcal{C}\cup\mathcal{B}},{\bf D}_{\Sigma,\mc{C},\mc{B}} \mc{R}_{\mathcal{C}\cup\mathcal{B}})_{2} \Big) \P_{\boldsymbol{Y}}.
\end{equation}
Now we claim for measurable bounded functions $F$
\begin{equation}\label{densityDN}
\begin{split}
& \int F( \tilde{\boldsymbol{\varphi}})\P_{\boldsymbol{Y}}(\dd \tilde{\boldsymbol{\varphi}})\\
&=\frac{  1}{  2^{\frac{k_h }{2}}\pi^{\frac{k_\ell+k_h}{2}}\det(\mathbf{D}_{\Sigma,\mc{C},\mc{B}}(2\mathbf{D}_0)^{-1})^{-1/2} }\int F(    \tilde{\boldsymbol{\varphi}})\exp(-\frac{1}{2}(\tilde{\boldsymbol{\varphi}},\widetilde{\mathbf{D}}_{\Sigma,\mc{C},\mc{B}} \tilde{\boldsymbol{\varphi}})_2)\dd\mu_0^{\otimes_{k_\ell}} (\tilde{\boldsymbol{\varphi}}^\ell){\dd\mu_0^+}^{\otimes_{k_h}} (\tilde{\boldsymbol{\varphi}}^h).
\end{split}\end{equation}
Indeed, consider the projection $\Pi_N:(H^s(\T))^{k_\ell}\times (H^s_{\rm even}(\T))^{k_h}\to (H^s(\T))^{k_\ell}\times (H^s_{\rm even}(\T))^{k_h}$ where the Fourier series of each component is truncated to its first $N$   Fourier components (i.e. with index $k$ such that $|k|\leq N$), and let $\mc{H}_N=\Pi_N\big((H^s(\T))^{k_\ell}\times (H^s_{\rm even}(\T))^{k_h}\big)$. Then we set $ \mathbf{D}_{\Sigma,{\mc{C}},\mc{B}}^N:=\Pi_N\circ  \mathbf{D}_{\Sigma,{\mc{C}},\mc{B}}\circ \Pi_N$, $ \mathbf{D}_0^N:=\Pi_N\circ  \mathbf{D}_0\circ \Pi_N$, $ \mathbf{D}^N:=\Pi_N\circ  \mathbf{D}\circ \Pi_N$, $\widetilde{\mathbf{D}}_{\Sigma,{\mc{C}},\mc{B}}^N:=\Pi_N\circ (\mathbf{D}_{\Sigma,{\mc{C}},\mc{B}}-2\mathbf{D})\circ \Pi_N$ and
$$
\dd\P^N(\tilde{\boldsymbol{\varphi}}):=\frac{1}{ 2^{\frac{k_h }{2}}\pi^{\frac{k_\ell+k_h}{2}}\det_{\mc{H}_N}(\mathbf{D}_{\Sigma,{\mc{C}},\mc{B}}^N(2\mathbf{D}^N_0)^{-1})^{-1/2} }\exp\Big(-\frac{1}{2}(\tilde{\boldsymbol{\varphi}},\widetilde{\mathbf{D}}^N_{\Sigma,{\mc{C}},\mc{B}} \tilde{\boldsymbol{\varphi}})_2\Big)\dd\mu_0^{\otimes_{k_\ell}} (\tilde{\boldsymbol{\varphi}}^\ell){\dd\mu_0^+}^{\otimes_{k_h}} (\tilde{\boldsymbol{\varphi}}^h).
$$
   Let  $\mathbf{c}:=(c^\ell_1,\dots,c^\ell_{k_\ell},c^h_1,\dots,c^h_{k_h})\in \R^{k_\ell+k_h}$ and $|\mathbf{c}|$ its Euclidean norm. First notice that  $(\tilde{\boldsymbol{\varphi}}, \mathbf{D}_{\Sigma,{\mc{C}},\mc{B}} \tilde{\boldsymbol{\varphi}})_2\geq a |\mathbf{c}|^2+a(\tilde{\boldsymbol{\varphi}},2 \mathbf{D}  \tilde{\boldsymbol{\varphi}})_2$ for some $a>0$. Indeed, by
 Lemma \ref{lemmaDSigma-D}, $(2{\bf D}_0)^{-1/2}\mathbf{D}_{\Sigma,{\mc{C},\mc{B}}}(2{\bf D}_0)^{-1/2}$ is a positive Fredholm self-adjoint  bounded operator  on $L^2$ of the form $1+K$ for $K$ compact, thus $(2{\bf D}_0)^{-1/2}\mathbf{D}_{\Sigma,{\mc{C}},\mc{B}}(2{\bf D}_0)^{-1/2}\geq a\, {\rm Id}$ for some $a>0$. We deduce that $(\tilde{\boldsymbol{\varphi}},\widetilde{\mathbf{D}}^N_{\Sigma,{\mc{C}},\mc{B}} \tilde{\boldsymbol{\varphi}})_2\geq  a |\mathbf{c}|^2+(a-1) (\Pi_N\tilde{\boldsymbol{\varphi}},2 \mathbf{D}  \Pi_N\tilde{\boldsymbol{\varphi}})_2$, and therefore $\P^N$ is a finite measure. 
   The computation of its mass shows it is a probability measure. Indeed, this follows from the two following facts. First the probability measure  $\P_\T ^{\otimes_{k_\ell}} \otimes \P_{\T^+} ^{\otimes_{k_h}}  \circ \Pi_N^{-1} $ is given by
\[\int F(\Pi_N\boldsymbol \varphi)\P_\T ^{\otimes_{k_\ell}} \otimes \P_{\T^+} ^{\otimes_{k_h}}   (\dd \boldsymbol \varphi) =\int F(\Pi_N\boldsymbol \phi) e^{-\frac{1}{2}(\Pi_N \boldsymbol \phi,2\mathbf{D}  \Pi_N \boldsymbol \phi)_2}\frac{\dd \Pi_N\boldsymbol \phi}{(2\pi)^{N(k_\ell+k_h/2)}{\det_{\mc{H}_N}'( 2\mathbf{D}^N)^{-1/2}}}\]
   with $\boldsymbol \phi:=(\phi^\ell_1,\dots,\phi^\ell_{k_\ell},\phi^h_1,\dots,\phi^h_{k_h})$ and 
\begin{align*}   
 \phi^\ell_j(\theta)=& \sum_{n>0}\big(u^j_n\sqrt{2}\cos(n\theta)-v^j_n\sqrt{2}\sin(n\theta)\big), & j=1,\dots,k_\ell,
 \\
 \phi^h_j(\theta)=& \sum_{n>0}2w^j_n\cos(n\theta), &  j=1,\dots,k_h,
 \\
\dd \Pi_N\boldsymbol \phi=& \prod_{j=1}^{k_\ell}\prod_{n=1}^N \dd u_n^j\dd v_n^j\prod_{j=1}^{k_h}\prod_{n=1}^N \dd w_n^j  .
\end{align*}
Also, note that $\det_{\mc{H}_N}'( 2\mathbf{D}^N)^{-1/2}=2^{\frac{k_\ell+k_h}{2}}\det_{\mc{H}_N}( 2\mathbf{D}_0^N)^{-1/2}$.

Let us recall  the following relation for  nonnegative symmetric matrices $A,B$ with $B$ positive definite and $F$ continuous bounded
$$\int_{\R^p}F(x)e^{-\frac{1}{2}\langle (B-A) x,x\rangle}e^{-\frac{1}{2}\langle A x,x\rangle}\,\frac{\dd x}{(2\pi)^{p/2}\det(A_0)^{-1/2}}=\det(BA_0^{-1})^{-1/2}\int_{\R^p}F(x)e^{-\frac{1}{2}\langle B x,x\rangle} \,\frac{\dd x}{(2\pi)^{p/2}\det(B)^{-1/2}},$$
where $A_0$ is the symmetric positive definite matrix of the operator on $\R^p$ that is equal to identity on the kernel of $A$ and equal to $A$ on the orthogonal of the kernel. From these two relations we get that
\begin{multline*}
\int F(\Pi_N \tilde{\boldsymbol{\varphi}})\P^N(\dd \tilde{\boldsymbol{\varphi}})
=\\
\frac{ 2^{-\frac{k_h}{2}}}{(2\pi)^{N(k_\ell+k_h/2)   +\frac{1}{2}(k_\ell+k_h)}\det_{\mc{H}_N}(\mathbf{D}_{\Sigma,{\mc{C}},\mc{B}}^N )^{-1/2} }\int F(\mathbf{c}+\Pi_N \boldsymbol{\phi})\exp(-\frac{1}{2}(\mathbf{c}+\boldsymbol{ \phi}, \mathbf{D}^N_{\Sigma,{\mc{C}},\mc{B}} (\mathbf{c}+\boldsymbol{ \phi}))_2)   \dd \mathbf{c}\dd \Pi_N\boldsymbol \phi.
\end{multline*}
In particular, for $f$   a trigonometric polynomial
$$\int e^{ ( f, \tilde{\boldsymbol{\varphi}})_2}\P^N(\dd \tilde{\boldsymbol{\varphi}})=e^{\frac{1}{2}(f, (\mathbf{D}_{\Sigma,{\mc{C}},\mc{B}}^N)^{-1}f)_2}\to e^{\frac{1}{4\pi}(f,Gf)_2}=\E[ e^{ (f, \boldsymbol{Y})_2}], \quad \text{as }N\to\infty$$   
where we have used the fact that $(\mathbf{D}_{\Sigma,{\mc{C}},\mc{B}}^N|_{\mc{H}_N})^{-1}\Pi_N\to G/(2\pi)$ when $N\to \infty$, as bounded operators on $(H^s(\T))^{k_\ell}\times (H^s_{\rm even}(\T))^{k_h}$ for all $s\in \R$.   Here we have also used the relation 
\[\frac{ 2^{-\frac{k_h}{2}}}{(2\pi)^{N(k_\ell+k_h/2)   +\frac{1}{2}(k_\ell+k_h)}  }\int  \exp(-\frac{1}{2}(\Pi_N(\mathbf{c}+\boldsymbol{ \phi}), \Pi_N (\mathbf{c}+\boldsymbol{ \phi}))_2)   \dd \mathbf{c}\dd \Pi_N\boldsymbol \phi  =1,\]
the odd factor $2^{\frac{k_h}{2}}$ coming from integration over the boundary zero mode, due to the $1/2$ factor in the relation \eqref{unpairing}.

 On the other hand, we also claim that the left-hand side above converges as $N\to\infty$ towards
\begin{equation}\label{limitconvdominee}
\frac{ 1}{ 2^{k_h/2}\pi^{\frac{1}{2}(k_\ell+k_h)}\det_{\rm Fr}(\mathbf{D}_{\Sigma,{\mc{C}},\mc{B}}(2\mathbf{D}_0)^{-1})^{-1/2} }\int e^{ ( f, \tilde{\boldsymbol{\varphi}})_2}\exp(-\frac{1}{2}(\tilde{\boldsymbol{\varphi}},\widetilde{\mathbf{D}}_{\Sigma,{\mc{C}},\mc{B}} \tilde{\boldsymbol{\varphi}})_2)\dd\mu_0^{\otimes_{k_\ell}} (\tilde{\boldsymbol{\varphi}}^\ell){\dd\mu_0^+}^{\otimes_{k_h}} (\tilde{\boldsymbol{\varphi}}^h).
\end{equation}
Indeed, since $\widetilde{\mathbf{D}}_{\Sigma,{\mc{C}},\mc{B}}$  is a smoothing operator (by Lemma \ref{lemmaDSigma-D}),  we then $(\tilde{\boldsymbol{\varphi}},\widetilde{\mathbf{D}}^N_{\Sigma,{\mc{C}},\mc{B}} \tilde{\boldsymbol{\varphi}})_2\to (\tilde{\boldsymbol{\varphi}},\widetilde{\mathbf{D}}_{\Sigma,{\mc{C}},\mc{B}} \tilde{\boldsymbol{\varphi}})_2$  a.e. with respect to  $\dd\mu_0^{\otimes_{k_\ell}}  {\dd\mu_0^+}^{\otimes_{k_h}} $ as $N\to \infty$. Moreover 
\begin{equation}\label{conv_det_D_HN}
\lim_{N\to \infty}\det_{\mc{H}_N}(\mathbf{D}_{\Sigma,{\mc{C}},\mc{B}}^N(2\mathbf{D}^N_0)^{-1})\to \det(\mathbf{D}_{\Sigma,{\mc{C}},\mc{B}}(2\mathbf{D}_0)^{-1}).
\end{equation} 
Indeed, since $\Pi_N$ commutes with ${\bf D}_0$, we have 
 $\mathbf{D}_{\Sigma,{\mc{C}},\mc{B}}^N(2\mathbf{D}^N_0|_{\mc{H}_N})^{-1}\Pi_N=\Pi_N\mathbf{D}_{\Sigma,{\mc{C}},\mc{B}}(2\mathbf{D}_0)^{-1}\Pi_N$ and since $ \mathbf{D}_{\Sigma,{\mc{C}},\mc{B}}(2\mathbf{D}_0)^{-1}-{\rm Id}$ is smoothing, we have 
 $\Pi_N(\mathbf{D}_{\Sigma,{\mc{C}},\mc{B}}(2\mathbf{D}_0)^{-1}-{\rm Id})\Pi_N\to \mathbf{D}_{\Sigma,{\mc{C}},\mc{B}}(2\mathbf{D}_0)^{-1}-{\rm Id}$ in the space of trace class operators on $L^2(\T)^{k_\ell}\times L^2_{\rm even}(\T)^{k_h}$; since $K\mapsto \det_{\rm Fr}({\rm Id}+K)$ is continuous on the space of trace class operators, we obtain \eqref{conv_det_D_HN}. We can then apply the dominated convergence theorem to obtain \eqref{limitconvdominee} using the following estimate: since for all $N_1>0$ and $\epsilon>0$, there is $N_0>0$ such that   for all $N>N_0$,  $\|\widetilde{\mathbf{D}}^N_{\Sigma,{\mc{C}},\mc{B}}-\widetilde{\mathbf{D}}_{\Sigma,{\mc{C}},\mc{B}}^{N_0}\|_{H^{-N_1}\to H^{N_1}}\leq \epsilon$ (again from Lemma \ref{lemmaDSigma-D}), there is $a>0$ such that for any $\epsilon>0$ and $N_1>0$, there is $N_0>0$ such that for all $N>N_0$ and $\tilde{\boldsymbol{\varphi}}\in (H^s(\T))^{k_\ell}\times (H^s_{\rm even}(\T))^{k_h}$ for $s<0$,
 \begin{equation}\label{lowerboundDNmaptilde}
 (\tilde{\boldsymbol{\varphi}},\widetilde{\mathbf{D}}^N_{\Sigma,{\mc{C}},\mc{B}} \tilde{\boldsymbol{\varphi}})_2\geq a|\mathbf{c}|^2+(a-1)(\Pi_{N_0}\tilde{\boldsymbol{\varphi}},2\mathbf{D}\Pi_{N_0}\tilde{\boldsymbol{\varphi}})_2-\epsilon\|\tilde{\boldsymbol{\varphi}}\|_{H^{-N_1}}^2.
 \end{equation}
This completes the proof of our claim \eqref{densityDN}.   

Finally we compute the ratio of free field amplitudes, and by density it suffices to consider $\tilde{\boldsymbol{\varphi}}_j,\tilde{\boldsymbol{\varphi}}$ smooth:  
\begin{align*}
 \frac{\caA^0_{\Sigma,g}(\tilde{\boldsymbol{\varphi}}_1,\tilde{\boldsymbol{\varphi}}_2)}{\caA^0_{\Sigma_1,g_1}(\tilde{\boldsymbol{\varphi}} ,\tilde{\boldsymbol{\varphi}}_1)\caA^0_{\Sigma_2,g_2}(\tilde{\boldsymbol{\varphi}} ,\tilde{\boldsymbol{\varphi}}_2)}  =&e^{-\frac{1}{2}(\tilde{\boldsymbol{\varphi}},2\mathbf{D}\tilde{\boldsymbol{\varphi}})_2+\frac{1}{2}((\tilde{\boldsymbol{\varphi}},\tilde{\boldsymbol{\varphi}}_1),\mathbf{D}_{\Sigma_1} (\tilde{\boldsymbol{\varphi}},\tilde{\boldsymbol{\varphi}}_1))_2+\frac{1}{2}((\tilde{\boldsymbol{\varphi}},\tilde{\boldsymbol{\varphi}}_2),\mathbf{D}_{\Sigma_2} (\tilde{\boldsymbol{\varphi}},\tilde{\boldsymbol{\varphi}}_2))_2-\frac{1}{2}((\tilde{\boldsymbol{\varphi}}_1,\tilde{\boldsymbol{\varphi}}_2),\mathbf{D}_{\Sigma} (\tilde{\boldsymbol{\varphi}}_1,\tilde{\boldsymbol{\varphi}}_2))_2} 
 \\ =& e^{-\frac{1}{2}(\tilde{\boldsymbol{\varphi}},2\mathbf{D}\tilde{\boldsymbol{\varphi}})_2-A_1-A_2+A_3+A_4}
 \end{align*}
with \begin{align*}
A_1&= ((0,\tilde{\boldsymbol{\varphi}}_1),\mathbf{D}_{\Sigma_1} (\mc{R}_{\mathcal{C}\cup\mathcal{B}},0))_2+((0,\tilde{\boldsymbol{\varphi}}_2),\mathbf{D}_{\Sigma_2} (\mc{R}_{\mathcal{C}\cup\mathcal{B}},0))_2\\
A_2&=\frac{1}{2}((\mc{R}_{\mathcal{C}\cup\mathcal{B}},0),\mathbf{D}_{\Sigma_1} (\mc{R}_{\mathcal{C}\cup\mathcal{B}},0))_2+\frac{1}{2}((\mc{R}_{\mathcal{C}\cup\mathcal{B}},0),\mathbf{D}_{\Sigma_2} (\mc{R}_{\mathcal{C}\cup\mathcal{B}},0))_2\\
A_3&=\frac{1}{2}((\tilde{\boldsymbol{\varphi}},0),\mathbf{D}_{\Sigma_1} (\tilde{\boldsymbol{\varphi}},0))_2+\frac{1}{2}((\tilde{\boldsymbol{\varphi}},0),\mathbf{D}_{\Sigma_2} (\tilde{\boldsymbol{\varphi}},0))_2\\
A_4&= ((0,\tilde{\boldsymbol{\varphi}}_1),\mathbf{D}_{\Sigma_1} (\tilde{\boldsymbol{\varphi}},0))_2+((0,\tilde{\boldsymbol{\varphi}}_2),\mathbf{D}_{\Sigma_2} (\tilde{\boldsymbol{\varphi}},0))_2,
 \end{align*}
 relation that we have obtained by using the relation 
\begin{equation}\label{zc1}
((\tilde{\boldsymbol{\varphi}}_1,\tilde{\boldsymbol{\varphi}}_2),\mathbf{D}_{\Sigma} (\tilde{\boldsymbol{\varphi}}_1,\tilde{\boldsymbol{\varphi}}_2))_2=((\mc{R}_{\mathcal{C}\cup\mathcal{B}},\tilde{\boldsymbol{\varphi}}_1),\mathbf{D}_{\Sigma_1} (\mc{R}_{\mathcal{C}\cup\mathcal{B}},\tilde{\boldsymbol{\varphi}}_1))_2+((\mc{R}_{\mathcal{C}\cup\mathcal{B}},\tilde{\boldsymbol{\varphi}}_2),\mathbf{D}_{\Sigma_2} (\mc{R}_{\mathcal{C}\cup\mathcal{B}},\tilde{\boldsymbol{\varphi}}_2))_2
\end{equation}
 in the first line above and then expanded all terms to get the second line. One important point to get \eqref{zc1} is the fact that the harmonic extension $P(\tilde{\boldsymbol{\varphi}}_1,\tilde{\boldsymbol{\varphi}}_2) $ on $\Sigma$ is smooth on $\mathcal{C}\cup\mc{B}$ so that 
 \begin{equation}\label{zc0}
 -\partial_{\nu_-} P(\tilde{\boldsymbol{\varphi}}_1,\tilde{\boldsymbol{\varphi}}_2) _{|\mc{C}\cup\mc{B}}-\partial_{\nu_+} P(\tilde{\boldsymbol{\varphi}}_1,\tilde{\boldsymbol{\varphi}}_2) _{|\mc{C}\cup\mc{B}} =0.
 \end{equation}
  Now, by definition of the DN map $\mathbf{D}_{\Sigma,{\mc{C}},\mc{B}} $, we have  
\begin{equation}\label{zc2}
((\tilde{\boldsymbol{\varphi}},0),\mathbf{D}_{\Sigma_1} (\tilde{\boldsymbol{\varphi}},0))_2+((\tilde{\boldsymbol{\varphi}},0),\mathbf{D}_{\Sigma_2}(\tilde{\boldsymbol{\varphi}},0))_2=(\tilde{\boldsymbol{\varphi}},\mathbf{D}_{\Sigma,{\mc{C}},\mc{B}}\tilde{\boldsymbol{\varphi}})_2. 
\end{equation}
 Therefore $A_2=\frac{1}{2}(\mc{R}_{\mathcal{C}\cup\mathcal{B}},\mathbf{D}_{\Sigma,{\mc{C}},\mc{B}} \mc{R}_{\mathcal{C}\cup\mathcal{B}})_2$ and $A_3=\frac{1}{2}(\tilde{\boldsymbol{\varphi}},\mathbf{D}_{\Sigma,{\mc{C}},\mc{B}} \tilde{\boldsymbol{\varphi}})_2$. For $A_1$, we insert $(0,\tilde{\boldsymbol{\varphi}}_1)=(\mc{R}_{\mathcal{C}\cup\mathcal{B}},\tilde{\boldsymbol{\varphi}}_1)-(\mc{R}_{\mathcal{C}\cup\mathcal{B}},0)$, and similarly for $(0,\tilde{\boldsymbol{\varphi}}_2)$, and using that ${\bf D}_{\Sigma_i}^*={\bf D}_{\Sigma_i}$ we obtain 
\begin{align*}
&((\mc{R}_{\mathcal{C}\cup\mathcal{B}},\tilde{\boldsymbol{\varphi}}_1),\mathbf{D}_{\Sigma_1} (\mc{R}_{\mathcal{C}\cup\mathcal{B}},0))_2+((\mc{R}_{\mathcal{C}\cup\mathcal{B}},\tilde{\boldsymbol{\varphi}}_2),\mathbf{D}_{\Sigma_2} (\mc{R}_{\mathcal{C}\cup\mathcal{B}},0))_2\\
& =-(\partial_{\nu_-} P(\tilde{\boldsymbol{\varphi}}_1,\tilde{\boldsymbol{\varphi}}_2) _{|\mc{C}\cup\mc{B}}+\partial_{\nu_+} P(\tilde{\boldsymbol{\varphi}}_1,\tilde{\boldsymbol{\varphi}}_2) _{|\mc{C}\cup\mc{B}},\mc{R}_{\mathcal{C}\cup\mathcal{B}})_2=0
\end{align*}
Thus $A_1=-(\mc{R}_{\mathcal{C}\cup\mathcal{B}},\mathbf{D}_{\Sigma,{\mc{C}}} \mc{R}_{\mathcal{C}\cup\mathcal{B}})_2$ using \eqref{zc2}. The same trick applied to $A_4$ gives $A_4=-(\tilde{\boldsymbol{\varphi}},\mathbf{D}_{\Sigma,{\mc{C}},\mc{B}}\mc{R}_{\mathcal{C}\cup\mathcal{B}})_2$. Combining everything, we deduce that
\[\frac{\caA^0_{\Sigma,g}(\tilde{\boldsymbol{\varphi}}_1,\tilde{\boldsymbol{\varphi}}_2)}{\caA^0_{\Sigma_1,g_1}(\tilde{\boldsymbol{\varphi}} ,\tilde{\boldsymbol{\varphi}}_1)\caA^0_{\Sigma_2,g_2}(\tilde{\boldsymbol{\varphi}} ,\tilde{\boldsymbol{\varphi}}_2)}  
=
 \exp\Big(\frac{1}{2}(\tilde{\boldsymbol{\varphi}},\widetilde{{\bf D}}_{\Sigma,{\mc{C}},\mc{B}} \tilde{\boldsymbol{\varphi}})_2\Big)
 \exp\Big(-(\tilde{\boldsymbol{\varphi}},{\bf D}_{\Sigma,{\mc{C}},\mc{B}} \mc{R}_{\mathcal{C}\cup\mathcal{B}})_2+\frac{1}{2}(\mc{R}_{\mathcal{C}\cup\mathcal{B}},{\bf D}_{\Sigma,{\mc{C}},\mc{B}} \mc{R}_{\mathcal{C}\cup\mathcal{B}})_2\Big),\] 
which proves the lemma by using also  \eqref{densityDN} and \eqref{CameronMartin}.
\end{proof}

Combining \eqref{comput1} with Lemmas \ref{density1} and Proposition \ref{detformula} concludes the case when $\partial\Sigma_D\not=\emptyset$.

\medskip
$\bullet$ Assume now $\partial\Sigma_D  = \emptyset$: Observe first that it suffices to consider either the case when $k_\ell=1$ and $k_h=0$ or the case $k_\ell=0$ and $k_h=1$. Indeed in the general case, we can  first cut $\Sigma$ along one interior or boundary cut so as get a new surface with non trivial Dirichlet boundary and then apply the case  $\partial\Sigma_D \not = \emptyset$.    
  
So we study first the case  $k_\ell=1$ and $k_h=0$. We can assume that $\Sigma$ has a non trivial Neumann boundary because the case when $\Sigma$ is closed has been treated  in \cite{GKRV21_Segal}. We denote the glued boundary component by  $\mathcal{C}:= \mc{C}_1$.   By definition, the amplitude for $\Sigma$ is  
\begin{equation}\label{amp0}
\caA_{\Sigma,g,{\bf x},\boldsymbol{\alpha},{\bf s},\boldsymbol{\beta} }(F_1\otimes F_2)=\langle F_1\otimes F_2(\phi_g)\prod_{i=1}^m V_{\alpha_i,g}(x_i)\prod_{j=1}^{m_B} V_{\frac{\beta_{j}}{2},g}(s_{j}) \rangle_{\Sigma,g} .
\end{equation}
Let now $X_1$ and $X_2$ be two  independent mixed boundary condition GFF respectively on $\Sigma_1$ and $\Sigma_2$, with Dirichlet b.c. on $\mc{C}$ and Neumann b.c. on $\partial\Sigma$. We assume that they are both defined on $\Sigma$ by setting $X_i=0$ outside of $\Sigma_i$. Then we have the following decomposition in law (see Proposition \ref{decompGFF} item 1.b)
\begin{align*}
X_{g,N}\stackrel{\rm law}=X_1+X_2+P{\bf X}-c_g
\end{align*}
where ${\bf X}$ is the restriction of the Neumann GFF $X_{g,N}$ to the  glued boundary component  $\mathcal{C}$ expressed in parametrized coordinates, i.e. ${\bf X} = X_{g,N|_{{\caC}}}\circ \zeta_{1,1} $,   $P{\bf X}$ is its harmonic extension to $ \Sigma $ and $c_g:=\frac{1}{{\rm v}_g(\Sigma)}\int_\Sigma (X_1+X_2+P{\bf X})\,\dd {\rm v}_g$. Therefore, plugging this relation into the amplitude  \eqref{amp0},  and then shifting the $c$-integral by $c_g$, we get      
\begin{multline*}
 \caA_{\Sigma,g,{\bf x},\boldsymbol{\alpha},{\bf s},\boldsymbol{\beta}}(F_1\otimes F_2)= \frac{({\det}'(\Delta_{g,N}))^{-\frac{1}{2}}}{{\rm v}_{g}(\Sigma)^{-\frac{1}{2}}Z_{\Sigma_1,m,g_1}Z_{\Sigma_2,m,g_2}}
 \\
 \int \frac{\caA_{\Sigma_1,g_1,{\bf x}_1,\boldsymbol{\alpha}_1,{\bf s}_1,\boldsymbol{\beta}_1,\boldsymbol{\zeta}_1}(F_1,c + \boldsymbol{  \varphi})\caA_{\Sigma_2,g_2,{\bf x}_2,\boldsymbol{\alpha}_2,{\bf s}_2,\boldsymbol{\beta}_2,\boldsymbol{\zeta}_2} (F_2,c + \boldsymbol{   \varphi})}{\caA^0_{\Sigma_1,g_1}(c + \boldsymbol{  \varphi})\caA^0_{\Sigma_2,g_2}(c + \boldsymbol{   \varphi} ) }  \dd c \,  \dd \P_{{\bf X}}    (\boldsymbol{\varphi}).
\end{multline*}
 Now we make a further shift in the $c$-variable in the expression above to subtract the mean $m_{\mc{C}}({\bf X}):=\frac{1}{2\pi}\int_0^{2\pi}{\bf X}(e^{i\theta})\,\dd \theta$ to the field ${\bf X}$. As a consequence we can replace the law $ \P_{{\bf X}} $ of ${\bf X}$ in the above expression by the law $  \P_{{\bf X}-m_{\mc{C}}({\bf X})}$ of the recentered field ${\bf X}-m_{\mc{C}}({\bf X})$.
 
Now we  claim, for measurable bounded functions $F$
\begin{equation}\label{densityDNbis}
\int F( \boldsymbol{ \varphi}) \dd\P_{{\bf X}-m_{\mc{C}}({\bf X})}(\boldsymbol{ \varphi})=\frac{ \sqrt{2}}{  \det(\mathbf{D}_{\Sigma,{\mc{C}},0}(2\mathbf{D}_0)^{-1})^{-1/2} }\int F(\boldsymbol{ \varphi})\exp\Big(-\frac{1}{2}(\boldsymbol{ \varphi},\widetilde{\mathbf{D}}_{\Sigma,{\mc{C}}} \boldsymbol{ \varphi})_2\Big)   \dd \P_\T   ( \boldsymbol{ \varphi})
\end{equation}
where $\mathbf{D}_{\Sigma,{\mc{C}},0}=\mathbf{D}_{\Sigma,\mc{C}}+\Pi_0'$ (recall the notation \eqref{Pi0'}).
The proof of this claim follows the same lines as the proof of \eqref{densityDN}, using  $\mathbf{D}_{\Sigma,{\mc{C}}}G=2\pi {\rm Id}$ on the space $\{f\in H^s(\T) \, |\, \int_0^{2\pi}f(e^{i\theta})  \dd \theta=0 \}$ (see Lemma \ref{rest_DNmap}) and the following estimate, proved exactly as \eqref{lowerboundDNmaptilde}: there is $a>0$ such that for any $\epsilon>0$, there is $N_0>0$ such that 
$$(\boldsymbol{ \varphi},\widetilde{\mathbf{D}}_{\Sigma,{\mc{C}}} \boldsymbol{ \varphi})_2\geq  (a-1)(\Pi_{N_0}\boldsymbol{ \varphi},2\mathbf{D}\Pi_{N_0}\boldsymbol{ \varphi})_2-\epsilon(\boldsymbol{ \varphi},2\mathbf{D}\boldsymbol{ \varphi})_2.$$

This achieves the proof by observing that $\exp(-\frac{1}{2}(\boldsymbol{ \varphi},\widetilde{\mathbf{D}}_{\Sigma,{\mc{C}}} \boldsymbol{ \varphi})_2)=\caA^0_{\Sigma_1,g_1}(c + \boldsymbol{   \varphi})\caA^0_{\Sigma_2,g_2}(c + \boldsymbol{  \varphi} ) $ (whatever the value of $c$, since $\mathbf{D}_{\Sigma,\mc{C}}1={\bf D}1=0$) and by using the determinant formula
\[
\frac{({\det}'(\Delta_{g,N})/{\rm v}_{g}(\Sigma))^{-1/2}}{Z_{\Sigma_1,m,g_1}Z_{\Sigma_2,m,g_2}} =   \det(\mathbf{D}_{\Sigma,{\mc{C}},0}(2\mathbf{D}_0)^{-1})^{-1/2},
\]
that follows from Proposition \ref{detformula} in the case $(b'_\ell,b'_h)=(1,0)$.

The case $k_\ell=0$ and $k_h=1$ is proven exactly in the same way, with a different constant coming from Proposition \ref{detformula} in the case $(b'_\ell,b'_h)=(0,1)$.
\qed

%

\section{Semigroups of annuli and half-annuli}\label{sec:semigroup}

In \cite{GKRV21_Segal}  was constructed a semigroup by gluing annuli. This semigroup served there to find a diagonalizing basis of the Hilbert space, used to cut surfaces along loops drawn on the surface. The basis was found by diagonalizing the generator of the  semigroup of annuli, called the bulk Hamiltonian of Liouville theory\footnote{In \cite{GKRV21_Segal}, this generator was called Hamiltonian of Liouville theory but we add here the prefix {\it bulk} to differentiate it from the boundary Hamiltonian that we are going to introduce.}. We recall below the construction of this operator. In this section we will also introduce another operator to deal with boundary cuts in the surface, called boundary Hamiltonian of Liouville theory, as the generator of the semigroup obtained by gluing half-annuli.  We stress that these operators are different and so are their diagonalizing basis.

\subsection{Semigroup of annuli}  
In \cite[Section 3 and 5]{GKRV20_bootstrap}, a contraction semigroup $S(t):L^2(\R\times\Omega_{\T})\to L^2(\R\times\Omega_{\T})$ is defined 
by the following expression: for $F\in L^2(\R\times\Omega_{\T})$ 
\begin{equation}\label{FKgeneral}
S(t)F(\tilde\varphi^\ell)=e^{-\frac{Q^2t}{2}}\E \big[ F(\phi_t) e^{-\mu e^{\gamma c}\int_{e^{-t}<|z|<1} |z|^{-\gamma Q}M_\gamma (\phi,\dd z)}\big]
\end{equation}
where  $\phi_t: \theta\mapsto \phi(e^{-t+i\theta})$ with $\phi:=P_{\D}\tilde\varphi^\ell+X_{\D,D}$,  $X_{\D,D}$ is the GFF with Dirichlet condition on the unit disk 
$(\D,|dz|^2)$ and $P_{\D}\tilde\varphi^\ell$ is the harmonic extension in $\D$ of the random Fourier series $\tilde\varphi^\ell\in H^{s}(\T)$ defined in \eqref{GFFcircle0}, and the expectation is with respect to $X_{\D, D}$.
This semigroup can be written under the form 
\[S(t)=e^{-t{\bf H}}\] 
where  $\mathbf{H}$ is an unbounded self-adjoint operator on $L^2(\R\times\Omega_{\T})$, which   formally reads 
\[ \mathbf{H}={\bf H}^0+\mu e^{\gamma c}V(\varphi^\ell)\]
with
\[  {\bf H}^0:=-\frac{1}{2}\pl_c^2+\frac{1}{2}Q^2+{\bf P}   \textrm{ and } V(\varphi^\ell):=\lim_{k\to\infty}\int_0^{2\pi}e^{\gamma \varphi^{\ell,k}(\theta)-\frac{\gamma^2}{2}\E[(\varphi^{\ell,k}(\theta))^2]}\dd \theta\]
and $\varphi^{\ell,k}$ corresponds to the Fourier series \eqref{GFFcircle0} restricted to the indices $n$ with $n\not=0$ and $|n|\leq k$.
Here ${\bf P}, V$ are symmetric unbounded non-negative operators on $L^2(\Omega_{\T})$ and ${\bf P}$ has discrete spectrum given by $\N_0=\N\cup \{0\}$. Moreover ${\bf H}^0$ is the generator of $S(t)$ when we set $\mu=0$ in \eqref{FKgeneral}.  Therefore $\mathbf{H}$ can be seen as the sum of the free Hamiltonian $\mathbf{H}^0$ and a non-negative operator  $\mu e^{\gamma c}V(\varphi^\ell)$.
 
 Following \cite[subsection 3.3]{GKRV20_bootstrap}, one can also consider the generator of rotations $\mathbf{\Pi}$, defined as the generator of the strongly continuous unitary group on $L^2(\R\times \Omega_{\T})$ (for $\vartheta\in\R$)
\begin{align}\label{pi:prob}
 e^{i\vartheta \mathbf{\Pi}}F(\tilde{\varphi})=F(\tilde{\varphi}(\cdot+\vartheta)).
\end{align}
 
Actually, this semigroup can be seen as a semigroup of annuli amplitudes (with Dirichlet boundary conditions). More precisely, for $q\in \D$ with $|q|=e^{-t}$ for $t>0$, consider the flat annulus
\begin{equation}\label{defaq}
\mathbb{A}_q:=\{ z\in \C\,|\, |z|\in [e^{-t},1]\}, \quad g_{\mathbb{A}}=|dz|^2/|z|^2
\end{equation}
with boundary parametrization $\boldsymbol{\zeta}_q=(\zeta_1,\zeta_2)$ given by $\zeta_2(e^{i\theta})=e^{i\theta}$ and $\zeta_1(e^{i\theta})=qe^{i\theta}$. For $q'\in \D$, denote by $q\mathbb{A}_{q'}:=\{ z\in \C\,|\, |z|\in [|qq'|,|q|]\}$ with metric $g_{\mathbb{A}}$ and boundary parametrization $\boldsymbol{\zeta}_{q,q'}'=(\zeta'_1,\zeta'_2)$ given by 
$\zeta'_2(e^{i\theta})=q'\zeta_2(e^{i\theta})$ and 
$\zeta'_1(e^{i\theta})=q'\zeta_1(e^{i\theta})$. Observe that one can glue $\mathbb{A}_{q}$ with $q\mathbb{A}_{q'}$ along the boundary circle $\{|z|=|q|\}$, the result becomes $\mathbb{A}_{q}\# (q\mathbb{A}_{q'})= \mathbb{A}_{qq'}$. 
Since $(q\mathbb{A}_{q'},g_{\mathbb{A}})$ is isometric to $(\mathbb{A}_{q'},g_{\mathbb{A}})$, we can use Proposition \ref{Weyl}  to express its amplitude as $\mc{A}_{\mathbb{A}_{q'},g_{\mathbb{A}},\boldsymbol{\zeta}_{q'}}$, 
 and using the  gluing property, see Proposition \ref{glueampli}, one obtains that the amplitudes of these annuli satisfy  
\[\mc{A}_{\mathbb{A}_{qq'},g_{\mathbb{A}},\boldsymbol{\zeta}_{qq'}}(\tilde\varphi_1^\ell,\tilde\varphi_2^\ell)=\frac{1}{\sqrt{2}\pi}\int \mc{A}_{\mathbb{A}_q,g_{\mathbb{A}},\boldsymbol{\zeta}_q}(\tilde\varphi_1^\ell,\tilde\varphi^\ell)
\mc{A}_{\mathbb{A}_{q'},g_{\mathbb{A}},\boldsymbol{\zeta}_{q'}}(\tilde\varphi^\ell,\tilde\varphi_2^\ell)d\mu_{0}(\tilde\varphi^\ell).\]
We can then view   $(\sqrt{2}\pi)^{-1/2}\mc{A}_{\mathbb{A}_{q},g_{\mathbb{A}},\boldsymbol{\zeta}_q}$ as the  integral kernel of an operator, which can be related to 
 $e^{-t{\bf H}}$ and $e^{i\vartheta \mathbf{\Pi}}$ (see \cite[section 6]{GKRV21_Segal}).
\begin{proposition} \label{prop:annulussimple}
For $q\in \D$, the annuli amplitudes in the metric $g_\mathbb{A}$ satisfy for $F,F'\in L^2(\R\times\Omega_\T)$  
\begin{align}
\int F(\tilde\varphi^\ell_1)\mathcal{A}_{\mathbb{A}_{|q|},g_\mathbb{A},\boldsymbol{\zeta}_q}(\tilde\varphi^\ell_1 ,\tilde\varphi^\ell_2) F'(\tilde\varphi^\ell_2) \dd\mu_0(\tilde\varphi^\ell_1)\dd\mu_0(\tilde\varphi^\ell_2) = &\sqrt{2}\pi   |q|^{-\frac{\mathbf{c}_L}{12}}\langle e^{\log |q|\mathbf{H}}F,\bar F'\rangle_{\mc{H}} \label{ident1annulus}\\
\int F(\tilde\varphi^\ell_1)\mathcal{A}_{\mathbb{A}_{q},g_\mathbb{A},\boldsymbol{\zeta}_q}(\tilde\varphi^\ell_1,\tilde\varphi^\ell_2) F'(\tilde\varphi^\ell_2) \dd\mu_0(\tilde\varphi^\ell_1)\dd\mu_0(\tilde\varphi^\ell_2) =& \sqrt{2}\pi |q|^{-\frac{\mathbf{c}_L}{12}}\langle  e^{i \arg(q)\boldsymbol{\Pi}} e^{\log |q|\mathbf{H}}F,\bar F'\rangle_{\mc{H}}\label{ident2annulus}
\end{align}
with $\mathbf{c}_L=1+6Q^2$ the central charge. 
\end{proposition}

  \subsection{Semigroup of half-annuli and boundary Hamiltonian }\label{s:semigroup_half} 
 In this subsection, following essentially the case of interior cuts and the semigroup of annuli, we  construct a semigroup of half-annuli in order to define the   boundary Hamiltonian of LCFT. Its spectral resolution will serve to  establish the conformal bootstrap for surfaces with a boundary. We stress that these two operators are by no means related to each other.

\subsubsection{Free boundary Hamiltonian and quadratic form}\label{sub:bfreehamiltonian}
Our first task is to construct the quadratic form of the semigroup of half-annuli and we follow again \cite[section 4.1]{GKRV20_bootstrap}.
The Hilbert space $L^2(\Omega_{\T^+},\P_{\T^+})$ (denoted from now on by  $L^2_+(\Omega_{\T^+})$) has the structure of Fock space. Let $\caP_+\subset L^2_{+}(\Omega_{\T^+})$ (resp. $ \mathcal{S}_+\subset L^2_+(\Omega_{\T^+})$) be the linear span of the functions of the form $F(x_1,\dots, x_N)$ for some $N \geq 1$ where $F$ is a polynomial on $\R^{N}$ (resp.   $F\in C^\infty(\R^N)$ with at most polynomial growth at infinity for $F$ and its derivatives). Obviously $\mathcal{P}_+\subset\mathcal{S}_+$ and they are both dense in $L^2_+(\Omega_{\T^+})$.

On $\mathcal{S}_+$ we define the annihilation and  creation operators (which is adjoint to annihilation) for $n\geq 1$
\begin{align}\label{crea}
\mathbf{X}_n&=\partial_{x_n},\ \ \ \mathbf{X}_n^\ast=-\partial_{x_n}+x_n.
\end{align}
They  form a representation of the algebra of canonical commutation relations on $\mathcal{S}_+$:
\begin{align}\label{ccr}
[\mathbf{X}_n,\mathbf{X}_m^\ast]=\delta_{nm} .
\end{align} 
We then consider  the  operator $ \mathbf{P}_+$  on  $\mathcal{S}_+$ defined by 
 \begin{align}\label{hdefi}
 \mathbf{P}_+=\sum_{n=1}^\infty n\mathbf{X}_n^\ast \mathbf{X}_n 
\end{align}
(only finite number of terms in the sum contributes when acting on $\mathcal{S}_+$) and extends uniquely to an unbounded 
self-adjoint positive operator on $L^2_+(\Omega_{\T^+})$: this follows from the fact that we can find a complete system of eigenfunctions in $\caP_+$, as described now.  Let $\mathcal{N}$ be the set of  non-negative integer valued sequences with only a finite number of non null integers, namely ${\bf k}=(k_1,k_2,\dots)\in \mathcal{N}$ iff ${\bf k}\in \N^{\N_+}$ and $k_n=0$ for all $n$ large enough. For   $\bf k  \in\mathcal{N}$ define the polynomials (here $1\in  L^2_+(\Omega_{\T^+})$  is the constant function)
\begin{align}\label{fbasishermite}
\hat{\psi}_{{\bf k}}=\prod_{n\geq 1} ( \mathbf{X}_n^\ast)^{k_n} 1 \in \caP_+.
\end{align}
Equivalently, $\hat{\psi}_{{\bf k}}= \prod_n {\rm He}_{k_n}(x_n) $ where $({\rm He}_k)_{k \geq 0}$ are the standard Hermite polynomials\footnote{More specifically, we consider the  Hermite polynomials  given by ${\rm He}_k(x)=(-1)^k e^{\frac{x^2}{2}}  \frac{d^k}{dx^k} (e^{-\frac{x^2}{2}}) $  where the coefficient of $x^k$ is $1$.}. Then, using \eqref{ccr}, one checks that these are eigenstates of $\bf P_+$:
 \begin{align}\label{fbasis2}
 \mathbf{P}_+\hat{\psi}_{{\bf k}}=|{\bf k}| 
 \hat{\psi}_{{\bf k}}
\end{align}
where we use the notations
\begin{equation}\label{firstlength}
|{\bf k}|:=\sum_{n=1}^\infty nk_n,
\end{equation}
 for ${\bf k}\in\caN$.
The family 
\begin{equation}\label{BONfamily}
\{\psi_{{\bf k}}=\hat{\psi}_{{\bf k}}/\|\hat{\psi}_{{\bf k}}\|_{L^2_+(\Omega_{\T^+})}\}
\end{equation}
(where $\|\cdot \|_{L^2_+(\Omega_{\T^+})}$ is the standard norm in $L^2_+(\Omega_{\T^+})$) forms an orthonormal basis of $L^2_+(\Omega_{\T^+})$.  Let 
\begin{equation}
	E_k:=\{F\in L^2(\Omega_{\T^+})\,|\, 1_{[0,\la_k]}(\mathbf{P}_+)F=F\}=\bigoplus_{j=0}^{k}\ker(\mathbf{P}_+-\la_j)
\end{equation}
be the sum of eigenspaces with eigenvalues of ${\bf P}_+$ less or equal to $\la_k$ and $\Pi_{k}: L^2(\Omega)\to E_k$ the orthogonal projection.

Introduce the bilinear form (with associated quadratic form still denoted by $\mathcal{Q}_{+,0}$)   
\begin{equation}\label{defQ0}
\forall u,v\in \mathcal{C}_+,\quad \mathcal{Q}_{+,0}(u,v):= \langle \partial_cu,\partial_c v\rangle_{ \mathcal{H}_+}  +\tfrac{Q^2}{4}\langle  u,v\rangle_{\mathcal{H}_+}+ \langle  \mathbf{P}_+u,v\rangle_{\mathcal{H}_+}  
\end{equation}
with
\begin{equation}\label{core}
\mathcal{C}_+=\mathrm{Span}\{ \psi(c)F\,|\,\psi\in C_c^\infty(\R)\text{ and }F\in\mathcal{S}_+ \}.
\end{equation}
We claim (adapt \cite[Prop. 4.3]{GKRV20_bootstrap} for the proof)
\begin{proposition}\label{FQ0:GFF}
The quadratic form \eqref{defQ0} is closable (and we still denote its closure by $\mathcal{Q}_{+,0}$ with domain $\mathcal{D}(\mathcal{Q}_{+,0})$) and lower semibounded: $\mathcal{Q}_{+,0}(u)\geq \frac{Q^2}{4}\|u\|_{\mc{H}_+}^2$.
It determines uniquely a self-adjoint operator $\mathbf{H}^0_+ $, called the \emph{Friedrichs extension},  with domain denoted by $\mc{D}(\mathbf{H}^0_+)$ such that:
$$\mc{D}(\mathbf{H}^0_+ )=\{u\in \mathcal{D}(\mathcal{Q}_{+,0})\, |\, \exists C>0,\forall v\in \mathcal{D}(\mathcal{Q}_{+,0}),\,\,\, |\mc{Q}_{+,0}(u,v)|\leq C\|v\|_{\mc{H}_+}\}$$
and for $u\in \mc{D}(\mathbf{H}^0_+)$, $\mathbf{H}^0_+  u$ is the unique element in $\mc{H}_+$  satisfying
$$\mc{Q}_{+,0}(u,v)=\langle \mathbf{H}^0_+  u,v\rangle_{\mc{H}_+} .$$
\end{proposition}

If we let $\mc{D}'(\mc{Q}_{+,0})$ be the dual to $\mc{D}(\mc{Q}_{+,0})$ (i.e. the space of bounded\footnote{Recall that in this context a linear functional $\ell$ is bounded if for all $u \in \mc{D}(\mc{Q}_{+,0})$ one has $|\ell(u)| \leq C \mathcal{Q}_{+,0}(u)^{\frac{1}{2}} $ for some $C>0$.} linear functionals on $\mc{D}(\mc{Q}_{+,0})$), the injection $ \mc{H}_+\subset \mc{D}'(\mc{Q}_{+,0})$ is continuous and the operator ${\bf H}^0_+$ can be extended as a bounded isomorphism 
\[{\bf H}^0_+:\mc{D}(\mc{Q}_{+,0})\to \mc{D}'(\mc{Q}_{+,0}).\] 
We also have $\mc{D}({\bf H}_{+,0})=\{ u\in\mc{D}(\mc{Q}_{+,0})\,|\, {\bf H}^0_+u\in \mc{H}_+\}$ and $({\bf H}^0_+)^{-1}:\mc{H}_+\to \mc{D}({\bf H}^0_+)$ is bounded. Furthermore, by the spectral theorem, it generates a strongly continuous contraction semigroup of self-adjoint operators $(e^{-t \mathbf{H}^0_+ } )_{t\geq 0}$ on $\mc{H}_+$.    

\subsubsection{Regularized boundary Hamiltonian and quadratic form}\label{sub:regbhamiltonian}
Now we need to plug the potentials in the quadratic form $\mc{Q}_{+,0}$. For this, since the potentials are highly singular, we need a regularization. Write $\varphi^{h,k}$ for the truncation of the series  \eqref{GFFcircle+} to the modes of the field $\varphi^h$ (recall that $\tilde{\varphi}^{h}=c+\varphi^h$) with indices $|n|\leq k$,  still with $n\not=0$. Then  we introduce the regularized potentials
  $V^{(k)}_+,L^{(k)},R^{(k)} :H^s_{\rm even}(\T)\to \R^+$ given by 
\begin{align}\label{defVk}
V^{(k)}_+(\varphi)=& \int_{\T^+} e^{\gamma  \varphi^{h,k}(\theta)-\frac{\gamma^2}{2}\E[  \varphi^{h,k}(\theta)^2]}\Big(\frac{1}{2|\sin(\theta)|+\frac{1}{k}}\Big)^{\frac{\gamma^2}{2}}\dd\theta
\\
R^{(k)}(\varphi)= &  e^{\frac{\gamma}{2} \varphi^{h,k}(0)-\frac{\gamma^2}{8} \mathbb{E}[( \varphi^{h,k}(0))^2]}
\\
L^{(k)}(\varphi)= & e^{\frac{\gamma}{2} \varphi^{h,k}(\pi)-\frac{\gamma^2}{8}\mathbb{E}[ \varphi^{h,k}(\pi)^2]}.
\end{align}
  For $k\geq 1$ and $\boldsymbol{\mu}=(\mu,\mu_L,\mu_R)$ with $\mu,\mu_L,\mu_R\geq 0$, we introduce the bilinear form (with associated quadratic form still denoted by $\mathcal{Q}^{(k)}_+$)    
\begin{equation}\label{defQn}
\mathcal{Q}^{(k)}_+(u,v):= \mathcal{Q}_{+,0}(u,v)+ \mu\langle    e^{\gamma c}V^{(k)}_+u,v\rangle_{\mc{H}_+}+\mu_L\langle    e^{\frac{\gamma}{2} c}L^{(k)} u,v\rangle_{\mc{H}_+}+\mu_R\langle    e^{\frac{\gamma}{2} c}R^{(k)} u,v\rangle_{\mc{H}_+}.
\end{equation}
Here   $u,v$ belong to the domain $\mathcal{D}(\mathcal{Q}^{(k)}_+)$ of the quadratic form, namely the completion  for the $\mathcal{Q}^{(k)}_+$-norm in $\mc{H}_+$  of the space $\mathcal{C}_+$ defined by  \eqref{core}. Also,  $ \mathcal{D}(\mathcal{Q}^{(k)}_+)$ embeds continuously and injectively in $\mc{H}_+$. Moreover, since $V^{(k)}_+,L^{(k)},R^{(k)}\geq 0$ it is clearly lower semi-bounded $\mathcal{Q}^{(k)}_+(u)\geq Q^2\|u\|_{\mc{H}_+}^2/4$ so that the construction of the Friedrichs extension then follows from \cite[Theorem 8.15]{Reed-Simon}. It   determines uniquely a self-adjoint operator $\mathbf{H}^{(k)}_+$, called the \emph{Friedrichs extension},  with domain denoted by $\mc{D}(\mathbf{H}^{(k)}_+)$ such that:
\[
\mc{D}(\mathbf{H}^{(k)}_+ )=\{u\in \mathcal{D}(\mathcal{Q}^{(k)}_+)\, |\,  \exists C>0,\forall v\in \mathcal{D}(\mathcal{Q}^{(k)}_+),\,\,\, \mc{Q}^{(k)}_+(u,v)\leq C\|v\|_{\mc{H}_+}\}
\]
and for $u\in \mc{D}(\mathbf{H}^{(k)}_+)$, $\mathbf{H}^{(k)}_+  u$ is the unique element in $L^2(\R\times \Omega_{\T^+})$  satisfying
\[
\mc{Q}^{(k)}_+(u,v)=\langle \mathbf{H}^{(k)}_+  u\,, \, v\rangle_{\mc{H}_+} .
\]
We also denote by  $\mathbf{R}_{\lambda}^{(k),+}$ the associated resolvent family.  The smoothness of the regularized potential allows us to identify the semigroup associated with this quadratic form as follows. 

Let $\D^+$ be the half disk equipped with the Euclidean metric $g=|dz|^2$.  We impose  Dirichlet boundary condition on  the half circle $\T^+$ and Neumann boundary condition on the real interval $[-1,1]$, namely $\pl \D^+_N=[-1,1]$ and  $\partial \D^+_D:=\T^+$. Consider the field $\phi=X_{\D^+,m}+P\tilde \varphi^h$, where $X_{\D^+,m}$ is the GFF\footnote{We omit the $g$ index  on the GFF in this section to simplify notations, i.e. we write $X_{\D^+,m}$ instead of $X_{\D^+,g,m}$.} with mixed boundary conditions on $\partial \D^+$ and  $P  \tilde \varphi^h$ is the harmonic extension to $\D^+$ of the boundary field $ \tilde \varphi^h\in H^s_{\rm even}(\T)$ prescribed on $\partial\D^+_D=\T^+$ with the condition $\partial_\nu^gP \tilde \varphi^h =0$ on $\partial\D^+_N=[-1,1]$. Recall that $X_{\D^+,m}$ and $\tilde{\varphi}^h=c+\varphi^h$ are assumed to be independent. 

Arguing exactly as in \cite[Prop. 5.3]{GKRV20_bootstrap}, we obtain:
\begin{proposition}\label{prop:fkhk}
The strongly continuous contraction semigroup $(e^{-t\mathbf{H}^{(k)}_+})_{t\geq 0}$ of self-adjoint operators on $\mc{H}_+$   obeys the Feynman-Kac formula $\forall F\in \mc{H}_+$
\begin{equation}\label{fkformulaforhk}
  e^{-t\mathbf{H}^{(k)}_+}F(\tilde{\varphi}^h)=e^{-\frac{Q^2t}{4}}\E_{\tilde{\varphi}^h}\big[ F(\tilde\phi_t)e^{-\mu\int_0^t e^{\gamma  \phi_{0,s} }V^{(k)}_+(\phi_s)\dd s-\mu_L\int_0^t  e^{\frac{\gamma}{2}  \phi_{0,s} }L^{(k)}(\phi_s)\,\dd s-\mu_R\int_0^t e^{\frac{\gamma}{2}  \phi_{0,s} }R^{(k)}(\phi_s)\,\dd s}\big],
\end{equation}
where we have decomposed the field $\tilde\phi_s:=X_{\D^+,m}(e^{-s}\cdot)+P\tilde{\varphi}^h(e^{-s}\cdot) $ along its harmonics i.e. $\phi_s(\theta):=\phi_{0,s}+ \phi_s(\theta)$
with $\phi_s(\theta):=\sum_{n\not=0}\phi_{s,n}e^{in\theta}$, in such a way that $\phi_{0,s}$ stands for the zero mode. 
\end{proposition}

Now our objective is to construct  the quadratic form corresponding to the limit $k\to\infty$ of the quadratic forms $(\mathcal{Q}^{(k)}_+,\mathcal{D}(\mathcal{Q}^{(k)}_+))$ and identify its semigroup. The first (and easiest) step concerns the identification of the limiting semigroup, which we describe now.

\subsubsection{Semigroup of half-annuli}\label{sub:bsemigroup}

Let $\phi=X_{g,m}+P\tilde{\varphi}^h$  be the GFF in the half-disk $(\D^+,dz^2)$ with conditional value $\tilde{\varphi}^h=c+\varphi^h$ on $\T^+$ and Neumann condition on $[-1,1]=\pl \D^+_N$ where $\varphi^h\in H^{-s}_{\rm even}(\T)$. 
Write   $\phi_t: \theta \in [0,\pi]\mapsto \phi(e^{-t+i\theta})$  and, similarly to \eqref{FKgeneral}, consider for $t>0$ the operator 
\begin{equation}\label{thehalfannulisemigroup}
	S_{\boldsymbol{\mu}}(t) F(\tilde{\varphi}^h)=e^{-\frac{Q^2t}{4}}   \E_{\varphi^h} \big[ F(\phi_t) e^{-\mu  \int_{\A^+_{e^{-t}}} |x|^{-\gamma Q}M^g_\gamma (\phi,\dd x) -\mu_L  \int_{\mathbb{I}^L_t} |x|^{-\gamma Q/2}M^g_{\gamma,\partial} (\phi,\dd x) -\mu_R \int_{\mathbb{I}^R_t} |x|^{-\gamma Q/2}M^g_{\gamma,\partial} (\phi,\dd x)}  \big ]
\end{equation}
where  the vector $\boldsymbol{\mu}$ collects three cosmological constants $\boldsymbol{\mu}=(\mu,\mu_L,\mu_R)$ with $\mu,\mu_L,\mu_R\geq 0$, $\A^+_q$ is the half annulus $\A^+_q:=\{z\in\C\,|\, q<|z|<1\,,\, {\rm Im}(z)>0\}$ for $0<q<1$, $\mathbb{I}^R_t,\mathbb{I}^L_t$ are respectively the intervals on the real line $[1-e^{-t},1]$ and $[-1,-1+e^{-t}]$ (see Figure \ref{demianneau}), and $F$ belongs to $\mc{H}_+$ (recall the setup introduced in Subsection \ref{sub:hilbert}). Here $ \E_{\varphi^h} $ stands for expectation conditionally on the field $\varphi^h$. The field $X_{\D^+,m}$ has covariance    given by
\[\E[X_{\D^+,m}(z)X_{\D^+,m}(z')]=\log\frac{|1-z\bar{z}'||1-zz'|}{|z-z'||z-\bar{z}'|}=  G_{\D^+,m}(z,z')=G_{\D,D}(z,z')+G_{\D,D}(z,\bar{z}')\]
so that the bulk measure  in  \eqref{thehalfannulisemigroup} reads
\[M^g_\gamma (\phi,\dd x)=  \lim_{\eps\to 0}e^{\gamma P  \tilde \varphi^h(x)}\Big(\frac{|1-x^2||1-|x|^2|}{|x-\bar x|}\Big)^{\frac{\gamma^2}{2}} e^{\gamma X_{\D^+,m,\eps}(x)-\frac{\gamma^2}{2}\E[ X_{\D^+,m,\eps}(x)^2]}\,\dd x\]
and the boundary measure  reads
$$ M^g_{\gamma,\partial} (\phi,\dd s)= \lim_{\eps\to 0} e^ {\frac{\gamma}{2} P  \tilde \varphi^h(s)} (1-s^2)^{\frac{\gamma^2}{4}} e^{\frac{\gamma}{2} X_{\D^+,m,\eps}(s)-\frac{\gamma^2}{8}\E[ X_{\D^+,m,\eps}(s)^2]}\,\dd \ell(s).$$
Proceeding exactly as in \cite[subsections 3.3 and 5.1]{GKRV20_bootstrap}, we prove that $(S_{\boldsymbol{\mu}}(t))_{t\geq 0}$ is a strongly continuous contraction semigroup of self-adjoint operators on $\mc{H}_+$.  By the Hille-Yosida theorem, this semigroup can be written under the form
$$S_{\boldsymbol{\mu}}(t)=e^{-t\mathbf{H}_{\boldsymbol{\mu}}} $$
where $\mathbf{H}_{\boldsymbol{\mu}}$ is a nonnegative self-adjoint operator with domain $\mc{D}(\mathbf{H}_{\boldsymbol{\mu}})\subset \mc{H}_+$, consisting of those $F$ such that the limit $\lim_{t\to 0}\tfrac{1}{t}(e^{-t\mathbf{H}_{\boldsymbol{\mu}}} -{\rm I})F$ exists in $ \mc{H}_+$.

Now we explain how this semigroup relates to    the half-annulus amplitude. The half-annulus amplitude is defined as $\caA_{\A_t^+,g_\A,\boldsymbol{\zeta}_t^+}(\tilde{\varphi}_1^{h},\tilde{\varphi}_2^{h})$,  with boundary parametrization $\boldsymbol{\zeta}^+_t=(\zeta^+_1,\zeta^+_2)$ given by $\zeta^+_2(e^{i\theta})=e^{i\theta}$ and $\zeta^+_1(e^{i\theta})=e^{-t+i\theta}$, both for $\theta\in [0,\pi]$, see Figure \ref{demianneau}.

 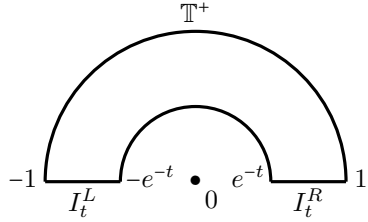
\begin{figure}[h] 
\begin{tikzpicture}[scale=1]
\draw[very thick] (1,0) arc (0:180:1);
\draw[very thick] (2,0) arc (0:180:2);
\draw[very thick] (-2,0) -- (-1,0);
\draw[very thick] (1,0) -- (2,0);
 \draw (-1.5,-0.6) node[above]{$I_t^L$} ;
 \draw (1.5,-0.6) node[above]{$I_t^R$} ;
  \draw (-2.3,-0.2) node[above]{$-1$} ;  
  \draw (-0.6,-0.2) node[above]{$-e^{-t}$} ;
   \draw (2.2,-0.2) node[above]{$1$} ;  
  \draw (0.7,-0.2) node[above]{$e^{-t}$} ;
\draw (0,0) node[below right] {$0$} node{$\bullet$};
 \draw (0,2) node[above]{$\T^+$} ;
\end{tikzpicture}
\caption{Notation for half-annulus $\A^+_{e^{-t}}$.}\label{demianneau}
\end{figure}

\begin{proposition}\label{hasemigroup}
	For  $0<q<1$, the half-annulus amplitude, viewed as operator on $\mc{H}$, satisfies 
	\begin{align}
		\caA_{\A^+_{q},g_\A,\boldsymbol{\zeta}_q^+} = (8\pi^3)^{\frac{1}{4}} q^{-\frac{\mathbf{c}_L}{24}} e^{(\log q)\mathbf{H}_{\boldsymbol{\mu}}}
	\end{align}
	with $\mathbf{c}_L=1+6Q^2$ the central charge. 
\end{proposition}
\begin{proof}
	The main thing to check is that  $Z_{\A^+_q,m,g_\A}\caA^0_{\A_q^+,g_\A,\boldsymbol{\zeta}_q^+}= (8\pi^3)^{\frac{1}{4}}q^{-\frac{\mathbf{c}_L}{24}} e^{(\log q)\mathbf{H}^0_+}$, since then the proof follows the same lines as \cite[Proposition 6.1]{GKRV21_Segal}.  Let $\tilde\varphi^h_1(\theta)=c_1+\sum_{n\neq 0}^\infty \frac{\sqrt{2}x_{1,|n|}}{2\sqrt{|n|}}e^{in\theta}$ and  $\tilde\varphi^h_2(\theta)=c_2+\sum_{n\neq 0}^\infty \frac{\sqrt{2}x_{2,|n|}}{2\sqrt{|n|}}e^{in\theta}$. The harmonic extension in $\A_q^+$ with boundary values $\tilde\varphi^h_1$ and $\tilde\varphi^h_2$ on $|z|=1$ and $|z|=q$ respectively is
\[ P\tilde{\boldsymbol{\varphi}}^h(z)= c_1+\frac{c_2-c_1}{\log q}\log|z|+\sum_{n\geq 1}\frac{(z^n+\bar{z}^n)}{q^n-q^{-n}}
\frac{(x_{2,n}-x_{1,n}q^{-n})}{\sqrt{2n}}+\frac{(z^{-n}+\bar{z}^{-n})}{q^n-q^{-n}}\frac{(x_{1,n}q^n-x_{2,n})}{\sqrt{2n}}.\]
We thus compute  (recall  \eqref{amplifree} and Definition~\ref{unpairing})
\[\caA^0_{\A_q^+,g_\A,\boldsymbol{\zeta}_q^+}(\tilde{\varphi}_1^{h},\tilde{\varphi}_2^{h})=\exp\Big(\frac{(c_1-c_2)^2}{4\log q}-\sum_{n=1}^\infty \Big(\frac{(x_{2,n}-q^{n} x_{1,n})^2}{2(1-q^{2n})}-\frac{x_{2,n}^2}{2}\Big)\Big).\]
By definition, $\mathbf{H}^0_+=-\partial_c^2+\frac{Q^2}{4}+\sum_{n=1}^\infty n\partial_{x_n}^*\partial_{x_n}$. The semigroup $e^{-t\partial_{x_n}^*\partial_{x_n}}$ associated to the one dimensional Ornstein-Uhlenbeck operator with respect to standard Gaussian measure has integral kernel 
\[ e^{-t\partial_{x_n}^* \partial_{x_n}}(x_n, x_n^{\prime})=\frac{1}{\sqrt{1-e^{-2 t}}} \exp \Big(\frac{(x_n^{\prime}-e^{-t} x_n)^2}{2(1-e^{-2 t})}-\frac{x_n^2}{2}\Big)\]
	and operator $e^{-t\partial_c^2}$ has the kernel $\frac{1}{\sqrt{4\pi t}}e^{-\frac{(c-c')^2}{4t}}$. 
	Combining with Prop~\ref{det of half annulus}, the claim follows.
\end{proof}

A straightforward observation is then the convergence of the semigroup $e^{-t\mathbf{H}^{(k)}_+}$ towards $S_{\boldsymbol{\mu}}(t) $.  This follows from the Feynman-Kac representations \eqref{fkformulaforhk} and \eqref{thehalfannulisemigroup} of the semigroups, as well as the fact that the 3 potentials in \eqref{fkformulaforhk} are regularisations respectively of the 3 GMC measures appearing in \eqref{thehalfannulisemigroup}, and thus the three regularized GMC converge almost surely towards the corresponding limiting GMC. We state this remark as a lemma:

\begin{lemma}\label{prop:semigrouplimit}
For $F\in\mc{C}_+$, we have the convergence  in $\mc{H}_+$ of the sequence $( e^{-t\mathbf{H}^{(k)}_+}F)_k$ towards   $  S_{\boldsymbol{\mu}}(t) F$. 
\end{lemma}

\subsubsection{Quadratic form of the Liouville boundary Hamiltonian}\label{sub:quadbsemigroup}

Here we have to stress a technical issue compared to \cite{GKRV20_bootstrap}. The term $V^{(k)}_+$ has a singularity coming from the $\sin$ when $k\to \infty$. This singularity is not integrable when $\gamma\in[\sqrt{2},2[$ so that the limiting quadratic form, if exists, does not contain smooth functions, namely $\mc{C}_+$. This is a   huge technical issue in order to construct later the eigenstates of the boundary Hamiltonian. There are two ways of getting around this issue: 
\begin{equation}\label{restriction}
\text{we assume } \quad \gamma\in (0,\sqrt{2}) \qquad \text{or }\qquad  \mu=0 \text{ and }\gamma\in (0,2) .
\end{equation}
We will work under these assumptions in the following. Of course, in the case $\mu=0$, the interesting situation is $\mu_L+\mu_R>0$, otherwise it is just the free field theory.
For $u,v\in \mathcal{C}_+$, we define 
\begin{equation}\label{defQ2}
\mathcal{Q}_+(u,v):=\lim_{k\to\infty}\mathcal{Q}^{(k)}_+(u,v).
\end{equation}
The existence of the limit is a non trivial statement, which relies on the analysis of each term appearing in \eqref{defQn}. The term $\mathcal{Q}_{+,0}$
is straightforward to deal with. What is not obvious is the convergence of the three other terms, corresponding to the potentials. Concerning the first one, we prove in Appendix  \ref{GMC} that when $\gamma\in (0,\sqrt{2})$,  $(V^{(k)}_+)_k$ converges towards 
\begin{equation}\label{defV+}
 V_{+}(\varphi^h)=  \int_{\T^+} e^{\gamma \varphi^{h}(\theta)-\frac{\gamma^2}{2}\E[ \varphi^{h}(\theta)^2]}\Big(\frac{1}{2|\sin(\theta)|}\Big)^{\frac{\gamma^2}{2}}\dd\theta
\end{equation} 
 in $L^p(\Omega_{\T^+})$ for $p<\frac{1}{4}+\sqrt{\frac{1}{16}+\frac{1}{\gamma^2}}$ (note that the rhs of this inequality is larger than 1 when $\gamma^2<2$). Here we have made a slight abuse of notation by writing $e^{\gamma \varphi^{h}(\theta)-\frac{\gamma^2}{2}\E[ \varphi^{h}(\theta)^2]}d\theta$ for the GMC measure on $\T^+$ associated to $\varphi^h$. 
 This entails that the term $\langle    e^{\gamma c}V^{(k)}_+u,v\rangle_{\mc{H}_+}$ in  \eqref{defQn} converges towards $\langle    e^{\gamma c}V_+u,v\rangle_{\mc{H}_+}$
as long as $u\bar{v}\in C^\infty_c(\R;L^q(\Omega_{\T^+}))$ with  $1/q+1/p=1$, which holds when $u,v\in\mc{C}_+$. For the remaining two boundary terms, we follow \cite[section 5.2]{GKRV20_bootstrap} and argue that existence of the limit   is guaranteed by the Cameron-Martin transform. Indeed, fix $u=u(x_1, \dots,x_n)$ and $v=v(x_1,\dots,x_n)$ in $\mc{C}_+$. Then for $k\geq n$ the terms involving $L^{(k)}$, $R^{(k)}$ in $\mathcal{Q}^{(k)}_+$ can be rewritten, using Cameron-Martin,  as 
\begin{equation}\label{girsuk}
\langle e^{\frac{\gamma}{2} c}R^{(k)}u,v \rangle_{\mc{H}_+}= \langle e^{\frac{\gamma}{2} c} u_{\rm right},v_{\rm right} \rangle_{\mc{H}_+} ,\qquad \langle e^{\frac{\gamma}{2} c}L^{(k)}u,v \rangle_{\mc{H}_+}= \langle e^{\frac{\gamma}{2} c} u_{\rm left},v_{\rm left} \rangle_{\mc{H}_+} 
\end{equation}
where the functions $u_{\rm right}$, $u_{\rm left}$ (and similarly for $v$) are defined by
\begin{align*}
& u_{\rm right}( c, x_1, \dots,x_n ):=u \Big(c,x_1+\gamma  \sqrt{2}, \dots,x_n+\gamma  \sqrt{2/n}  \Big),\\ 
& u_{\rm left}( c, x_1, \dots,x_n ):=u \Big(c,x_1-\gamma  \sqrt{2}, \dots,x_n+\gamma (-1)^n \sqrt{2/n}  \Big).
\end{align*}
Hence the terms $\langle e^{\frac{\gamma}{2} c}R^{(k)}u,v \rangle_{\mc{H}_+}$, $\langle e^{\frac{\gamma}{2} c}L^{(k)}u,v \rangle_{\mc{H}_+}$ do not depend on $k\geq n$, and we get
\begin{equation}
\mathcal{Q}_+(u,v)=\mathcal{Q}_{+,0}(u,v)+ \mu\langle    e^{\gamma c}V_+u,v\rangle_{\mc{H}_+}+\mu_L\langle e^{\frac{\gamma}{2} c} u_{\rm left},v_{\rm left} \rangle_{\mc{H}_+} +\mu_R \langle e^{\frac{\gamma}{2} c} u_{\rm right},v_{\rm right} \rangle_{\mc{H}_+} .
\end{equation}
We stress that our argument crashes from here in the case  $ \gamma\geq \sqrt{2}$ because $\lim_{k\to\infty}\mathcal{Q}_+(u,u)=+\infty$ then.

Let us denote by $\mathbf{R}_{\boldsymbol{\mu}}(\lambda)$ the resolvent family   associated with  the Feynman-Kac semigroup $S_{\boldsymbol{\mu}}(t)$. Let   $\mathcal{Q}_{\boldsymbol{\mu}}$ denote the associated quadratic form  with domain $\{u\in \mc{H}_+\, |\,  \lim_{t\to0}\langle u, \frac{u-S_{\boldsymbol{\mu}}(t)u}{t}\rangle_{\mc{H}_+}<\infty\}$, and
 $\mathbf{H}_{\boldsymbol{\mu}}$ its generator. Using \cite[Section 1.4]{sznitman}, $\mathcal{Q}_{\boldsymbol{\mu}}$ is closed. We claim:

\begin{lemma}
Under the condition \eqref{restriction}, for $u,v\in\mathcal{C}_+$, we have  $ \mathcal{Q}_{\boldsymbol{\mu}}(u,v)= \mathcal{Q}_+(u,v)$. In particular, $ \mathcal{Q}_+$ is closable.
\end{lemma}

\begin{proof} The proof is  a straightforward adaptation of  \cite[Lemma 5.4]{GKRV20_bootstrap} and consists in identifying $\mathcal{Q}_{\boldsymbol{\mu}}$ on $\mc{C}_+$ by differentiating in $t$ the expression \eqref{thehalfannulisemigroup} for  $S_{\boldsymbol{\mu}}(t)$. Details are left to the reader.
\end{proof}

So, the quadratic form $(\mathcal{C}_+,\mathcal{Q}_+ )$  is closable.  Let $ \mathcal{D}(\mathcal{Q}_+)$ be the completion of $ \mathcal{C}_+$ for the $\mathcal{Q}_+$-norm. Then, following \cite[Prop. 5.5]{GKRV20_bootstrap}, we can identify   $ \mathcal{D}(\mathcal{Q}_+)$ as being the  quadratic form associated to the semigroup of half-annuli:
\begin{proposition}
Under the condition \eqref{restriction}, the quadratic form $( \mathcal{D}(\mathcal{Q}_+),\mathcal{Q}_+)$  defines an operator $\mathbf{H}_+$ (the Friedrichs extension), which satisfies  $\mathbf{H}_{\boldsymbol{\mu}}=\mathbf{H}_+$.
\end{proposition}

Again, we stress that $\mc{D}({\bf H}_+)=\{ u\in\mc{D}(\mc{Q}^+)\,|\, {\bf H}_+u\in \mc{H}_+\}$ and $({\bf H}_+)^{-1}:\mc{H}_+\to \mc{D}({\bf H}_+)$ is bounded. Furthermore, by the spectral theorem, $\mathbf{H}_+$ generates a strongly continuous contraction semigroup of self-adjoint operators $(e^{-t \mathbf{H}_+ } )_{t\geq 0}$ on $\mc{H}_+$, which coincides with    $S_{\boldsymbol{\mu}}(t)$.

If we let $\mc{D}'(\mc{Q}_+)$ be the dual to $\mc{D}(\mc{Q}^+)$ (i.e. the space of bounded conjugate linear functionals on $\mc{D}(\mc{Q}_+)$), the injection $\mc{H}_+\subset \mc{D}'(\mc{Q}_+)$ is continuous and the operator ${\bf H}_+$ can be extended as a bounded isomorphism 
\[{\bf H}_+:\mc{D}(\mc{Q}_+)\to \mc{D}'(\mc{Q}_+).\] 

Finally we conclude this section with a lemma that will be useful for the spectral analysis of the quadratic form
\begin{lemma}
Let $\mc{Q}_{V_+},\mc{Q}_{L},\mc{Q}_{R} $ be the bilinear forms respectively defined on $E_k$ by 
\begin{equation*}
\mc{Q}_{V_+}(u,v):=(V_{+}(\varphi^h)u,v)_2,\quad  \mc{Q}_{L}( u,v):=\lim_{k'\to\infty}(L^{(k')}( \varphi^h)u,v)_2 , \quad  \mc{Q}_{R}(u,v):=\lim_{k'\to\infty}(R^{(k')}(\varphi^h)u,v)_2
\end{equation*}
where $(\cdot,\cdot)_2$ stands here for the inner product in $L^2(\Omega_{\T^+})$. Then there exists a constant $C_k>0$ such that for all $u\in E_k$
$$\mc{Q}_{V_+}(u,u)+ \mc{Q}_{L}( u,u)+\mc{Q}_{R}(u,u)\leq C_k\|u\|_2^2. $$
\end{lemma}

\begin{proof}
A fonction $u\in E_k$ is measurable  with respect to the field $ \varphi^{h,k}$ (recall that this is the truncation of the field $ \varphi^{h}$ to its Fourier modes $|n|\leq k$). Also, from Appendix \ref{app:GMC_estimates}, the conditional expectation of $V_+$ with respect to the filtration generated by $ \varphi^{h,k}$ is the random variable $W^k( \varphi^{h,k}):=  \int_{\T^+} e^{\gamma  \varphi^{h,k}(\theta)-\frac{\gamma^2}{2}\E[  \varphi^{h,k}(\theta)^2]}\Big(\frac{1}{2|\sin(\theta)|}\Big)^{\frac{\gamma^2}{2}}\dd\theta$. Therefore
\begin{align*}
\mc{Q}_{V_+}(u,u)=&(W^k( \varphi^{h,k})u,u)_2=\|W^k(\varphi^{h,k})\|_{L^p(\Omega_{\T^+})}\|u^2\|_{L^q(\Omega_{\T^+})}<+\infty
\end{align*}
for some $p>1$ \footnote{Actually it is plain to see that $W^k(\tilde{\varphi}^{h,k})$ is in  $L^p(\Omega_{\T^+})$   for any $p>1$ but we don't need it here. } (given by  Appendix \ref{app:GMC_estimates}) and the corresponding conjugate exponent $q$ . Indeed, Hermite polynomials are in every $L^q(\Omega_{\T^+})$ spaces. By equivalence of norms on $E_k$, which is finite dimensional, our claim follows. The same argument holds for the two other potentials.
\end{proof}

\appendix
 
\section{Markov property of the GFF}
 
\begin{proposition}\label{decompGFF}
Let $(\Sigma,g)$ be a Riemann surface with corners as in Definition \eqref{def:mfd_with_corners}. Let $\mathcal{C}$ be a  union of  interior cuts in $\mathring{\Sigma}$ and  $\mathcal{B}$ be a  union of  boundary cuts, all of these curves being non overlapping.
We consider the following alternative possibilities:
\begin{enumerate}
\item Assume $\mathcal{C}\cup \mathcal{B}$ splits $\Sigma$ into two connected components $\Sigma_1$ and $\Sigma_2$. Equip the trace of the cuts on $\Sigma_1$ and $\Sigma_2$ with the marking $D$ (other markings on the other boundary components remain unchanged). Then:
\begin{enumerate}
\item  
either $\partial\Sigma_D\not=\emptyset$, in which case   the mixed b.c. GFF $Y_g$ on $\Sigma$ admits the following decomposition in law as a sum of independent processes
$$Y_g\stackrel{\rm law}{=}Y_1+Y_2+P$$
with $Y_q$ a mixed b.c. GFF on $\Sigma_q$ for $q=1,2$ and $P$ the harmonic extension on $\Sigma\setminus(\mathcal{C}\cup \mathcal{B})$ of the restriction of $Y_g$ to $\mathcal{C}\cup \mathcal{B}$ with Dirichlet boundary condition on $\partial \Sigma_D$ and Neumann boundary condition on $\partial \Sigma_N$.
\item 
 or $\partial\Sigma_D=\emptyset$, in which case   the Neumann GFF $Y_g$ on $\Sigma$ admits the following decomposition in law as a sum of independent processes
$$Y_g\stackrel{\rm law}{=}Y_1+Y_2+P-c_g$$
with $Y_q$ a mixed b.c. GFF on $\Sigma_q$ for $q=1,2$ and $P$ the harmonic extension on $\Sigma\setminus(\mathcal{C}\cup \mathcal{B})$ of the restriction of $Y_g$ to $\mathcal{C}\cup \mathcal{B}$ with  Neumann boundary condition on $\partial \Sigma_N$, and $c_g:=\frac{1}{{\rm v}_g(\Sigma)}\int_\Sigma(Y_1+Y_2+P)\,\dd {\rm v}_g $.
\end{enumerate}
\item  Assume $\mathcal{C}\cup \mathcal{B}$ does not disconnect $\Sigma$ and write $\Sigma'=\Sigma\setminus(\mathcal{C}\cup \mathcal{B})$. Equip the trace of the cuts on $\Sigma'$ with the marking $D$. Then 
\begin{enumerate}
\item either $\partial\Sigma_D\not=\emptyset$, in which case   the mixed b.c. GFF $Y_g$ on $\Sigma$ admits the following decomposition in law as a sum of independent processes
$$Y_g\stackrel{\rm law}{=}Y'+P$$
with $Y'$ a mixed b.c. GFF on $\Sigma'$   and $P$ the harmonic extension on $\Sigma'$ of the restriction of $Y_g$ to $\mathcal{C}\cup \mathcal{B}$ with Dirichlet boundary condition on $\partial \Sigma_D$ and Neumann boundary condition on $\partial \Sigma_N$.
\item or $\partial\Sigma_D=\emptyset$, in which case   the Neumann GFF $Y_g$ on $\Sigma$ admits the following decomposition in law as a sum of independent processes
$$Y_g\stackrel{\rm law}{=}Y'+P-c_g$$
with $Y'$ a mixed b.c. GFF on $\Sigma'$, $P$ the harmonic extension on $\Sigma'$ of the restriction of $Y_g$ to $\mathcal{C}\cup \mathcal{B}$ with  Neumann boundary condition on $\partial \Sigma_N$, and $c_g:=\frac{1}{{\rm v}_g(\Sigma)}\int_\Sigma(Y'+P)\,\dd {\rm v}_g $.
\end{enumerate}
\end{enumerate}
 \end{proposition}     
 
  \begin{proof} We first prove 1.a). Define  $G_{m,\Sigma}:=G_{m,\Sigma_1}\mathbf{1}_{\Sigma_1\times \Sigma_1}+G_{m,\Sigma_2}\mathbf{1}_{\Sigma_2\times \Sigma_2}$ with $G_{m,\Sigma_1},G_{m,\Sigma_2}$ the mixed b.c. Green functions on $\Sigma_1,\Sigma_2$. Set
  \[ u(x,y):=G_{g,m}(x,y)-G_{m,\Sigma}(x,y)\]
  with $G_{g,m}$ the mixed b.c. Green function on $\Sigma$. Then  $u\in C^0(\Sigma\times \Sigma\setminus {\rm diag}(\Sigma))$. For $y\in \Sigma\setminus  (\mathcal{C}\cup\mathcal{B})$ fixed, we have  
\[\Delta_gu(\cdot,y)_{|  \Sigma\setminus  (\mathcal{C}\cup\mathcal{B})}=\delta(y)-\delta(y)=0 , \quad u(\cdot,y)|_{\mathcal{C}\cup\mathcal{B}}=G_{g,m}(\cdot,y)|_{\mathcal{C}\cup\mathcal{B}}\]
with boundary conditions
$$u(\cdot,y)|_{\partial\Sigma_D}=0, \quad \partial_n u(\cdot,y)|_{\partial\Sigma_N}=0.
$$
If $\mc{P}_g$ is the Poisson kernel for the pair $(\Sigma,\mathcal{C}\cup\mathcal{B})$ with corresponding boundary conditions on $\partial\Sigma_D$ and  $\partial\Sigma_N$, we deduce
\[ u(x,y)=\int_{\mathcal{C}\cup\mathcal{B}} \mc{P}_g(x,x')G_{g,m}(x',y){\rm v}_{\mathcal{C}\cup\mathcal{B}}(\dd x') \]
for all $x\in \Sigma,y\in \Sigma\setminus (\mathcal{C}\cup\mathcal{B})$, with ${\rm v}_{\mathcal{C}\cup\mathcal{B}}$ the length measure on $\mathcal{C}\cup\mathcal{B}$. But this extends by continuity to $(\Sigma\times \Sigma)\setminus \textrm{diag}(\Sigma)$. Similarly, for $x\in \Sigma\setminus(\mathcal{C}\cup\mathcal{B})$, we have on $\Sigma  \setminus(\mathcal{C}\cup\mathcal{B})$
\[\Delta_gu(x,\cdot)=0,\quad \textrm{ with } u(x,\cdot)|_{\mathcal{C}\cup\mathcal{B}}=\int_{\mathcal{C}\cup\mathcal{B}} \mc{P}_g(x,x')G_{g,m}(x',\cdot)|_{\mathcal{C}\cup\mathcal{B}}{\rm v}_{\mathcal{C}\cup\mathcal{B}}(\dd x')\]
and corresponding boundary conditions on $\partial\Sigma_D$  and   $\partial\Sigma_N$. 
This implies that for $y\in \Sigma$, $x\in \Sigma\setminus (\mathcal{C}\cup\mathcal{B})$
\[u(x,y)= \int_{\mathcal{C}\cup\mathcal{B}} \int_{\mathcal{C}\cup\mathcal{B}} \mc{P}_g(x,x')G_{g,m}(x',y')\mc{P}_g(y,y'){\rm v}_{\mathcal{C}\cup\mathcal{B}}(\dd x'){\rm v}_{\mathcal{C}\cup\mathcal{B}}(\dd y').\]
This extends by continuity to $(x,y)\in (\Sigma\times \Sigma)\setminus \textrm{diag}(\Sigma)$. Hence our claim.

  Now we prove 1.b). Define  $G_{N,\Sigma}:=G_{\Sigma_1}\mathbf{1}_{\Sigma_1\times \Sigma_1}+G_{\Sigma_2}\mathbf{1}_{\Sigma_2\times \Sigma_2}$ with $G_{\Sigma_1},G_{\Sigma_2}$ the mixed b.c. Green functions on $\Sigma_1,\Sigma_2$ with Neumann b.c. on $\partial\Sigma$ and Dirichlet b.c. on $\mathcal{C}\cup\mathcal{B}$. Let us consider
  \[ u(x,y):=G_{g,N}(x,y)-G_{N,\Sigma}(x,y)+\frac{1}{{\rm v}_g(\Sigma)}\int_\Sigma G_{N,\Sigma}(x,y)\,  {\rm v}_g(\dd x) +\frac{1}{{\rm v}_g(\Sigma)}\int_\Sigma G_{N,\Sigma}(x,y)\, {\rm v}_g(\dd y).\]
Note that $u\in C^0(\Sigma\times \Sigma\setminus {\rm diag}(\Sigma))$. 
For $y\in \Sigma\setminus  (\mathcal{C}\cup\mathcal{B})$ fixed, we have on  $  \Sigma\setminus (\mathcal{C}\cup\mathcal{B})$ 
\[\Delta_gu(\cdot,y)=2\pi(\delta(y)-\frac{1}{{\rm v}_g(\Sigma)}-\delta(y)+\frac{1}{{\rm v}_g(\Sigma)})=0 , \quad u(\cdot,y)|_{\mathcal{C}\cup\mathcal{B}}=G_{g,N}(\cdot,y)|_{\mathcal{C}\cup\mathcal{B}}+c(y)\]
where $c(y):=\frac{1}{{\rm v}_g(\Sigma)}\int_\Sigma G_{N,\Sigma}(x,y)\,{\rm v}_g(\dd x)$.
Thus, if $\mc{P}_g$ is the Poisson kernel for the pair $(\Sigma,\mathcal{C}\cup\mathcal{B})$ with Neumann b.c. on $\partial\Sigma$ then
\[ u(x,y)=\int_{\mathcal{C}\cup\mathcal{B}} \mc{P}_g(x,x')G_{g,N}(x',y){\rm v}_{\mathcal{C}\cup\mathcal{B}}(\dd x')+c(y)\]
for all $x\in \Sigma,y\in \Sigma\setminus(\mathcal{C}\cup\mathcal{B})$, with ${\rm v}_{\mathcal{C}\cup\mathcal{B}}$ the length measure on $\mathcal{C}\cup\mathcal{B}$. But this extends by continuity to $(\Sigma\times \Sigma)\setminus \textrm{diag}(\Sigma)$.
Note that $c(y)=0$ if $y\in \mathcal{C}\cup\mathcal{B}$. 

Now similarly, for $x\in \Sigma\setminus (\mathcal{C}\cup\mathcal{B})$, we have on $  \Sigma\setminus (\mathcal{C}\cup\mathcal{B})$
\[\Delta_gu(x,\cdot)=0,\quad \textrm{ with } u(x,\cdot)|_{\mathcal{C}\cup\mathcal{B}}=\int_{\mathcal{C}\cup\mathcal{B}} \mc{P}_g(x,x')G_{g,N}(x',\cdot)|_{\mathcal{C}}{\rm v}_{\mathcal{C}\cup\mathcal{B}}(\dd x').\] 
This implies that for $y\in \Sigma$, $x\in \Sigma\setminus (\mathcal{C}\cup\mathcal{B})$
\[u(x,y)= \int_{\mathcal{C}\cup\mathcal{B}}  \int_{\mathcal{C}\cup\mathcal{B}} \mc{P}_g(x,x')G_{g,N}(x',y')\mc{P}_g(y,y'){\rm v}_{\mathcal{C}\cup\mathcal{B}}(\dd x'){\rm v}_{\mathcal{C}\cup\mathcal{B}}(\dd y').\]
This extends by continuity to $(x,y)\in (\Sigma\times \Sigma)\setminus \textrm{diag}(\Sigma)$.
Thus,  denoting  by $\mc{S}_g$ the restriction of $G_{g,N}$  to $(\mathcal{C}\cup\mathcal{B})^2$
\[ G_{g,N}(x,y)=G_{N,\Sigma}(x,y)-\frac{1}{{\rm v}_g(\Sigma)}\int_\Sigma G_{N,\Sigma}(x,y) {\rm v}_g(\dd x) -\frac{1}{{\rm v}_g(\Sigma)}\int_\Sigma G_{N,\Sigma}(x,y)\,{\rm v}_g(\dd y)+\mc{P}_g\mc{S}_g\mc{P}_g^*(x,y).\]
Observe that, if $C:=\int G_{N,\Sigma}(x',y'){\rm v}_g(\dd x'){\rm v}_g(\dd y')$,  
\[\int_\Sigma u(x,y){\rm v}_g(\dd x)=\int_\Sigma u(x,y){\rm v}_g(\dd y)=C\frac{1}{{\rm v}_g(\Sigma)}, \quad \int_{\Sigma\times \Sigma}u(x,y){\rm v}_g(\dd x){\rm v}_g(\dd y)=C.\]
This implies that 
\[\begin{split}
\mc{P}_g\mc{S}_g\mc{P}_g^*(x,y)=& \mc{P}_g\mc{S}_g\mc{P}_g^*(x,y)-\frac{1}{{\rm v}_g(\Sigma)}\int_\Sigma\mc{P}_g\mc{S}_g\mc{P}_g^*(x,y'){\rm v}_g(\dd y')-\frac{1}{{\rm v}_g(\Sigma)}\int_\Sigma\mc{P}_g\mc{S}_g\mc{P}_g^*(x',y){\rm v}_g(\dd x')\\
& +\frac{1}{{\rm v}_g(\Sigma)^2}\int_{\Sigma\times \Sigma}\mc{P}_g\mc{S}_g\mc{P}_g^*(x',y'){\rm v}_g(\dd y'){\rm v}_g(\dd x')+\frac{1}{{\rm v}_g(\Sigma)^2}\int G_{N,\Sigma}(x',y'){\rm v}_g(\dd x'){\rm v}_g(\dd y').
\end{split}\]
In particular we have proved that 
\[\begin{split}
G_{g,N}(x,y)=& G_{N,\Sigma}(x,y)+\mc{P}_g\mc{S}_g\mc{P}_g^*(x,y) -\frac{1}{{\rm v}_g(\Sigma)}\int_\Sigma G_{N,\Sigma}(x',y){\rm v}_g(\dd x') -\frac{1}{{\rm v}_g(\Sigma)}\int_\Sigma G_{N,\Sigma}(x,y'){\rm v}_g(\dd y')\\ 
 & +\frac{1}{{\rm v}_g(\Sigma)^2} \int_\Sigma G_{N,\Sigma}(x',y'){\rm v}_g(\dd x'){\rm v}_g(\dd y')-\frac{1}{{\rm v}_g(\Sigma)}\int_\Sigma\mc{P}_g\mc{S}_g\mc{P}_g^*(x,y'){\rm v}_g(\dd y')\\\
 & -\frac{1}{{\rm v}_g(\Sigma)}\int_\Sigma\mc{P}_g\mc{S}_g\mc{P}_g^*(x',y){\rm v}_g(\dd x')+\frac{1}{{\rm v}_g(\Sigma)^2}\int_{\Sigma\times \Sigma}\mc{P}_g\mc{S}_g\mc{P}_g^*(x',y'){\rm v}_g(\dd x'){\rm v}_g(\dd y').
\end{split}\]
This proves our second claim.

The proofs of 2.a) and 2.b) follow the same lines.
   \end{proof}
 \section{Several identities about determinant of Laplacians}\label{app:determinants}

 \begin{lemma}\label{asympt_exp_heat}
Let $\Sigma^{\#2}$ be the Neumann double of a Riemann surface $\Sigma$ with corners and $g$ be an admissible 
metric on $\Sigma$ extended to a symmetric metric $g^{\#2}$ with respect to the involution $\tau_\Sigma$. The integral  kernel 
$k_{\Sigma^{\#2}}(t,x,x')$ of the heat operator $e^{-t\Delta_{\Sigma^{\#2},g^{\#2},D}}$ of the Laplacian with Dirichlet condition on $\pl \Sigma^{\#2}$ satisfies
\[ \int_{\Sigma} (k_{\Sigma^{\#2}}(t,x,x)+k_{\Sigma^{\#2}}(t,x,\tau_\Sigma(x))){\rm dv}_g(x)=\sum_{k=0}^{2} c_k t^{-1+k/2}+\mc{O}(t^{1/2})\]
as $t\to 0$, for some $c_j\in \R$. 
 \end{lemma}
\begin{proof}
The heat kernel $k_{\Sigma^{\#2}}(t,x,x')$ on $\Sigma^{\#2}$ with the metric $g^{\#2}$ is  studied in \cite{McKean-Singer}. 
We cover $\Sigma^{\#2}$ by local holomorphic charts $\omega_j: V_j\to \C$ where the metric is, in holomorphic coordinates, $e^{\rho_j}|dz|^2$. We can assume that either $V_j\cap \pl\Sigma_N=\emptyset$ (case 1), or we split in the two other cases 
$V_j\cap \pl\Sigma_N \not=\emptyset$ and $ V_j\cap \pl \Sigma^{\#2}=\emptyset$  (case 2), or $V_j\cap \pl\Sigma_N \not=\emptyset$ and $V_j\cap \pl\Sigma^{\#2}\not=\emptyset$ (case 3); in case 2, we can assume that $\omega_j(V_j)=\D$ while in case 1 one has $\omega_j(V_j)=\D$ or $\omega_j(V_j)=\D^+$ depending on whether $V_j\cap \pl \Sigma^{\#2}=\emptyset$ or not. In case 3 we can assume that $\omega_j(V_j)=\D\cap \{{\rm Re}(z)\geq 0\}$ with, in case $2$ and $3$,  $\tau_\Sigma(\omega_j^{-1}(z))=\omega_j(\bar{z})$ (in particular $\pl\Sigma_N=\{{\rm Im}(z)=0\}$ there). 
We write $J_1,J_2,J_3$ the three sets of indices $j$ according to the partition of charts just described.  Moreover since $g$ is admissible, up to shrinking slightly $V_j$ we can assume that $\rho_j=0$ for $j\in J_3$ (see Remark \ref{remark_metric}). 
In $V_j$ for $j\in J_1\cup J_2$, we have by \cite[Section 5]{McKean-Singer}
\[ \int_{V_j}(k_{\Sigma^{\#2}}(t,x,x)+k_{\Sigma^{\#2}}(t,x,\tau_\Sigma(x)){\rm dv}_g(x)=\sum_{j=1}^2  c_jt^{-1+j/2}+\mc{O}(t^{1/2})\]
for some $c_j\in \R$. To describe the behaviour near the corners, i.e. in $V_j$ for $j\in J_3$, we  use cutoff functions $\chi_j\in C_c^\infty(V_j)$ equal to $1$ near the corner and $\tilde{\chi}_j\in C_c^\infty(V_j)$  equal to $1$ on the support of $\chi_j$. Let $k_0(t,z,z')=(4\pi t)^{-1}e^{-|z-z'|^2/4t}$ be the Euclidean heat kernel and let
$k^D_0(t,z,z'):=k_0(t,z,z')-k_0(t,z,-\bar{z}')$  be the Dirichlet heat kernel on the half plane $\{{\rm Re}(z)\geq 0\}$, and let $K_0(t)$ be the associated operator. Here we emphasize that the measure of integration for these integral kernels is the Riemanian measure ${\rm dv}_g$ in $V_j$ which is the Euclidian measure in $V_j$ with $j\in J_3$. We have for $K^j_0(t):=\tilde{\chi}_jK_0(t)\chi_j$
\[ (\pl_t +\Delta_{\Sigma^{\#2},g^{\#2}})K^j_0(t)=[\Delta_{\Sigma^{\#2},g^{\#2}},\tilde{\chi}_j]K_0(t)\chi_j=: S(t), \quad K_0^j(0)=\chi_j.\]
Then, we obtain
\[ e^{-t\Delta_{\Sigma^{\#2},g^{\#2},D}}\chi_j=K^j_0(t)-\int_0^t e^{-(t-s)\Delta_{\Sigma^{\#2},g^{\#2},D}}S(s)\dd s.\]
Since $[\Delta_{\Sigma^{\#2},g^{\#2}},\tilde{\chi}_j]$ is a differential operator of order $1$ with support disjoint from $\chi_j$, $S(t)$ is an operator with smooth integral kernel, whose support is at positive distance from the diagonal. This implies in particular that for all multi-indices $\alpha,\beta$, in local coordinates we have uniformly in $x,x'$ and $s\in [0,1]$
\[ |\pl_{x}^{\alpha}\pl_{x'}^{\beta}S(s,x,x')|\leq C_{\alpha \beta}e^{-\eps/s}.\]
Using this and the standard estimate $e^{-t\Delta_{\Sigma^{\#2},g^{\#2},D}}(x,x')\leq Ct^{-1}e^{-d_g(x,x')^2/4t}$ (e.g. \cite{McKean-Singer}), we easily check that for $x,x'$ close to the corner 
\[e^{-t\Delta_{\Sigma^{\#2},g^{\#2},D}}(x,x')=K_0(t,x,x')\chi_j(x')\tilde{\chi}_j(x)+tZ(t,x,x')\]
with $Z$ smooth on $[0,1)\times \Sigma^{\#2}\times \Sigma^{\#2}$. This implies that, on $W_j=\omega_j^{-1}(\D^{++}\cap\{|z|\leq \eps\})$ with $\eps>0$ small, 
\[ \begin{split}
\int_{W_j}(k_{\Sigma^{\#2}}(t,x,x)+k_{\Sigma^{\#2}}(t,x,\tau_\Sigma(x))){\rm dv}_g(x)=& \frac{1}{4\pi t}\int_{\D^{++}\cap \{|z|\leq \eps\}}(1-e^{-\frac{|z+\bar{z}|}{4t}^2} +e^{-\frac{|z-\bar{z}|}{4t}^2}-e^{-\frac{|2z|}{4t}^2}) dz+\mc{O}(t)\\
 =&c_0t^{-1}+c_1t^{-1/2}+c_2+\mc{O}(t^{1/2}) 
\end{split}\]
for some $c_0,c_1\in \R$.
\end{proof}
 
Let $\Sigma$ be a Riemann surface with $b$ boundary circles in the sense of Definition \ref{d:surf_with_bdry}, $\zeta_j:\T \to \pl_j\Sigma$ some parametrizations of the boundary. Let $g$ be an admissible metric, i.e. $g=|dz|^2/|z|^2$ in an annular coordinate neighborhood of each boundary circle induced by $\zeta_j^{-1}$, that is biholomorphic to $\A_\delta$ for some $\delta<1$. Then we define the Dirichlet-to-Neumann operator, for $\la>-\eps$ with $\eps>0$ small (so that the Dirichlet Laplacian $\Delta_{\Sigma,g,D}$ on $\Sigma$ has no spectrum in $(-\infty,\eps)$), 
\[{\bf D}_{\Sigma}(\la):C^\infty(\T)^{b}\to C^\infty(\T)^{b}, \quad \]
 associated to the equation $(\Delta_{\Sigma,g}+\la)u=0$, defined exactly as in \eqref{DNmap_def} by taking ${\bf D}_{\Sigma}(\la)f=((-\pl_\nu P_\Sigma (\la)f)\circ \zeta_1,\dots, (-\pl_\nu P_\Sigma (\la)f)\circ \zeta_b)$
 where $P_\Sigma (\la)f$ is the solution of $(\Delta_{\Sigma,g}+\la)P_\Sigma(\la)f=0$ with boundary condition $f$ (via the parametrizations $\zeta_j$). We let ${\bf D}(\la)$ be the DN map on the infinite half cylinder $\A:=[0,\infty)\times \T$ with metric $dt^2+d\theta^2$, which is isometric to $(\D\setminus\{0\},|dz|^2/|z|^2)$, defined by ${\bf D}_{\A}(\la)f=-\pl_t P_{\A}(\la)f|_{t=0}$ if $P_{\A}(\la)f$ is the unique $L^2(\A)$ solution of $(\Delta_{\A}+\la)u=0$ with $u|_{t=0}=f$. Then we set ${\bf D}(\la):=({\bf D}_{\A}(\la),\dots,{\bf D}_{\A}(\la))$ acting on $C^\infty(\T)^b$ and on its dual $C^{-\infty}(\T)^b$ (the space of distributions).
\begin{lemma}\label{analyse_D}
For all $\la>0$ and $\Pi$ the projector on constants, the operator $K(\la):={\bf D}_{\Sigma}(\la)({\bf D}(\la)+\Pi)^{-1}-{\rm Id}$  has a smooth integral kernel, depending in a $C^1$ fashion on $\la$, and it satisfies as $\la\to \infty$,  
\begin{equation}\label{boundonK}
  \|K(\la;\cdot,\cdot)\|_{C^2}+\|\pl_\la K(\la;\cdot ,\cdot)\|_{C^2} \leq Ce^{-c\la}
  \end{equation}
for some $c>0$. The operators ${\bf D}_\Sigma(\la)$ and ${\bf D}(\la)$ satisfy that 
$(\pl_\la {\bf D}_\Sigma(\la)){\bf D}^{-1}_\Sigma(\la)$ and $(\pl_\la {\bf D}(\la)){\bf D}^{-1}(\la)$ are pseudo-differential operators of order $-2$ on $\sqcup_{j=1}^b \T$ and are of trace class.
\end{lemma}
 \begin{proof} A direct computation gives that the solution of $(\Delta_\A+\la)u=0$ with $u|_{t=0}=f$ is 
 \[ P_\A(\la)f(t,\theta):=\sum_{k\in \Z}e^{-t\sqrt{k^2+\la}}f_k e^{ik\theta} \]
 and thus
 \begin{equation}\label{DNmaplambda}
  {\bf D}_\A (\la)f(\theta)= \sum_{k\in \Z} \sqrt{k^2+\la}f_ke^{ik\theta}, \quad  \textrm{ if }f(\theta)=\sum_{k\in \Z}f_ke^{ik\theta}.
  \end{equation}
We identify $\A$ with $\D\setminus \{0\}$ via the map $z=e^{-t+i\theta}$.
Take $f\in H^{-N}(\T)^b$ for some $N$ and
let $\chi_j\in C^\infty(\Sigma)$ be supported in a small neighborhod $V_j$ of $\pl_j \Sigma$ isometric via $\omega_j$ to $(\A_\delta,|dz|^2/|z|^2)$ and $(\omega_j)_*\chi_j=1$  in $\{|z|>(1+\delta)/2\}$. We have 
\[ P_\Sigma (\la)f=\sum_{j}\chi_j\omega_j^*P_{\A}(\la)f- (\Delta_{\Sigma,g,D}+\la)^{-1}\sum_j [\Delta_{\Sigma,g},\chi_j]\omega_j^*P_{\A}(\la)f.\]
Since $ (\Delta_{\Sigma,g,D}+\la)^{-1}$ has integral kernel that is smooth outside the diagonal and 
$[\Delta_{\Sigma,g},\chi_j]$ is a differential operator supported in $\omega_j^{-1}(\{|z|>(1+\delta)/2\})$, i.e. at a positive distance from $\pl \Sigma$, we see that $P_\Sigma (\la)f-\sum_{j}\chi_j\omega_j^*P_{\A}(\la)f$ is smooth near $\pl \Sigma$. We also notice that $u(\la):=\sum_j [\Delta_{\Sigma,g},\chi_j]\omega_j^*P_{\A}(\la)f$ satisfies that there is $c>0$ such that  for each $N\geq 0$ and $\ell=0,1$ there is $C_N$
\[ \| \pl_\la^{\ell}u(\la)\|_{H^N(\Sigma)}\leq C_N (1+\la)^{N}e^{-c\la}\|f\|_{H^{-N}(\Sigma)}.\]
Since we also have for each $N$ (where $H^N(\Sigma)$ is any fixed Sobolev norm of order $N$ on $\Sigma$) and $\ell=0,1$
\[ \|\pl_\la^{\ell}(\Delta_{\Sigma,g,D}+\la)^{-1}\|_{\mc{L}(H^N)}\leq C_N\]
for some $C_N>0$ depending on $N$, we deduce that  (up to changing $C_N$) 
\begin{equation}\label{boundonP}
 \|(\pl_\la^{\ell}(P_\Sigma (\la)f-\sum_{j}\chi_j\omega_j^*P_{\A}(\la)f)\|_{H^N(\Sigma)}\leq C_Ne^{-c\la/2}.
 \end{equation}
Thus $\tilde{{\bf D}}_\Sigma(\la):= {\bf D}_\Sigma (\la)-{\bf D}(\la)$ is given by 
\[ (\tilde{{\bf D}}_\Sigma (\la)f)\circ \zeta_i^{-1}= -\pl_\nu (\Delta_{\Sigma,g,D}+\la)^{-1}\sum_j ([\Delta_{\Sigma,g},\chi_j]\omega_j^*P_{\A}(\la)f)|_{\pl_i\Sigma}\]
and using \eqref{boundonP} together with Sobolev embedding, we obtain \eqref{boundonK}. Since $\tilde{{\bf D}}(\la)$ and its $\la$ derivative are smoothing, to prove that $(\pl_\la {\bf D}_\Sigma(\la)){\bf D}^{-1}_\Sigma(\la)$ is of order $-2$ it suffices to show it for $(\pl_\la {\bf D}(\la)){\bf D}^{-1}(\la)$, but this is a direct consequence of the expression \eqref{DNmaplambda}
\end{proof}
 
 \begin{proposition}\label{detformula}
Let $\Sigma$ be a Riemann surface with corners in the sense of Definition \ref{def:mfd_with_corners} or with boundary in the sense of Definition \ref{d:surf_with_bdry} and let $g$ be an admissible metric on $\Sigma$. Let $\mc{C}'=\cup_{j=1}^{b'_\ell}\mc{C}_j,\mc{B}'=\cup_{j=1}^{b'_h}\mc{B}'_j$ be interior and boundary cuts in $\Sigma$. Consider the Laplacian $\Delta_{\Sigma,g,m}$ with mixed boundary condition on $\Sigma$ and $\Delta_{\Sigma,\mc{C}'\cup \mc{B}',g,m}$ with Dirichlet condition at the cuts $\mc{C}'\cup \mc{B}'$ and mixed boundary condition at the boundary.
 Then, if $\pl \Sigma_D\not=\emptyset$ one has  
\[\det(\Delta_{\Sigma,g,m})= \det(\Delta_{\Sigma,\mc{C}'\cup \mc{B}',g,m}) {\det}_{\rm Fr}(\mathbf{D}_{\Sigma,\mc{C}',\mc{B}'}(2{\bf D}_0)^{-1})(2\pi)^{b_\ell'+b_h'/2}\]
and if $\pl \Sigma_D=\emptyset$ one has
\[{\det}'(\Delta_{\Sigma,g,m})=\frac{\sqrt{2}}{\sqrt{\pi}}{\rm v}_g(\Sigma)\det(\Delta_{\Sigma, \mc{B}',g,m}){\det}_{\rm Fr}({\bf D}_{\Sigma,\mc{B}',0}(2{\bf D}_0)^{-1}), \textrm{ if }(b_\ell',b_h')=(0,1) \]
\[{\det}'(\Delta_{\Sigma,g,m})={\rm v}_g(\Sigma)\det(\Delta_{\Sigma, \mc{C}',g,m}){\det}_{\rm Fr}({\bf D}_{\Sigma,\mc{C}',0}(2{\bf D}_0)^{-1}), \textrm{ if }(b_\ell',b_h')=(1,0)\]
where $\mathbf{D}_{\Sigma,{\mc{C}'},0}=\mathbf{D}_{\Sigma,\mc{C}}+\Pi_0'$ and $\mathbf{D}_{\Sigma,{\mc{B}'},0}=\mathbf{D}_{\Sigma,\mc{B}'}+\Pi_0'$ (recall \eqref{Pi0'} for definition of $\Pi'_0$).
 \end{proposition}
\begin{proof} We shall follow the proof of \cite{BurgheleaFK92} by working on the Neumann double $\Sigma^{\#2}$ and restricting to even functions with respect to $\tau_\Sigma$. Let $\mc{E}:L^2(\T)^{2b'_\ell+b'_h}\to L^2_{\rm even}(\T)^{2b'_\ell+b'_h}$ 
be the orthogonal projection 
\[ \mc{E}f= \frac{1}{2}(f+f\circ \tau_\Sigma)\]
on even functions with respect to
 $\tau_\Sigma$ restricted to $\mc{C}''\cup \mc{B}''$ where  $\mc{C}''=\cup_{j} (\mc{C}_j'\cup \tau_\Sigma(\mc{C}'_j))$  and $\mc{B}''=\cup_{j} (\mc{B}_j'\cup \tau_\Sigma(\mc{B}_j'))$ identified to $2b'_\ell$ copies of $\T$ for the doubled loops and $b'_h$ copies of $\T$ for the doubled half-circles thanks to the parametrizations. of the cuts. All our Laplacians will be for the metric $g$ on $\Sigma$ or  the double $g^{\#2}$ on $\Sigma^{\#2}$, and to lighten the notations we remove the index corresponding to the metric when writing Laplacians.
For  $\la>0$, we can consider $\det(\Delta_{\Sigma,m}+\la)$ and $\det(\Delta_{\Sigma,\mc{C}'\cup \mc{B'},m}+\la)$, which are equal to $\pl_s\zeta_{\Sigma}(0,\la)$ and $\pl_s\zeta_{\Sigma'}(0,\la)$ if 
\[ \zeta_{\Sigma}(s,\la):={\rm Tr}(\mc{E}(\Delta_{\Sigma^{\#2},D}+\la)^{-s}), \quad \zeta_{\Sigma,\mc{C}'\cup \mc{B}'}(s,\la):={\rm Tr}(\mc{E}(\Delta_{\Sigma^{\#2},\mc{C}''\cup\mc{B}'',D}+\la)^{-s}).\]
Here $\Delta_{\Sigma^{\#2},D}:=\Delta_{\Sigma^{\#2},g^{\#2},D}$  denotes the Laplacian on $\Sigma^{\#2}$, with Dirichlet condition if $\Sigma^{\#2}$ has a non-trivial boundary component, and $\Delta_{{\Sigma}^{\#2},\mc{C}''\cup\mc{B}'',D}:=\Delta_{{\Sigma}^{\#2},\mc{C}''\cup\mc{B}'',g^{\#2},D}$  is the Laplacian on ${\Sigma}^{\#2}$ with Dirichlet condition at the boundary (in case of boundary) and on the curves of $\mc{C}''\cup \mc{B''}$.
Notice that one can write for ${\rm Re}(s)>1$ (using Lemma \ref{asympt_exp_heat})
\begin{equation}\label{zetaSigmas}
\zeta_{\Sigma}(s,\la):=\frac{1}{\Gamma(s)}\int_0^\infty {\rm Tr}(\mc{E}e^{-t\Delta_{\Sigma^{\#2},D}})e^{-t\la}t^{s-1}\dd t.
\end{equation}
and similarly for $\zeta_{\Sigma,\mc{C}'\cup \mc{B}'}$. By Lemma \ref{asympt_exp_heat}, for each $\la>0$, $\zeta_\Sigma(s,\la)$ and $\zeta_{\Sigma'}(s,\la)$ extend analytically in $s$ around $s=0$.
As in \cite{BurgheleaFK92}, we compute the derivative for ${\rm Re}(s)\gg 1$
\[ \pl_\la  \zeta_{\Sigma}(s,\la)- \pl_\la  \zeta_{\Sigma,\mc{C}'\cup \mc{B}'}(s,\la)=-s{\rm Tr}(\mc{E}((\Delta_{\Sigma^{\#2},D}+\la)^{-s-1}-(\Delta_{{\Sigma}^{\#2},\mc{C}''\cup\mc{B}'',D}+\la)^{-s-1})),\]
and this extends analytically around $s=0$ and 
\begin{equation}\label{variationzeta} 
 \pl_s\pl_\la  \zeta_{\Sigma}(0,\la)=-{\rm Tr}(\mc{E}((\Delta_{\Sigma^{\#2},D}+\la)^{-1}-(\Delta_{{\Sigma}^{\#2},\mc{C}''\cup\mc{B}'',D}+\la)^{-1})).
 \end{equation}

Next, we remark that, as in \cite[Lemma 3.6]{BurgheleaFK92} or the proof of \cite[Proposition D1]{GKRV21_Segal}, 
\begin{equation}\label{diff_resolv}
(\Delta_{\Sigma^{\#2},D}+\la)^{-1}-(\Delta_{{\Sigma}^{\#2},\mc{C}''\cup\mc{B}'',D}+\la)^{-1}=P_{\Sigma^{\#2},\mc{C}'', \mc{B}'}B_{\mc{C}''\cup\mc{B''}}(\Delta_{\Sigma^{\#2},D}+\la)^{-1}
\end{equation}
where $P_{\Sigma^{\#2},\mc{C}'', \mc{B}'}(\la): C^\infty(\T)^{2b'_\ell+b'_h}\to C^0(\Sigma^{\#2})$ is the Poisson operator sending a function $f$ to the solution of $(\Delta_{\Sigma^{\#2},D}+\la)u=0$ with value $f$ at $\mc{C}''\cup \mc{B}''$, using the parametrization of the circles 
$\mc{C}'_j$, $\tau_\Sigma(\mc{C}_j')$ and $\mc{B}_j'\cup \tau_\Sigma(\mc{B}_j')$ 
induced by ${\zeta_j'}^{\ell}$ and ${\zeta_j'}^{h}$ (similar to Definition \ref{def_harmoniccut}), and 
$B_{\mc{C}''\cup \mc{B}''}$ is the operator of restriction to $\mc{C}''\cup \mc{B}''$.  Now, we let 
${\bf D}(\la):=({\bf D}^{\ell}(\la),{\bf D}^h(\la))$ acting on $C^\infty(\T)^{2b'_\ell}\times C^\infty(\T)^{b'_h}$ where 
\[ {\bf D}^{\ell}(\la)=({\bf D}_{\A}(\la),\dots,{\bf D}_{\A}(\la)), \quad  {\bf D}^{h}(\la):=({\bf D}_\A(\la),\dots,{\bf D}_\A(\la))\]
and ${\bf D}_0(\la)={\bf D}(\la)+\Pi_0'$ (recall \eqref{Pi0'}). We have 
\begin{equation}\label{der_det_F} 
\pl_\la \log {\det}_{\rm Fr}(\mathbf{D}_{\Sigma,\mc{C}',\mc{B}'}(\la)(2{\bf D}_0(\la))^{-1})=
{\rm Tr}(K_1(\la)+K_2(\la)) 
\end{equation}
\[ K_1(\la):= \pl_\la \mathbf{D}_{\Sigma,\mc{C}',\mc{B}'}(\la) \mathbf{D}^{-1}_{\Sigma,\mc{C}',\mc{B}'}(\la), \quad K_2(\la)= \mathbf{D}_{\Sigma,\mc{C}',\mc{B}'}(\la)  \pl_\la ((2{\bf D}_0(\la))^{-1}))2{\bf D}_0(\la) \mathbf{D}_{\Sigma,\mc{C}',\mc{B}'}(\la)^{-1}\]
We notice that 
\begin{align*} 
&K_1(\la)=\mc{E}\pl_\la \mathbf{D}_{\Sigma^{\#2},\mc{C}'',\mc{B}''}(\la) \mathbf{D}^{-1}_{\Sigma^{\#2},\mc{C}'',\mc{B}''}(\la)\\
& K_2(\la)=\mc{E}\mathbf{D}_{\Sigma^{\#2},\mc{C}'',\mc{B}''}(\la) \pl_\la ((2{\bf D}_0(\la))^{-1}))2{\bf D}_0(\la)  \mathbf{D}^{-1}_{\Sigma^{\#2},\mc{C}'',\mc{B}''}(\la).
\end{align*}
By Lemma \ref{analyse_D}, $K_1(\la)$ and $K_2$ are of the form $\mc{E}$ times a pseudo-differential operator of order $-2$, thus of trace class. Using that  $\mathbf{D}_{\Sigma^{\#2},\mc{C}'',\mc{B}''}(\la)$ commutes with $\mc{E}$ and the cyclicity of the trace, we deduce that 
 \[
 {\rm Tr}(K_2(\la))={\rm Tr}(\mc{E} \pl_\la {\bf D}^{-1}_0(\la) {\bf D}_0(\la)), \quad {\rm Tr}(K_1(\la))={\rm Tr}(\mc{E}\pl_\la \mathbf{D}_{\Sigma^{\#2},\mc{C}'',\mc{B}''}(\la) \mathbf{D}^{-1}_{\Sigma^{\#2},\mc{C}'',\mc{B}''}(\la)).
 \]
We compute using the expression \eqref{DNmaplambda} and the fact that the projection $\mc{E}$ selects 
only positive Fourier modes on the circles of $\mc{B}''$ and selects only one copy of each pair $\mc{C}'_j\cup \tau_\Sigma(\mc{C}'_j)$ in $\mc{C}''$
\[{\rm Tr}(K_2(\la))= b_\ell' (-\frac{1}{2}(\sqrt{\la}+1)^{-1}\la^{-1/2}
-\sum_{k=1}^\infty (k^2+\la)^{-1})+b_h'(-\frac{1}{2}(\sqrt{\la}+1)^{-1}\la^{-1/2}
-\frac{1}{2}\sum_{k=1}^\infty (k^2+\la)^{-1}).\]
The operator ${\bf D}_{\A,0}(\la):={\bf D}_\A(\la)+(1,\cdot)_h$ is an elliptic positive definite self-adjoint pseudo-differential operator of order $1$ and its spectral determinant is well-defined (see \cite{BurgheleaFK92}) by $\det({\bf D}_{\A,0}(\la))=\exp(-\zeta_{{\bf D}_{\A,0}(\la)}'(0))$ where for ${\rm Re}(s)\gg 1$, the spectral zeta function is 
\[ \zeta_{{\bf D}_{\A,0}(\la)}(s):={\rm Tr}({\bf D}_{\A,0}(\la)^{-s})=(\sqrt{\la}+1)^{-s}+2\sum_{k=1}^\infty (\la+k^2)^{-s/2}.\]
Notice that $\det({\bf D}_{\A,0}(0))=e^{-\pl_s\zeta_{{\bf D}_{\A,0}(0)}(0)}=e^{-2\zeta'_R(0)}$ where $\zeta_R(s)$ is Riemann zeta function.
We then set
\begin{equation}\label{detED0}
\det(\mc{E}{\bf D}_0(\la)):=\det({\bf D}_{\A,0}(\la))^{b_\ell'+b_h'/2}(\sqrt{\la}+1)^{b_h'/2}
\end{equation} 
and check that 
\begin{equation}\label{TrK2}
{\rm Tr}(K_2(\la))=-\pl_\la \log \det(\mc{E}{\bf D}_0(\la)).
\end{equation}
By \cite[Corollary 3.8]{BurgheleaFK92}, we also have 
\[ \mathbf{D}^{-1}_{\Sigma^{\#2},\mc{C}'',\mc{B}''}(\la) \pl_\la \mathbf{D}_{\Sigma^{\#2},\mc{C}'',\mc{B}''}(\la)=B_{\mc{C}'', \mc{B}'} (\Delta_{\Sigma^{\#2},D}+\la)^{-1}P_{\Sigma^{\#2},\mc{C}'', \mc{B}'}(\la).\]
Thus by combining with \eqref{TrK2} and using the cyclicity of the trace we get 
\[\begin{split}
 {\rm Tr}(K_1(\la))=& {\rm Tr}(\mc{E} P_{\Sigma^{\#2},\mc{C}'', \mc{B}'}B_{\mc{C}''\cup\mc{B''}}(\Delta_{\Sigma^{\#2},D}+\la)^{-1})\\
 =&- \pl_s\pl_\la  \zeta_{\Sigma}(0,\la)
 \end{split}\]
 where we used \eqref{variationzeta} and \eqref{diff_resolv} in the second line. Combining with \eqref{der_det_F} and 
 \eqref{TrK2}, we deduce that
 \[\begin{split} 
 \pl_\la \log\Big(\frac{\det(\Delta_{\Sigma,m}+\la)}{\det(\Delta_{\Sigma,\mc{C}'\cup \mc{B}',m}+\la)}\Big)=&
 {\rm Tr}(K_1(\la))= \pl_\la \log \Big({\det}_{\rm Fr}(\mathbf{D}_{\Sigma,\mc{C}',\mc{B}'}(\la)(2{\bf D}_0(\la))^{-1})\det(\mc{E}{\bf D}_0(\la))\Big)
 \end{split}\]
 and there exists a constant $A\in \R$ such that
 \begin{equation}\label{identity_of_det}
 \frac{\det(\Delta_{\Sigma,m}+\la)}{\det(\Delta_{\Sigma,\mc{C}'\cup \mc{B}',m}+\la)}=A
{\det}_{\rm Fr}(\mathbf{D}_{\Sigma,\mc{C}',\mc{B}'}(\la)(2{\bf D}_0(\la))^{-1})\det(\mc{E}{\bf D}_0(\la)).
\end{equation}
To determine the constant $A$, we consider the asymptotics of these quantities as $\la\to \infty$. By Lemma \ref{asympt_exp_heat} and \eqref{zetaSigmas}, 
\[\begin{split} 
\zeta_{\Sigma}(s,\la)=& \frac{1}{\Gamma(s)}\int_0^\infty \sum_{j=0}^{2}c_j t^{s-2+j/2}e^{-t\la}\dd t+G(s,\la)\\
=& \sum_{j=0}^2 \la^{1-s-j/2}c_j \frac{\Gamma(s-1+j/2)}{\Gamma(s)}+G(\la,s)
\end{split}\]
with $G(s,\la)=o(1)$ as $\la\to \infty$, uniformly in $s$ in a small neighborhood of $0$. We apply $\pl_s$ and observe that there are $c_j,d_j\in \R$ such that
\[ \log \det(\Delta_{\Sigma,m}+\la)=-\pl_s\zeta_{\Sigma}(0,\la)=\sum_{j=0}^1  \la^{1-j/2}(c_j\log\la+d_j)+c_2\log(\la)+o(1)\]
as $\la\to \infty$, in particular the coefficient of order $\la^0$ vanishes in the asymptotic exansion. Since $\Delta_{\Sigma,\mc{C}'\cup \mc{B}',m}$ can also be viewed as a Dirichlet Laplacian on a surface with corners $\Sigma'$ by taking $\Sigma\setminus (\mc{C}'\cup\mc{B}')$ and adding two copies of $\mc{C}'_j$ and $\mc{B}_j$ (one from each side when approaching $\mc{C}'_j$ or $\mc{B}_j'$), we can apply the same argument as for $\det(\Delta_{\Sigma,m})$ to get 
\[ \log \det(\Delta_{\Sigma,\mc{C}'\cup \mc{B}',m}+\la)=\sum_{j=0}^1  \la^{1-j/2}(c'_j\log\la+d_j')+c'_2\log(\la)+o(1)\]
as $\la\to \infty$, for some $c_j',d_j'\in \R$. Similarly $\Sigma^{\#2}\setminus (\mc{C}''\cup \mc{B}'')$ can be compactified by adding two copies of $\mc{C}''$ and $\mc{B}''$ (one from each side) and 
$\mathbf{D}_{\Sigma^{\#2},\mc{C}'',\mc{B}''}(\la)$ on $\mc{C}''\cup\mc{B}''$ 
is the sum of the two Dirichlet-to-Neumann defined on each copy. We can thus apply  
Lemma \ref{analyse_D} to deduce that 
\begin{equation}\label{relativeDN} 
\mc{E}\mathbf{D}_{\Sigma^{\#2},\mc{C}'',\mc{B}''}(\la)(2{\bf D}_0(\la))^{-1}=\mc{E}(1+K(\la))
\end{equation}
with $K(\la)$ smoothing and satisfying the estimates  \eqref{boundonK}. This implies that as $\la\to \infty$
\[ \log{\det}_{\rm Fr}(\mathbf{D}_{\Sigma,\mc{C}',\mc{B}'}(\la)(2{\bf D}_0(\la))^{-1})=o(1).\]
Next we consider the asymptotic of \eqref{detED0} as $\la\to \infty$: by \cite[Sec 4.7]{BurgheleaFK92}\footnote{The main argument here is that ${\bf D}_{\A}(\la)^2=\Delta_{\T}+\la$ on $\T$ so we can reduce the analysis to the Laplacian on the the circle.}, we have 
\[ \log \det({\bf D}_{\A,0}(\la))=\sum_{j=0}^1  \la^{1-j/2}(c''_j\log\la+d_j'')+c''_2\log(\la)+o(1)\]
for some $c_j'',d_j''$, in particular the term of order $\la^0$ vanishes too. Since $\log(1+\sqrt{\la})$ does not have $\la^0$ term in its asymptotic expansion at $\la\to \infty$, we have proved that $A=1$. We can now take $\la\to 0$ in \eqref{identity_of_det}. 

First, if $\Sigma$ has at least one Dirichlet boundary, $\Delta_{\Sigma,m}$ is invertible and we get
\[\frac{\det(\Delta_{\Sigma,m})}{\det(\Delta_{\Sigma,\mc{C}'\cup \mc{B}',m})}=
{\det}_{\rm Fr}(\mathbf{D}_{\Sigma,\mc{C}',\mc{B}'}(2{\bf D}_0)^{-1}) \det(\mc{E}{\bf D}_0)={\det}_{\rm Fr}(\mathbf{D}_{\Sigma,\mc{C}',\mc{B}'}(2{\bf D}_0)^{-1})(2\pi)^{b_\ell'+b_h'/2}\]
where we notice that $\det(\mc{E}{\bf D}_0):=e^{-2\zeta_R'(0)(b_\ell'+b_h'/2)}$ and $\zeta_R'(0)=-\frac{1}{2}\log(2\pi)$.

Second, if $\Sigma$ has no Dirichlet boundary, the kernel of $\Delta_{\Sigma,m}$ is made of constants and $\Delta_{\Sigma,\mc{C}'\cup \mc{B}',m}$ has no kernel. We can assume that the cut is a single boundary cut $\mc{B}'$ as the case $\mc{C}'$ is similar (and indeed simpler). We follow the proof of \cite[Theorem B*]{BurgheleaFK92} but with slight variations, by letting $\la\to 0^+$ in \eqref{identity_of_det} (with $A=1$). We have for $\la\to 0^+$
\[  \det(\Delta_{\Sigma,m}+\la)=\la {\det}'(\Delta_{\Sigma,m})(1+o(1)).\]
. 
We shall compute the asymptotic expansion as $\la\to 0^+$ of 
\[d(\la):={\det}_{\rm Fr}(2{\bf D}_0(\la)\mathbf{D}^{-1}_{\Sigma,\mc{B}'}(\la))=\frac{1}{{\det}_{\rm Fr}(\mathbf{D}_{\Sigma,\mc{B}'}(\la)(2{\bf D}_0(\la))^{-1})}.\]
By the same argument as the proof of Lemma \ref{rest_DNmap} again, ${\bf D}_{\Sigma,\mc{B}'}(\la)$ is equal to 
$\mc{E} {\bf D}_{\Sigma^{\#2},\mc{B}''}(\la)\mc{E}$ and thus is self-adjoint on $L^2_{\rm even}(\T)$ and preserves the orthogonal decomposition ${\rm Im}(\Pi_0')+\ker(\Pi_0')$ of $L^2_{\rm even}(\T)$.
By the same argument as Lemma \ref{rest_DNmap}, ${\bf D}_{\Sigma,\mc{B}'}(\la)^{-1}$ has integral kernel $(\Delta_{\Sigma,m}+\la)^{-1}(\cdot,\cdot)|_{{\mc{B}'}^2}$.
One then has as $\la\to 0^+$
\[ {\bf D}_{\Sigma,\mc{B}'}(\la)^{-1}u=\frac{1}{\la  {\rm v}_g(\Sigma)}\int_{\mc{B}'}u \dd\ell_g+H(\la)u=\Big(\frac{\pi}{\la {\rm v}_g(\Sigma)}\Pi_0'+H(\la)\Big)u\]
for some $H(\la)$ holomorphic near $\la=0$. Let ${\bf D}_{\Sigma,\mc{B}',0}(\la)={\bf D}_{\Sigma,\mc{B}'}(\la)+\Pi_0'$, then
 $ {\bf D}_{\Sigma,\mc{B}',0}(\la)^{-1}(1-\Pi'_0)=H(\la)(1-\Pi'_0)$. Since ${\bf D}_0(\la)$ also commutes with $\Pi'_0$, we obtain
\[\begin{split} 
2{\bf D}_0(\la) {\bf D}_{\Sigma,\mc{B}'}(\la)^{-1}=& (\frac{2\pi(1+\sqrt{\la})}{\la  {\rm v}_g(\Sigma)}+2H(\la))\Pi'_0+ 2{\bf D}_0(\la){\bf D}_{\Sigma,\mc{B}',0}(\la)^{-1}(1-\Pi'_0) \end{split}\]
using ${\bf D}_0(\la)\Pi_0'=(1+\sqrt{\la})\Pi_0'$. Then, we conclude that
\[ d(\la)= \frac{\pi(1+\mc{O}(\sqrt{\la}))}{\la  {\rm v}_g(\Sigma)}{\det}_{\rm Fr}(2{\bf D}_0(\la){\bf D}_{\Sigma,\mc{B}',0}(\la)^{-1})= \frac{\pi+\mc{O}(\sqrt{\la})}{\la  {\rm v}_g(\Sigma)}\frac{1}{{\det}_{\rm Fr}({\bf D}_{\Sigma,\mc{B}',0}(\la)(2{\bf D}_0(\la))^{-1})}.\]
We can use that ${\det}_{\rm Fr}({\bf D}_{\Sigma,\mc{B}',0}(\la)(2{\bf D}_0(\la))^{-1})={\det}_{\rm Fr}({\bf D}_{\Sigma,\mc{B}',0}(2{\bf D}_0)^{-1})(1+o(1))$ as $\la\to 0^+$ and $\det(\mc{E}{\bf D}_0(\la))=\det(\mc{E}{\bf D}_0)(1+o(1))=(2\pi)^{1/2}(1+o(1))$ 
to deduce 
\[ {\det}'(\Delta_{\Sigma,m})=\frac{\sqrt{2}}{\sqrt{\pi}}{\rm v}_g(\Sigma)\det(\Delta_{\Sigma,\mc{B}',m}){\det}_{\rm Fr}({\bf D}_{\Sigma,\mc{B}',0}(2{\bf D}_0)^{-1}).\]
This concludes the proof for $\mc{B}'$ cut. The same argument in the $\mc{C}'$ cut case gives
\[ {\det}'(\Delta_{\Sigma,m})={\rm v}_g(\Sigma)\det(\Delta_{\Sigma,\mc{C}',m}){\det}_{\rm Fr}({\bf D}_{\Sigma,\mc{C}',0}(2{\bf D}_0)^{-1}).\qedhere\]
\end{proof}

Next, we compute the determinant of Laplacian for the half-annulus with mixed boundary condition or Dirichlet boundary condition. 
For the mixed condition on the half annulus $\A^+_{\delta}=\A_\delta \cap\H$ (with $\delta\in (0,1)$), we put Dirichlet boundary condition on the two half circles $\T^+$ and $\delta\T^+$ and Neumann boundary condition on the two intervals $[-1,-\delta]\cup [\delta,1]$, and  the metric is $g_\A=|dz|^2/|z|^2$ as before. 
We also recall that the ordinary Riemann zeta function $\zeta_{\text{Rie}}(s):=\sum_{k=1}^\infty \frac{1}{k^s}$ satisfies $\zeta_{\text{Rie}}(0)=-\frac{1}{2}$, $ \zeta_{\text{Rie}}'(0)= -\frac{\log 2 \pi}{2}$.

\begin{proposition}\label{det of half annulus}
	For $\delta \in (0,1)$, the determinant of the Laplacian with mixed boundary condition on $\A_\delta^+$ is given by 
	\begin{equation*}
			\det(\Delta_{ \A_\delta^+,m,g_\A})= -\sqrt{\frac{2}{\pi} }(\log \delta)  \delta^{1/12} \prod_{n \geq 1}  (1- \delta^{2n})
	\end{equation*}
while the determinant of the Dirichlet Laplacian on  $\A_\delta^+$ is given by 
 \begin{equation}\label{detA^+deltaD}
	\text{det}(\Delta_{\mathbb A^+_\delta,D,g_\A})= \frac{1}{\sqrt{2\pi}}\delta^{1/12}\prod_{n \geq 1}  (1- \delta^{2n}).
\end{equation}
\end{proposition}
\begin{proof}
The spectral zeta function for the mixed boundary condition half-annulus $\A_\delta^+$ is 
\[\zeta_{\A^+_\delta,m,g_\A}(s)= \sum_{m \geq 0, n \geq 1}  \frac{1}{  (m^2 +(\frac{n}{2\ell})^2)^s }\]
where $\ell:=-\frac{1}{2\pi}\log\delta$, the eigenvalues $\lambda_{m,n}=m^2 +(\frac{n}{2\ell})^2$ appearing in the above formula correspond to the eigenvectors $ \sin( \sqrt{\lambda_{m,n}- m^2} x)  \cos(m \theta)$ when we map $\A_\delta^+$ to the strip $[0, \log \frac{1}{|\delta|}] \times [0,2 \pi]$. Similarly, the spectral zeta of $\A_\delta$ with Dirichlet boundary condition is given by 
\[\zeta_{\mathbb A_\delta,D,g_\A}(s)= \sum_{m \in \mathbb Z, n \geq 1}  \frac{1}{  (m^2 +(\frac{n}{2\ell})^2)^s }.\]
Then we have
\[\zeta_{\mathbb A_\delta,D,g_\A}(s)= 2 \zeta_{\mathbb A^+_\delta,m,g_\A}(s)- \sum_{ n \geq 1}  \frac{1}{ (\frac{n}{2\ell})^{2s} }= 2 \zeta_{\A^+_\delta,m,g_\A}(s)-(2\ell)^{2s} \zeta_{\text{Rie}}(2s)\]
and therefore
\[\zeta_{\mathbb A_\delta,D,g_\A}'(0)= 2 \zeta_{\mathbb A^+_\delta,m,g_\A}'(0)- 2 \log (2\ell) \zeta_{\text{Rie}}(0)-2 \zeta_{\text{Rie}}'(0)= 2 \zeta_{\A^+_\delta,m,g_\A}'(0)+ \log (4 \pi \ell).\]
Finally, we have
\begin{equation}\label{detA^+delta}
\det(\Delta_{ \mathbb A^+_\delta,m,g_\A})= \sqrt{4 \pi \ell} \text{det}(\Delta_{\mathbb A_\delta,D,g_\A})^{1/2}= -\sqrt{\frac{2}{\pi} }(\log \delta) \delta^{1/12} \prod_{n \geq 1}  (1- \delta^{2n}).
\end{equation}
where we used the formula (see \cite{Wei})
\begin{equation}\label{detAdeltaD}
\text{det}(\Delta_{\mathbb A_\delta,D,g_\A})= 2\ell  \delta^{1/6} \prod_{n \geq 1}  (1- \delta^{2n})^{2}= \frac{- (\log \delta)}{\pi}  \delta^{1/6} \prod_{n \geq 1}  (1- \delta^{2n})^{2}.
\end{equation}
Using \eqref{identity_det}, \eqref{detA^+delta} and \eqref{detAdeltaD}, we deduce that  the determinant of the Dirichlet Laplacian on the half-annulus $\A_\delta^+$ is given by \eqref{detA^+deltaD}.
\end{proof}

\section{GMC estimates}\label{GMC}\label{app:GMC_estimates}
In this appendix, we prove some results announced in Section \ref{sub:regbhamiltonian} (and we will accordingly stick to the notations there). We will prove a result that is slightly more general than needed. For $\alpha\in (0,1)$, we set

\begin{align}\label{defValphak}
V^{(k)}_{\alpha}(\varphi^h):=& \int_{0}^\pi e^{\gamma  \varphi^{h,k}(\theta)-\frac{\gamma^2}{2}\E[  \varphi^{h,k}(\theta)^2]}\Big(\frac{1}{2|\sin(\theta)|+\frac{1}{k}}\Big)^{\alpha}\dd\theta.
\end{align}
In particular, we recover $V^{(k)}_{+}(\varphi^h)$ for $\alpha=\tfrac{\gamma}{2}$.

The sequence $(V^{(k)}_{\alpha}(\varphi^h))_k$ is not a martingale but it is natural to expect that it converges towards a GMC defined as follows. We consider the family of random measures on $[0,\pi]$ indexed by $k$ and defined by
$$
W^{(k)}_{\alpha}(\varphi^h,B):= \int_{B} e^{\gamma  \varphi^{h,k}(\theta)-\frac{\gamma^2}{2}\E[  \varphi^{h,k}(\theta)^2]}\Big(\frac{1}{2|\sin(\theta)| }\Big)^{\alpha}\dd\theta
$$
for any Borel measurable subset $B\subset [0,\pi]$.  This family is a nonnegative martingale for each such $B$, and therefore it converges almost surely. Standard GMC theory (see \cite{Kahane85,rhodes2014_gmcReview}) ensures the uniform integrability of this martingale if $B$ stays away from the singularities located at $\theta=0,\pi$, namely if the closure of $B$ is contained in $(0,\pi)$. More precisely, this martingale is then bounded in $L^p(\Omega_{\T^+})$ for $p<2/\gamma^2$. Actually, GMC theory (see \cite{Kahane85,rhodes2014_gmcReview})  also ensures that, almost surely, the family of random measures $(W^{(k)}_{\alpha}(\varphi^h,\dd x))_k$ converges vaguely towards a limiting measure denoted by 
 \begin{align}\label{defValpha+}
B\mapsto V_{\alpha}(\varphi^h,B):=& \int_{B} e^{\gamma  \varphi^{h}(\theta)-\frac{\gamma^2}{2}\E[  \varphi^{h}(\theta)^2]}\Big(\frac{1}{2|\sin(\theta)|}\Big)^{\alpha}\dd\theta
\end{align}
for measurable subsets $B$ of $(0,\pi)$. Such a convergence does not give information on the behaviour of the limiting measure near the singularities at $\theta=0,\pi$. Yet we can consider the total mass of this measure
$$V_{\alpha}(\varphi^h):=\sup_B\int_{B} e^{\gamma  \varphi^{h}(\theta)-\frac{\gamma^2}{2}\E[  \varphi^{h}(\theta)^2]}\Big(\frac{1}{2|\sin(\theta)|}\Big)^{\alpha}\dd\theta$$
where the supremum runs over the closed subsets $B\subset (0,\pi)$.  It is not clear a priori that this quantity is finite. Actually, the multifractal analysis argument in \cite[Lemma 3.3]{DKRV16} shows directly that the  total mass $V_{\alpha}(\varphi^h)$ is finite for $\alpha<1+\tfrac{\gamma^2}{2}$, even though we will recover this fact without using this argument. Another point that has to be elucidated is whether the random variable $V_{\alpha}(\varphi^h)$ coincides with the almost sure  limit of the martingale $W^{(k)}_{\alpha}(\varphi^h,(0,\pi))$. If so, one may also wonder in which $L^p$ spaces this convergence holds. The following Lemma adresses this point.

 \begin{lemma}
Let  $\alpha\in (0,1)$ and set\footnote{Note that $p_2>1$ iff $\alpha\in (0,1)$.} $p_2:=\tfrac{1}{2}(1-\tfrac{\alpha}{\gamma^2})+ \sqrt{\tfrac{1}{\gamma^2}+\tfrac{1}{4}(1-\tfrac{\alpha}{\gamma^2})^2}$. We have
\begin{align}\label{defValpha}
V_{\alpha}(\varphi^h)=\lim_{k\to\infty}V_\alpha^{(k)}(\varphi^h)=\lim_{k\to\infty}\int_{0}^\pi e^{\gamma  \varphi^{h,k}(\theta)-\frac{\gamma^2}{2}\E[  \varphi^{h,k}(\theta)^2]}\Big(\frac{1}{2|\sin(\theta)|}\Big)^{\alpha}\dd\theta
\end{align}
almost surely and both convergences also hold in $L^p(\Omega_{\T^+})$ for $p \in (0, p_2)$.
\end{lemma}

\begin{proof}
In this proof, it will be convenient for us to   see $V^{(k)}_{\alpha}(\varphi^h)$ as a measure so that we extend the definition to
$$
V^{(k)}_{\alpha}(\varphi^h,B):= \int_{B} e^{\gamma  \varphi^{h,k}(\theta)-\frac{\gamma^2}{2}\E[  \varphi^{h,k}(\theta)^2]}\Big(\frac{1}{2|\sin(\theta)|+\frac{1}{k}}\Big)^{\alpha}\dd\theta
$$
for any Borel measurable subset $B\subset (0,\pi)$. To make the notations less cluttered, we will write $\T^+$ for the open interval $(0,\pi)$ even though this is a slight abuse of notations.

Our first goal is to identify  the limit of the martingale $(W^{(k)}_{\alpha}(\varphi^h,\T^+))_k$, and in the same process that of the sequence $(V^{(k)}_{\alpha}(\varphi^h,\T^+))_k$.  It is straightforward to see that 
\begin{equation}\label{card}
\lim_{k\to\infty} V^{(k)}_{\alpha}(\varphi^h,B) =\lim_{k\to\infty} W^{(k)}_{\alpha}(\varphi^h,B)=V_{\alpha}(\varphi,B)
\end{equation}
for any closed set $B\subset(0,\pi)$. What is not completely obvious is to replace $B$ by $\T^+$ in this relation (in that case,  standard GMC theory does not provide the uniform integrability of the martingale $(W^{(k)}_{\alpha}(\varphi^h,\T^+))_k$). From \eqref{card}, we deduce that a.s. $\liminf _{k\to\infty} V^{(k)}_{\alpha}(\varphi^h,\T^+)\geq V_{\alpha}(\varphi^h,\T^+)$. From the relation 
\[V^{(k)}_{\alpha}(\varphi^h,\T^+)\leq W^{(k)}_{\alpha}(\varphi^h,\T^+),\] 
we also deduce  $\limsup _{k\to\infty} V^{(k)}_{\alpha}(\varphi^h,\T^+)\leq \lim _{k\to\infty} W^{(k)}_{\alpha}(\varphi^h,\T^+)$.  Now we claim that we have the almost sure convergence
\begin{equation}\label{limartin}
\lim _{k\to\infty} W^{(k)}_{\alpha}(\varphi^h,\T^+)=V_{\alpha}(\varphi^h,\T^+).
\end{equation}
Let us denote by $\mc{F}_k$ the sigma algebra generated by the field $\varphi^{h,k}$, so that $(\mc{F}_k)_k$ is a filtration. Consider an increasing family $(f_n)_n$ of continuous functions on $\T^+$, with compact support in $(0,\pi)$, with $0\leq f_n\leq 1$, and $f_n\nearrow 1$ as $n\to\infty$. For $k,n\in\N$, we have
$$\E[V_{\alpha}(\varphi^h,\T^+)|\mc{F}_k]\geq \E[V_{\alpha}(\varphi^h,f_n)|\mc{F}_k]=W^{(k)}_{\alpha}(\varphi^h,f_n),$$
the last equality following from the the convergence in  $L^p(\Omega_{\T^+})$ for $p<2/\gamma^2$ of the martingale $(W^{(k)}_{\alpha}(\varphi^h,f_n))_k$ towards $V_{\alpha}(\varphi^h,f_n)$ (note that the martingale $(W^{(k)}_{\alpha}(\varphi^h,f_n))_k$  is bounded by $(W^{(k)}_{\alpha}(\varphi^h,B))_k$ for some closed $B$ contained in $(0,\pi)$, hence bounded in   $L^p(\Omega_{\T^+})$ for $p<2/\gamma^2$). Since $n$ is arbitrary, we can send it to $\infty$ and deduce that 
\[\E[V_{\alpha}(\varphi^h,\T^+)|\mc{F}_k]\geq W^{(k)}_{\alpha}(\varphi^h,\T^+).\] Also, noticing that $\mathbf{1}_{\T^+}=\lim_nf_n$ pointwise, by Fatou's Lemma we have 
$$\E[V_{\alpha}(\varphi^h,\T^+)|\mc{F}_k]\leq \liminf_n \E[V_{\alpha}(\varphi^h,f_n)|\mc{F}_k]= \liminf_n W^{(k)}_{\alpha}(\varphi^h,f_n)=W^{(k)}_{\alpha}(\varphi^h,\T^+),$$
hence our claim. From all these considerations, we deduce that, almost surely,
$$\lim_{k\to\infty}V^{(k)}_{\alpha}(\varphi^h,\T^+)=\lim_{k\to\infty}W^{(k)}_{\alpha}(\varphi^h,\T^+)=V_{\alpha}(\varphi^h,\T^+).$$

 Now we show that 
\begin{equation}\label{Lpmom}
\E\Big[\Big(V_{\alpha}(\varphi^h,\T^+)\Big)^p\Big]<+\infty
\end{equation}
for $p\in (1,p_2)$. First recall that $\E\Big[\Big(V_{\alpha}(\varphi^h,B)\Big)^p\Big]<+\infty$ for any closed set $B\subset (0,\pi)$ and $p<2/\gamma^2$ by standard GMC theory, as explained before.  The region that will affect these $L^p$-estimates is thus located near the extremal points $\{0,\pi\}$. Since they play a symmetric role, we only focus on $0$. We want to show $\E\Big[\Big(V_{\alpha}(\varphi^h,B)\Big)^p\Big]<+\infty$ for some small interval $B\subset(0,\pi)$ with closure containing $0$, and for $p\in (1,p_2)$. For $\theta,\theta'\in (0,\delta)$ with $\delta$ small enough we have
\[\E[\varphi^h(\theta)\varphi^h(\theta')]=\log\frac{1}{|e^{i\theta}-e^{i\theta'}||e^{i\theta}-e^{-i\theta'}|}\leq \log\frac{1}{|\theta-\theta'||\theta+\theta'|}+C\] 
for some $C>0$. This inequality serves to use Kahane's convexity inequality (see \cite{Kahane85} or \cite{rhodes2014_gmcReview}): it says that the bound \eqref{Lpmom} holds provided we can show the same estimate for a Gaussian field with covariance $\log\frac{1}{|\theta-\theta'||\theta+\theta'|}$. Let us first identify such a Gaussian field. We consider a centered Gaussian field $\tilde Y$ with covariance $\log\frac{1}{|z-z'|}$ on the unit disk $\D$. Then we set $Y(z)=\frac{1}{\sqrt{2}}(\tilde{Y}(z)+\tilde{Y}(-z))$, which has covariance $\log\frac{1}{|z-z'||z+z'|}$. Therefore, this field has the desired covariance: it suffices to consider its restriction to the real line. Thus we have to show
\begin{equation}\label{Lpmombis}
\E\Big[\Big(Z((0,\delta))\Big)^p\Big]<+\infty
\end{equation}
where the GMC $Z$ is  defined by
$$Z(B):=\lim_{\epsilon\to 0}Z_\epsilon(B),\quad \text{with }\quad Z_\epsilon(B):=\int_Be^{\gamma Y_\epsilon(x)-\frac{\gamma^2}{2}\E[Y_\epsilon(x)^2]}\frac{\dd x}{|x|^\alpha}$$
for any  open subset  $B$  of $(0,\delta)$, and $Y_\epsilon$ is the circle average regularisation of $Y$ in the metric $|dz|^2$. The choice of this regularisation is especially convenient as it obeys a nice scaling relation: indeed it is straightforward to check that, for $\lambda \in (0,1)$, 
$$Y_{\lambda\epsilon}(\lambda\cdot)\stackrel{\rm law}{=}Y_{\epsilon}(\cdot)+N_\lambda$$
where $N_\lambda$ is a centered Gaussian random variable independent of the process $Y_{\epsilon}$ and with variance $-2\log\lambda$.  We are now in position to focus on the $L^p$-moments. This implies, using a change of variables, that
\begin{align*}
\E\big[\big(  Z_{\lambda\epsilon}((0,\lambda\delta))   \big)^p\big]=&\E\big[\big( \lambda^{1-\alpha}e^{\gamma N_\lambda-\frac{\gamma^2}{2}\E[N_{\lambda}^2]} Z_{ \epsilon}((0, \delta))   \big)^p\big]\label{play}
\\
=&\lambda^{p(1-\alpha+\gamma^2)-p^2\gamma^2}\E\big[\big(   Z_{ \epsilon}((0, \delta))   \big)^p\big].
\end{align*}
Using the inequalities $\lambda^p(a+b)^p\leq (\lambda a+(1-\lambda)b)^p\leq \lambda a^p+(1-\lambda)b^p$ for $\lambda\in(0,1/2)$ and $a,b>0$, we deduce  
\begin{align}
\E\big[\big(  Z_{\lambda\epsilon}((0,\delta))  \big)^p\big]=& \E\big[\big(  Z_{\lambda\epsilon}((0,\lambda\delta))+Z_{\lambda\epsilon}((\lambda\delta,\delta))  \big)^p\big]
\\
\leq&
 \lambda^{1-p}  \E\big[\big(  Z_{\lambda\epsilon}((0,\lambda\delta))   \big)^p\big]+\lambda^{-p}(1-\lambda) \E\big[\big(   Z_{\lambda\epsilon}((\lambda\delta,\delta))  \big)^p\big] \nonumber
\\
=&
\lambda^{1-p}  \lambda^{p(1-\alpha+\gamma^2)-p^2\gamma^2}\E\big[\big(   Z_{ \epsilon}((0, \delta))   \big)^p\big]  +\lambda^{-p}(1-\lambda)  \E\big[\big(   Z_{\lambda\epsilon}((\lambda\delta,\delta))  \big)^p\big]\nonumber
\\ 
=&
\lambda^{ f(p)}\E\big[\big(   Z_{ \epsilon}((0, \delta))   \big)^p\big]  + \lambda^{-p}(1-\lambda) \E\big[\big(   Z_{\lambda\epsilon}((\lambda\delta,\delta))  \big)^p\big]\nonumber
\end{align}
where we have set $f(p):=1+p(\gamma^2-\alpha)-p^2\gamma^2$. This polynomial of degree two has two roots $p_1,p_2=\tfrac{1}{2}(1-\tfrac{\alpha}{\gamma^2})\pm \sqrt{\tfrac{1}{\gamma^2}+\tfrac{1}{4}(1-\tfrac{\alpha}{\gamma^2})^2}$ and we have $p_2>1$ for $\alpha\in (0,1)$. Now we claim that, for $p>1$, there exists a constant $C_p>0$ such that for every $\lambda\in(0,1)$ and every $\epsilon>0$
\begin{equation}\label{inegkahp}
\E\big[\big(   Z_{ \epsilon}((0, \delta))   \big)^p\big]\leq C_p \E\big[\big(   Z_{\lambda \epsilon}((0, \delta))   \big)^p\big].
\end{equation}
This follows easily from Kahane's inequality and the fact that, for some $C>0$
$$\E[Y_\epsilon(x)Y_\epsilon(x')]\leq C+\E[Y_{\lambda\epsilon}(x)Y_{\lambda\epsilon}(x')] $$
for arbitrary $\epsilon>0$ and $\lambda\in (0,1)$. Indeed, this follows from the relations $\E[Y_\epsilon(x)Y_\epsilon(x')]=\log\frac{1}{|x-x'|}$ if $|x-x'|>\epsilon$, and $\E[Y_\epsilon(x)Y_\epsilon(x')]=-\log\epsilon+f((x-x')/\epsilon)$ for $|x-x'|\leq\epsilon$ with the function $f$ defined for $|z|\leq 1$ by
$$f(z):=\frac{1}{(2\pi)^2}\iint_{[0,2\pi]^2}\log\frac{1}{|z+e^{i\theta}-e^{i\theta'}|}\dd \theta\dd\theta'.$$
The function $f$ is bounded for $|z|\leq 1$, and our claim follows.

We deduce from \eqref{inegkahp} and \eqref{play} that, for $p\in (1,p_2)$, 
\[\begin{split}
\E\big[\big(  Z_{\epsilon}((0,\delta))  \big)^p\big]\leq & C_p \E\big[\big(   Z_{\lambda \epsilon}((0, \delta))   \big)^p\big]\\
\leq & C_p\Big(C_p \lambda^{ f(p)}\E\big[\big(  Z_{ \lambda \epsilon}((0, \delta))   \big)^p\big] + \lambda^{-p}(1-\lambda) \E\big[\big(   Z_{\lambda\epsilon}((\lambda\delta,\delta))  \big)^p\big]\Big).
\end{split}\]
Since $f(p)>0$ we can choose $\lambda>0$ small enough so as to make $C_p \lambda^{ f(p)}<1$ and we end up with
\[\E\big[\big(  Z_{\epsilon}((0,\delta))  \big)^p\big]\leq C_p\frac{\lambda^{-p}(1-\lambda)}{1- C_p \lambda^{ f(p)}  } \E\big[\big(   Z_{\lambda\epsilon}((\lambda\delta,\delta))  \big)^p\big]\leq C\frac{\lambda^{-p}(1-\lambda)}{1- C_p \lambda^{ f(p)}  } \E\big[\big(   Z((\lambda\delta,\delta))  \big)^p\big].\]
We have used again Kahane's inequality for the last inequality (which produces eventually the further constant $C>0$). Note that $ \E\big[\big(   Z((\lambda\delta,\delta))  \big)^p\big]<+\infty$ since $p<2/\gamma^2$. This shows \eqref{Lpmombis}, and then \eqref{Lpmom}.  As a consequence this shows that the convergence \eqref{limartin} also holds in $L^p$ for $p<p_2$ since $(W^{(k)}_{\alpha}(\varphi^h,\T^+))_k$ is a  martingale. 

Finally the relation $0\leq V^{(k)}_{\alpha}(\varphi^h,\T^+)\leq W^{(k)}_{\alpha}(\varphi^h,\T^+)$,  the convergences of $V^{(k)}_{\alpha}(\varphi^h,\T^+)- W^{(k)}_{\alpha}(\varphi^h,\T^+)\to 0$ almost surely  and the convergence of the family $(W^{(k)}_{\alpha}(\varphi^h,\T^+))_k$ in $L^p(\Omega_{\T^+})$ for $1<p<p_2$ imply the convergence of the family $(V^{(k)}_{\alpha}(\varphi^h,\T^+))_k$ in $L^p(\Omega_{\T^+})$ for $1<p<p_2$. Indeed, consider $1<p<p'<p_2$, and $R>0$. We have
\begin{align*}
\E\big[\big|  V^{(k)}_{\alpha}(\varphi^h,\T^+)&-W^{(k)}_{\alpha}(\varphi^h,\T^+)\big|^p\big] 
\\
\leq & \E\big[\big| V^{(k)}_{\alpha}(\varphi^h,\T^+)-W^{(k)}_{\alpha}(\varphi^h,\T^+)\big|^p\mathbf{1}_{|V^{(k)}_{\alpha}(\varphi^h,\T^+)-W^{(k)}_{\alpha}(\varphi^h,\T^+) |\leq R}\big] 
\\
&+2^p\E\big[\big| W^{(k)}_{\alpha}(\varphi^h,\T^+)\big|^p\mathbf{1}_{| W^{(k)}_{\alpha}(\varphi^h,\T^+) |>R/2}\big] 
\\
\leq & \E\big[\big| V^{(k)}_{\alpha}(\varphi^h,\T^+)-W^{(k)}_{\alpha}(\varphi^h,\T^+)\big|^p\mathbf{1}_{|V^{(k)}_{\alpha}(\varphi^h,\T^+)-W^{(k)}_{\alpha}(\varphi^h,\T^+) |\leq R}\big] 
\\
& +2^{p'}R^{-(p'-p)}\sup_k\E\big[\big| W^{(k)}_{\alpha}(\varphi^h,\T^+)\big|^{p'} \big] .
\end{align*}
We can choose $R$ large enough so as to make the second term in the right hand side arbitrarily small, independently of $k$, and then send $k$ to $\infty$ to get the first term arbitrarily small.
\end{proof}

\bibliographystyle{alpha}
\bibliography{references}

\end{document}